\documentclass[11pt]{article}
\usepackage[margin=1in]{geometry}
\usepackage{amsmath}
\usepackage{amssymb}
\usepackage{amsthm}
\usepackage{forloop}
\usepackage{nicefrac}
\usepackage{mathtools}
\usepackage[table]{xcolor}
\usepackage{url}
\usepackage{bm}
\usepackage{float}
\usepackage{hyperref}
\hypersetup{
    colorlinks,
    linkcolor={blue!100!black!70},
    citecolor={blue!75!black},
    urlcolor={blue!75!black}
}
\usepackage[nameinlink,capitalise]{cleveref}
\usepackage[shortlabels]{enumitem}
\usepackage{cancel}
\usepackage{multirow}
\usepackage{makecell}
\usepackage{caption}
\usepackage{subcaption}
\usepackage{framed}
\usepackage{wrapfig}
\usepackage{booktabs}
\usepackage{mathrsfs}
\usepackage[linesnumbered, ruled]{algorithm2e}

\newtheorem{theorem}{Theorem}[section]
\newtheorem{lemma}{Lemma}[section]

\newtheorem{definition}{Definition}[section]
\newtheorem{proposition}{Proposition}[section]

\crefname{definition}{Def.}{Defs.}
\crefname{equation}{Eq.}{Eqs.}

\newcommand{\defcal} [1]{\expandafter\newcommand\csname cal#1\endcsname{{\mathcal #1}}}
\newcommand{\defsf} [1]{\expandafter\newcommand\csname sf#1\endcsname{{\mathsf #1}}}
\newcommand{\defbf} [1]{\expandafter\newcommand\csname bf#1\endcsname{{\mathbf #1}}}
\newcommand{\defbb} [1]{\expandafter\newcommand\csname bb#1\endcsname{{\mathbb{#1}}}}
\newcommand{\deffrak} [1]{\expandafter\newcommand\csname frak#1\endcsname{{\mathfrak{#1}}}}

\newcounter{ct}
\forLoop{1}{26}{ct}{
    \edef\letter{\Alph{ct}}
    \expandafter\defcal\letter
    \expandafter\defsf\letter
    \expandafter\defbf\letter
    \expandafter\defbb\letter
    \expandafter\deffrak\letter
}
\forLoop{1}{26}{ct}{
    \edef\letter{\alph{ct}}
    \expandafter\defcal\letter
    \expandafter\defsf\letter
    \expandafter\defbf\letter
    \expandafter\defbb\letter
    \expandafter\deffrak\letter
}

\newcommand{\rmD}{\mathrm{D}}
\newcommand{\rmd}{\mathrm{d}}
\newcommand{\rmI}{\mathrm{I}}

\newcommand{\argmin}{\mathop{\mathrm{argmin}}}
\newcommand{\argmax}{\mathop{\mathrm{argmax}}}

\newcommand{\frob}{\mathrm{F}}
\newcommand{\op}{\mathrm{op}}
\newcommand{\numberthis}{\addtocounter{equation}{1}\tag{\theequation}}

\newcommand{\TVdist}{\mathrm{d}_{\mathrm{TV}}}
\newcommand{\KLdist}{\mathrm{d}_{\mathrm{KL}}}
\newcommand{\mirrorfunc}[1]{
    \ensuremath{\if$#1$\phi \else \phi(#1)\fi}
}
\newcommand{\mirrorfuncdual}[1]{
    \ensuremath{\if$#1$\phi^{\star}\else \phi^{\star}(#1)\fi}
}
\newcommand{\primtodual}[1]{\nabla\mirrorfunc{#1}}
\newcommand{\dualtoprim}[1]{\nabla\mirrorfuncdual{#1}}
\newcommand{\hessmirror}[1]{\nabla^{2}\mirrorfunc{#1}}
\newcommand{\hessmirrorinv}[1]{\nabla^{2}\mirrorfuncdual{#1}}
\newcommand{\trace}{\mathrm{trace}}
\newcommand{\defeq}{\overset{\mathrm{def}}{=}}

\newcommand{\subscript}[2]{$#1 _ #2$}
\newlist{assumplist}{enumerate}{1}
\setlist[assumplist]{label=\subscript{\textbf{\textsf{A}}}{\textsf{{\arabic*}}},leftmargin=*, itemsep=0pt}

\newcommand{\potential}[1]{
    \ensuremath{\if$#1$f\else f(#1)\fi}
}
\newcommand{\metric}[1]{
    \ensuremath{\if$#1$\mathscr{G}\else \mathscr{G}(#1)\fi}
}
\newcommand{\gradpotential}[1]{\nabla\potential{#1}}
\newcommand{\hesspotential}[1]{\nabla^{2}\potential{#1}}
\newcommand{\mixingtime}[1]{\tau_{\mathrm{mix}}(#1)}
\newcommand{\interior}[1]{\mathrm{int}(#1)}
\newcommand{\crossratio}[3]{\mathrm{CR}(#1, #2; #3)}
\newcommand{\gaussian}[2]{\mathcal{N}\left(#1,~#2\right)}

\newcommand{\MAMLA}{\textsf{MAMLA}}
\newcommand{\MAPLA}{\textsf{MAPLA}}
\newcommand{\Dikin}{\textsf{DikinWalk}}
\newcommand{\primalspace}{\mathcal{K}}
\newcommand{\targetdens}{\pi}
\newcommand{\targetdist}{\Pi}

\makeatletter
\newcommand*{\algotitle}[2]{%
  \stepcounter{algocf}%
  \hypertarget{algocf.title.\theHalgocf}{}%
  \NR@gettitle{#1}%
  \label{#2}%
  \addtocounter{algocf}{-1}%
}
\makeatother

\usepackage[round]{natbib}
\usepackage[splitrule]{footmisc}

\allowdisplaybreaks

\crefname{paragraph}{part}{parts}
\makeatletter
\def\@maketitle{%
  \newpage
  \begin{center}%
  \let \footnote \thanks
    {\Large \bf \@title \par}%
  \end{center}%
  \par
  \vskip 0.5em}
\makeatother

\title{High-accuracy sampling from constrained spaces with the \\
Metropolis-adjusted Preconditioned Langevin Algorithm}

\begin{document}
\maketitle

\begin{center}
{\large
\begin{tabular}{ccc}
    \makecell{Vishwak Srinivasan\(^\dagger\)\\{\normalsize\texttt{vishwaks@mit.edu}}} & \makecell{Andre Wibisono\(^\star\)\\{\normalsize\texttt{andre.wibisono@yale.edu}}} & \makecell{Ashia Wilson\(^\dagger\)\\{\normalsize\texttt{ashia07@mit.edu}}}
\end{tabular}
\vskip 0.5em

\normalsize
\begin{tabular}{c}
\({}^{\dagger}\)Department of Electrical Engineering and Computer Science, MIT
\\
[0.25em]
\({}^{\star}\)Department of Computer Science, Yale University \\
\end{tabular}
}
\end{center}

\begin{abstract}
  \noindent In this work, we propose a first-order sampling method called the Metropolis-adjusted Preconditioned Langevin Algorithm for approximate sampling from a target distribution whose support is a proper convex subset of \(\bbR^{d}\).
  Our proposed method is the result of applying a Metropolis-Hastings filter to the Markov chain formed by a single step of the preconditioned Langevin algorithm with a metric \(\metric{}\), and is motivated by the natural gradient descent algorithm for optimisation.
  We derive non-asymptotic upper bounds for the mixing time of this method for sampling from target distributions whose potentials are bounded relative to \(\metric{}\), and for exponential distributions restricted to the support.
  Our analysis suggests that if \(\metric{}\) satisfies stronger notions of self-concordance introduced in \citep{kook2024gaussian}, then these mixing time upper bounds have a strictly better dependence on the dimension than when is merely self-concordant.
  We also provide numerical experiments that demonstrates the practicality of our proposed method.
  Our method is a \emph{high-accuracy} sampler due to the polylogarithmic dependence on the error tolerance in our mixing time upper bounds.
\end{abstract}

\setcounter{page}{1}

\section{Introduction}
\label{sec:intro}
Several statistical estimation and inference tasks involve drawing samples from a distribution; examples include estimating functionals, generating credible intervals for point estimates, and structured exploration of state spaces.
The complexity of the distribution is influenced by the problem being studied, or by modelling choices made which can make it infeasible to sample from such distributions exactly.
Markov chain Monte Carlo (MCMC) algorithms have helped tackle this challenge over the past several decades \citep{brooks2011handbook}, and have seen renewed interest in machine learning and statistics, especially in high-dimensional settings.
These distributions of interest could be supported on the entire space (say \(\bbR^{d}\)), or on a proper subset of \(\bbR^{d}\).
Our focus is on the latter kind of distributions, which we refer to as \emph{constrained}, and these arise in several practical problems.
A non-exhaustive collection of applications are Bayesian modelling \citep{gelfand1992bayesian,pakman2014exact}, regularised regression \citep{tian2008efficient,celeux2012regularization}, modelling metabolic networks \citep{heirendt2019creation}, and differential privacy \citep{bassily2014private}.

Formally, we are interested in generating (approximate) samples from a distribution \(\targetdist\) with support \(\primalspace \subset \bbR^{d}\) that is convex, and density \(\targetdens\) (with respect to the Lebesgue measure) of the form
\begin{equation*}
    \targetdens(x) \propto \exp(-\potential{x})~.
\end{equation*}
The function \(\potential{} : \interior{\primalspace} \to \bbR\) is termed the \emph{potential} of \(\targetdist\), and we refer to the above sampling problem as the \textbf{constrained sampling problem} henceforth.
MCMC methods have been proposed and studied for the constrained sampling problem since at least as early as the 1980s \citep{smith1984efficient}.
The earliest methods were originally catered towards generating points uniformly distributed over \(\primalspace\) (where \(\potential{} \equiv 0\)) with applications in volume computation.
The aforementioned collection of applications also concern non-uniform distributions over \(\primalspace\).
An example of such a non-uniform distribution is the truncated Gaussian distribution which arises in various statistical applications.
Generating a sample from a (unconstrained) Gaussian distribution with a specified mean and covariance can be performed exactly and also efficiently using the Box-Muller transform.
However, sampling from a Gaussian distribution whose support is restricted to a convex subset poses non-trivial difficulties beyond \(d \geq 1\), and hence approximate sampling schemes have to be considered.

In general, the task of sampling from a distribution \(\targetdist\) with potential \(\potential{}\) supported on \(\primalspace\) can be transformed into a constrained sampling problem from a distribution \(\widetilde{\targetdist}\) whose potential is linear and supported over a larger domain through \emph{lifting}.
More precisely, suppose the potential \(\potential{}\) satisfies \(\potential{} = \sum_{i=1}^{N}\potential{}_{i}\) for a collection of functions \(\{\potential{}_{i}\}_{i=1}^{N}\).
Lifting defines a new domain \(\widetilde{\primalspace}\) and potential \(\widetilde{\potential{}}\) (referred to as the \emph{lifted domain} and \emph{lifted potential} respectively) which are stated below.
\begin{equation*}
    \widetilde{\primalspace} = \left\{y = (x, t) : x \in \primalspace,~t \in \bbR^{N}, \potential{}_{i}(x) \leq t_{i} ~\forall ~i \in [N]\right\}, \quad
    \widetilde{\potential{}}(y) = \widetilde{\potential{}}(x,t) = [\underbrace{0, \ldots, 0}_{d \text{ times}}, \underbrace{1, \ldots, 1}_{N \text{ times}}]^{\top}(y)~.
\end{equation*}
The lifted domain \(\widetilde{\primalspace}\) is convex if every \(f_{i}\) is convex.
This is because the lifted domain is formed by the intersection of \(N\) epigraphs defined by each \(f_{i}\).
Lifting described here is the sampling analogue of converting an optimisation problem of the form \(\min\limits_{x \in \primalspace} f(x)\) into \(\min\limits_{(x, t) \in \widetilde{\primalspace}} \bm{1}^{\top}t\).
The principle behind \emph{lifting} for sampling is the fact that the marginal distribution of the first \(d\) elements of \(y \sim \widetilde{\targetdist}\) coincides with the distribution of \(x \sim \targetdist\).
This is made evident by the calculation below.
\begin{equation*}
    \int_{t \in \widetilde{\primalspace}} e^{-\widetilde{\potential{}}((x, t))} \rmd t = \prod_{i=1}^{N}\int_{t_{i} \geq \potential{}_{i}(x)} e^{-t_{i}}\rmd t_{i} = \prod_{i=1}^{N} e^{-\potential{}_{i}(x)} = e^{-\potential{x}}~.
\end{equation*}
Thus, lifting is a versatile technique which enables working with general potentials \(\potential{}\).
However, this is not necessary to deal with non-uniform distributions supported over a constrained domain, particularly in settings where the potential \(\potential{}\) is ``compatible'' with geometric properties of the domain \(\primalspace\).
We expand on such notions of compatibility in the sequel.

In this work, we propose a method for the constrained sampling problem, and our method is motivated by the rich connection between optimisation methods and sampling algorithms.
The analogue of the constrained sampling problem in optimisation is the task of minimising \(\potential{}\) over a constrained feasibility set (\(\primalspace\)).
For this constrained optimisation problem, it is usually beneficial to impose a non-Euclidean geometric structure over this feasibility set through a \emph{metric} \(\metric{}\), and this defines a Riemannian manifold \((\primalspace, \metric{})\).
For example, when \(\primalspace\) is a polytope formed by \(m\) linear constraints, a natural choice of \(\metric{}\) is the Hessian of the log-barrier function of this polytope \citep{alvarez2004hessian}.
A consequence of imposing this metric is that standard first-order methods for minimising \(\potential{}\) like gradient descent have to be modified as the Euclidean gradient \(\gradpotential{}\) is no longer the direction of steepest descent.
\citet{amari1998natural} adapts the gradient descent method to incorporate the metric \(\metric{}\), which they refer to as the \emph{natural gradient descent} method.
In this work, we propose a new MCMC method inspired by the natural gradient descent method called the \emph{Metropolis-adjusted Preconditioned Langevin Algorithm} (\MAPLA{}) for the constrained sampling problem.
Each iteration of this algorithm is composed of the steps below.
\begin{enumerate}[leftmargin=*, itemsep=0pt, label=(\arabic*)]
\item Generate a proposal \(Z\) from \(X\) (current iterate) with one step of \ref{eq:PLA} with metric \(\metric{}\).
\item If \(Z \not\in \primalspace\), reject \(Z\) and set \(X\) to be \(X'\) (next iterate).
\item If \(Z \in \primalspace\), compute the Metropolis-Hastings acceptance probability \(p_{X \to Z}\).
\item With probability \(p_{X \to Z}\), set \(X' = Z\) (accept); otherwise set \(X' = X\) (reject).
\end{enumerate}

Steps (3) and (4) are commonly referred to as the Metropolis-Hastings filter, and the above steps form the Metropolis adjustment of the Preconditioned Langevin Algorithm (\ref{eq:PLA}) with respect to \(\targetdist\).
The Markov chain defined by \MAPLA{} has two essential properties: it is both reversible with respect to \(\targetdist\) and ergodic.
The reversibility is due to the fact that the Metropolis adjustment of any Markov chain with respect to an arbitrary distribution \(\nu\) creates a new Markov chain that is reversible with respect to \(\nu\).
The ergodicity of \MAPLA{} follows from the ergodicity of \textsf{PLA}.
Therefore, \MAPLA{} is an \emph{unbiased} MCMC algorithm; this means that the law of an iterate generated by \MAPLA{} as the number of iterations tend to infinity is the target distribution \(\targetdist\).
While \MAPLA{} draws inspiration from natural gradient methods in optimisation, it is not the first MCMC method based on incorporating geometric properties of \(\primalspace\) encoded in the metric \(\metric{}\).
\Dikin{} and \textsf{ManifoldMALA} are two other methods that directly use the metric \(\metric{}\) in its design, and both of these methods are formed by including a Metropolis-Hastings filter to accept/reject proposals as well. 
In \cref{tab:proposal-table}, we show the forms of the proposal distributions of these algorithms and that of \textsf{MAPLA}, which makes the relation of \textsf{MAPLA} to each of them more apparent.
\begin{table}[H]
\centering
\caption{Proposal distributions \(\calP_{X}\) of certain algorithms.
All of these proposal distributions are Gaussians with covariance \(2h \cdot \metric{X}^{-1}\) and only differ in the mean \(\frakm(X)\); \(h > 0\) is the step size.
}
\label{tab:proposal-table}
\renewcommand{\arraystretch}{1.3}
\begin{tabular}{cc}
Algorithm & Mean \(\frakm(X)\) \\
\hline
\(\underset{\text{This work}}{\MAPLA{}}\) & \(X - h \cdot \metric{X}^{-1}\gradpotential{X}\) \\
\(\underset{\text{\citet{girolami2011riemann}}}{\textsf{ManifoldMALA}}\) & \(X + h \cdot \{(\nabla \cdot \metric{}^{-1})(X) - \metric{X}^{-1}\gradpotential{X}\}\) \\
\(\underset{\text{\citet{kook2024gaussian}}}{\Dikin{}}\) & \(X\)
\end{tabular}
\end{table}

From this comparison, we can view \MAPLA{} as an interpolation between \textsf{ManifoldMALA} (based on the weighted Langevin algorithm introduced later) and \Dikin{} (based on a geometric random walk with metric \(\metric{}\)).
When \(\metric{}\) is identity, \Dikin{} is equivalent to the \textsf{MetropolisRandomWalk} \citep{mengersen1996rates}, and both \textsf{ManifoldMALA} and \MAPLA{} reduce to the Metropolis-adjusted Langevin Algorithm (\textsf{MALA}) \citep{roberts1996exponential}.
The key difference between \MAPLA{} and \textsf{ManifoldMALA} is the absence / presence of the \((\nabla \cdot \metric{}^{-1})\) term\footnote{\(\nabla \cdot M\) for a matrix function \(M\) is the vector-valued function defined as \((\nabla \cdot M)(x)_{i} = (\nabla \cdot M_{:,i})(x)\).} in the mean of the proposal distributions, respectively.
The omission of this term in \MAPLA{} is computationally advantageous, as this can be difficult to compute for choices of \(\metric{}\) suited for the constrained sampling problem, like for instance the Hessian of the log-barrier function of a polytope when \(\primalspace\) is that polytope.
On the other hand, while the proposal distribution of \Dikin{} is the simplest amongst the algorithms, it lacks any information about the target distribution \(\targetdist\), and hence \MAPLA{} can be interpreted as performing a natural drift correction to \Dikin{}.

Another algorithm that is related to \MAPLA{} is the Metropolis-adjusted Mirror Langevin algorithm (\MAMLA{}) \citep{srinivasan2024fast} that was recently proposed and analysed for the constrained sampling problem.
Mirror Langevin methods are based on another approach to solve constrained optimisation problems called mirror methods \citep{nemirovskii1983problem}, and their key feature is the use of a Legendre type function \citep[Chap. 26]{rockafellar1997convex} \(\mirrorfunc{}\) referred to as the \emph{mirror function}.
The primary role of the mirror function is to define a mirror map \(\primtodual{}\) which maps points in \(\primalspace\) (called the \emph{primal} space) to \(\primalspace^{\star} = \mathrm{range}(\primtodual{})\) (called the \emph{dual space}) where updates are performed, and are mapped back to \(\primalspace\) using the inverse of \(\primtodual{}\).
Since \(\mirrorfunc{}\) is of Legendre type, the inverse of \(\primtodual{}\) is \(\dualtoprim{}\), where \(\mirrorfuncdual{}\) is the convex conjugate of \(\mirrorfunc{}\).
The mirror function also induces non-Euclidean geometric structures on \(\primalspace\) and \(\primalspace^{\star}\) through metrics \(\hessmirror{}\) and \(\hessmirrorinv{}\) respectively.
The proposal distribution \(\calP_{X}\) of \MAMLA{} is the law of the iterate obtained with one step of the Mirror Langevin algorithm \citep{zhang2020wasserstein} from \(X\) and is defined as
\begin{equation*}
    \calP_{X} = \mathrm{Law}(\dualtoprim{Z}); \quad Z \sim \widetilde{\calP}_{X} = \gaussian{\primtodual{X} - h \cdot \gradpotential{X}}{2h \cdot \hessmirror{X}}~.
\end{equation*}
Let \(Y = \primtodual{X}\) and \(\potential{}^{\star} = f \circ \dualtoprim{}\).
Using the property that \(\primtodual{} \circ \dualtoprim{} = \mathrm{Id}\), it can be shown that the distribution \(\widetilde{\calP}_{X}\) above is equivalent to \(\gaussian{Y - h \cdot \hessmirrorinv{Y}^{-1}\nabla f^{\star}(Y)}{2h \cdot \hessmirrorinv{Y}^{-1}}\).
This lends to the interpretation of \MAMLA{} as the \emph{dual} version of \MAPLA{} with the metric \(\hessmirrorinv{}\) in the dual space, which also complements the connection between the mirror descent and natural gradient descent methods in optimisation \citep{raskutti2015information}.
However, a central challenge with mirror methods is computing \(\dualtoprim{}\) which in generality requires solving the convex program that defines \(\mirrorfuncdual{}\), although in a few cases, a closed form expression exists.

As noted previously, \MAPLA{} is unbiased.
The quality of a MCMC method is related to its \emph{mixing time} (a non-asymptotic notion) that is roughly defined as the minimum number of iterations of the method to generate an iterate whose distribution is ``close'' to the target distribution (see \cref{sec:prelims:mixing} for a precise definition).
From a practical standpoint, it is useful to understand the effect of problem parameters such as the dimension \(d\) and the properties of \(\targetdist\) on the mixing time of \MAPLA{}.
In this work, we also obtain non-asymptotic upper bounds for the mixing time of \MAPLA{} under certain sufficient conditions on the potential \(\potential{}\) and the metric \(\metric{}\) for the constrained sampling problem.
The central problem parameters as characterised by the sufficient conditions and their dependence on the \(\delta\)-mixing time of \MAPLA{} are summarised in \Cref{tab:mixing-time-scalings}.

\begin{table}[H]
\centering
\setlength{\tabcolsep}{2pt}
\begin{tabular}{rlc}
\multicolumn{2}{c}{Conditions} & Mixing time scaling \\
\hline \\[-0.5em]
\multirow{2}{*}{\(\potential{}\) is \(\begin{cases}
    (\mu,\metric{})\text{-curv. lower-bdd.} \\
    (\lambda,\metric{})\text{-curv. upper-bdd.} \\
    (\beta,\metric{})\text{-grad. upper-bdd.}
\end{cases}\) \& } & \(\metric{}\) is \(\begin{cases}
    \text{self-concordant} \\
    \nu\text{-symmetric}
\end{cases}\) & \(\underset{(\text{\scriptsize \Cref{thm:just-SC-mapcla}})}{\min\left\{\nu, \frac{1}{\mu}\right\} \cdot \max\{d^{3},d\lambda,\beta^{2}\}}\) \\
[2mm] \\
& \(\metric{}\) is \(\begin{cases}
    \text{self-concordant}_{++} \\
    \nu\text{-symmetric}
\end{cases}\) & \(\underset{(\text{\scriptsize \Cref{thm:more-than-SC-mapcla}})}{\min\left\{\nu, \frac{1}{\mu}\right\} \cdot \max\{d\beta,d\lambda,\beta^{2}\}}\) \\
\arrayrulecolor{lightgray}\hline
\rule{0pt}{1.75\normalbaselineskip}
\(\potential{}\) is linear ~~~~ \& & \(\metric{}\) is \(\begin{cases}
    \text{self-concordant}_{++} \\
    \nu\text{-symmetric}
\end{cases}\) & \(\underset{(\text{\scriptsize \Cref{thm:exp-densities-mapcla}})}{\nu \cdot d^{2}}\)
\end{tabular}
\caption{Summary of dependence of problem parameters on \(\delta\)-mixing times of \MAPLA{} from a warm start.
All scalings above hide a polynomial dependence on \(\log(\nicefrac{1}{\delta})\), and assume \(\mu, \beta \geq 1\).
}
\label{tab:mixing-time-scalings}
\end{table}

The precise definitions of the conditions are discussed in \cref{sec:prelims:precond,sec:prelims:func}.
To contextualise these mixing time scalings, we provide a brief comparison of these mixing time scalings to those derived for other constrained sampling algorithms.
We begin by remarking that \textsf{ManifoldMALA} has no known mixing time guarantees for either the unconstrained or the constrained sampling problem.
For \Dikin{}, \citet{kook2024gaussian} show that the mixing time of \Dikin{} scales as \(\min\left\{\nu, \nicefrac{1}{\mu}\right\} \cdot \max\{d, d\lambda\}\) under the same conditions on \(\potential{}\) and \(\metric{}\) as stated in the second row of \cref{tab:mixing-time-scalings}, modulo the \((\beta,\metric{})\)-gradient upper-bounded condition.
The ``self-concordant\textsubscript{++}'' condition in \cref{tab:mixing-time-scalings} is a collection of stronger notions of self-concordance that were identified by \citet{kook2024gaussian}.
These conditions are sufficient to establish mixing time guarantees for \Dikin{} beyond uniform sampling from polytopes, and interestingly also suffice to provide mixing time guarantees for \MAPLA{}.
These stronger notions of self-concordance yield a strictly better dependence on the dimension \(d\) on the mixing time than when \(\metric{}\) is considered to be just self-concordant (first row vs. second row).
As a special case, consider the choice where \(\metric{} = \hessmirror{}\) for a \emph{Legendre type function} \(\mirrorfunc{}\), and assume that potential \(\potential{}\) and \(\hessmirror{}\) satisfy the conditions stated in the first row of \cref{tab:mixing-time-scalings}.
In this setting, the mixing time scaling of \MAMLA{} with \(\phi\) as the mirror function scales as \(\min\{\nu, \nicefrac{1}{\mu}\} \cdot \max\{d^{3},d\lambda,\beta^{2}\}\), which matches the scaling derived for \MAPLA{}.
However, when \(\mirrorfunc{}\) is such that \(\hessmirror{}\) satisfies ``self-concordant\textsubscript{++}'', the analysis of \MAMLA{} fails to leverage the stronger self-concordant properties due to the \emph{dual} nature of the method as highlighted previously.
Lastly, we also derive mixing time guarantees for constrained sampling with exponential densities like those which arise from lifting.
Specifically, the potential here is defined as \(\potential{x} = \sigma^{\top}x\) for non-zero \(\sigma\), and the guarantees we derive are notably independent of \(\sigma\).
This independence of \(\sigma\) is also a feature of the mixing time guarantee for the Riemannian Hamiltonian Monte Carlo (\textsf{RHMC}) algorithm \citep{girolami2011riemann}, which is another MCMC method that also uses a geometric properties of the domain in the form of a metric \(\metric{}\), and can be applied to the constrained sampling problem.
\citet{kook2023condition} specifically show that the mixing time of discretised \textsf{RHMC} for sampling from \(\targetdist\) with a linear potential whose support \(\primalspace\) is a polytope formed by \(m\) constraints scales as \(m \cdot d^{3}\) when \(\metric{}\) is given by the Hessian of the log-barrier function of this polytope.
For this specific constrained sampling problem and choice of \(\metric{}\), the mixing time of \MAPLA{} scales as \(m \cdot d^{2}\) as this metric is  both self-concordant\textsubscript{++} and \(\nu\)-symmetric with \(\nu = m\).

\textbf{Organisation~~} The remainder of this paper is organised as follows. The next section (\cref{sec:back-rel-work}) is dedicated to a review of related work.
We formally introduce \MAPLA{} and discuss its implementation in \cref{sec:algo}, and state its mixing time guarantees in \cref{sec:mixing-guarantees}.
In \cref{sec:expts}, we provide some numerical experiments, and in \cref{sec:proofs}, we present detailed proofs of the theorems that appear in \cref{sec:algo}.
We conclude with a summary and some open questions in \Cref{sec:conclusion}.
Miscellaneous propositions used in the proofs that appear in \cref{sec:proofs} are deferred to \cref{sec:addendum}.

\section{Background and related work}
\label{sec:back-rel-work}
Many of the proposed MCMC methods for sampling can be roughly separated into two classes: zeroth-order and first-order methods.
Zeroth-order methods are based on querying the potential \(\potential{}\) or equivalently, the density up to normalisation constants, while first-order methods are based on querying the gradient of the potential \(\gradpotential{}\) in addition to \(\potential{}\).
For constrained sampling, \Dikin{} \citep{narayanan2017efficient,kook2024gaussian} is a practical and well-studied zeroth order method, which was proposed originally for the task of uniform sampling from \(\primalspace\) when \(\primalspace\) is a polytope \citep{kannan2009random}, and with applications in approximate volume computation.
This is a notable innovation over earlier methods for uniform sampling \(\primalspace\) applied to polytopes such as \textsf{Hit-And-Run} \citep{smith1984efficient,belisle1993hit,lovasz1999hit} and \textsf{BallWalk} \citep{lovasz1993random,kannan1997random} due to its independence on the conditioning of \(\primalspace\).
\textsf{MetropolisRandomWalk} \citep{mengersen1996rates,dwivedi2018log} is also another zeroth-order method related to \textsf{BallWalk} that is suited for unconstrained sampling.

Several first-order algorithms for sampling are based on discretisations of the Langevin dynamics and its variants.
For unconstrained sampling, the unadjusted Langevin algorithm (\ref{eq:ULA}) is a simple yet popular first-order method which is derived as the forward Euler-Maruyama discretisation of the continuous time Langevin dynamics (\ref{eq:LD}).
The Langevin dynamics is equivalent to the gradient flow of the KL divergence \(\mu \mapsto \KLdist(\mu \| \targetdist)\) in the space of probability measures equipped with the Wasserstein metric (\(\calP_{2}(\bbR^{d}), W_{2}\)) \citep{jordan1998variational,wibisono2018sampling}.
\begin{align*}
    \rmd X_{t} = -\nabla \potential{X_{t}} \rmd t + \sqrt{2}~\rmd B_{t}
    \enskip \Leftrightarrow &\enskip
    \partial_{t}\rho_{t} = \nabla \cdot \left( \rho_{t}\nabla \log \frac{\rho_{t}}{\targetdens}\right),~\rho_{t} = \text{Law}(X_{t})~.
    \tag{\textsf{LD}}\label{eq:LD}\\
    X_{k + 1} - X_{k} = -h \cdot \gradpotential{X_{k}} &+ \sqrt{2h} \cdot \xi_{k}; \quad \xi_{k} \sim \calN(\bm{0}, \rmI_{d \times d})~.
    \tag{\textsf{ULA}}\label{eq:ULA}
\end{align*}
When the Brownian motion \(\rmd B_{t}\) is excluded in \ref{eq:LD}, we obtain an ODE that is the gradient flow of \(f\).
Analogously, without the Gaussian vector \(\xi_{k}\) in \ref{eq:ULA}, the iteration resembles the gradient descent algorithm with step size \(h\).
\ref{eq:ULA} has been shown to be \emph{biased} i.e., the limiting distribution of \(X_{k}\) as \(k \to \infty\) does not coincide with \(\targetdist\) \citep{durmus2017convergence,dalalyan2017theoretical,dalalyan2019user}.
Applying the Metropolis-Hastings filter to the proposal distribution defined by a single step of \ref{eq:ULA} leads to the Metropolis-adjusted Langevin algorithm (\textsf{MALA}) \citep{roberts1996exponential} which is \emph{unbiased}.
Much like \ref{eq:ULA}, \textsf{MALA} has seen a variety of non-asymptotic analyses in recent years (for e.g., \citet{dwivedi2018log,chewi2021optimal,durmus2022geometric,wu2022minimax}).
Other first-order methods related to the Langevin dynamics (\ref{eq:LD}) include the underdamped Langevin MCMC \citep{cheng2018convergence,eberle2019couplings} and Hamiltonian Monte Carlo (\textsf{HMC}) \citep{neal2011mcmc,durmus2017convergence,bou2020coupling}.

Turning to the task of optimising \(f\), the gradient flow of \(\potential{}\) can be modified to incorporate geometric information about the state space encoded as a metric \(\metric{}\), and this results in the natural gradient flow \citep{amari1998natural}, where the gradient \(\gradpotential{}\) is replaced by the \emph{natural gradient} as mentioned briefly earlier.
This \emph{natural gradient} is defined as \(\metric{}^{-1}\gradpotential{}\), and corresponds to the direction of steepest ascent under the metric \(\metric{}\).
Similarly for sampling, the Langevin dynamics (\ref{eq:LD}) can also be modified to incorporate this geometric information, and this leads to the weighted Langevin dynamics (\ref{eq:WLD}) \citep{girolami2011riemann}.
The SDE that defines \ref{eq:WLD} differs from \ref{eq:LD} in both the drift and the diffusion matrix: the drift consists of the natural gradient of \(\potential{}\) and an additional correction term \((\nabla \cdot \metric{}^{-1})\) that accounts for change in the diffusion matrix in this SDE.
The correction term ensures that the stationary distribution of the SDE is \(\targetdist\), a fact that is apparent by considering the equivalent PDE of the weighted Langevin dynamics.
The forward Euler-Maruyama discretisation of \ref{eq:WLD} leads to the weighted Langevin algorithm (\ref{eq:WLA}).
\begin{gather*}
    \begin{aligned}
        \rmd X_{t} &= \frakm(X_{t}) \rmd t + \sqrt{2} \cdot \metric{X_{t}}^{-\nicefrac{1}{2}}~\rmd B_{t} \\
        \frakm(X_{t}) &= (\nabla \cdot \metric{}^{-1})(X_{t}) - \metric{X_{t}}^{-1}\gradpotential{X_{t}}
    \end{aligned}
    \enskip\Leftrightarrow\enskip
    \partial_{t}\rho_{t} = \nabla \cdot \left(\rho_{t}\metric{}^{-1}\nabla \log \frac{\rho_{t}}{\targetdens}\right), ~\rho_{t} = \mathrm{Law}(X_{t})~.
    \tag{\textsf{WLD}}\label{eq:WLD} \\
    X_{k + 1} - X_{k} = h \cdot \frakm(X_{k}) + \sqrt{2h} \cdot \metric{X_{k}}^{-\nicefrac{1}{2}}~\xi_{k}; \quad \xi_{k} \sim \gaussian{\bm{0}}{\rmI_{d \times d}}~.\tag{\textsf{WLA}}\label{eq:WLA}
\end{gather*}
If \(\metric{}\) is \emph{position-independent}, the correction term \(\nabla \cdot \metric{}^{-1}\) is identically \(\bm{0}\).
Well-chosen position-independent metrics have proven to be useful by yielding faster mixing of the Markov chain (either of \ref{eq:WLA} or its Metropolis-adjusted variant) \citep{cotter2013mcmc,titsias2018auxiliary,titsias2024optimal} when the distribution \(\targetdist\) is highly anisotropic.
On the other hand, \emph{position-dependent} metrics enable leveraging the geometry of the state space better.
The PDE form of \ref{eq:WLD} also lends to the interpretation of \ref{eq:WLD} as the gradient flow of \(\mu \mapsto \KLdist(\mu \| \targetdist)\) with respect to the Wasserstein metric on the Riemannian manifold defined by the state space and the metric \(\metric{}\).
\citet{girolami2011riemann} propose a Metropolis adjustment to \ref{eq:WLA} which they call \textsf{ManifoldMALA}, and observe that setting \(\metric{}\) to be the Fisher information matrix in \textsf{ManifoldMALA} results in quicker mixing compared to \textsf{MALA} on a collection of unconstrained sampling problems.
In general however, computing \(\nabla \cdot \metric{}^{-1}\) for \ref{eq:WLA} can be difficult, as also previously highlighted in \citet{bourabee2014metropolis} who note that \ref{eq:WLA} corresponds to an Ermak-Cammon scheme for Brownian dynamics with hydrodynamic interactions.
\citet{bourabee2014metropolis} propose a novel discretisation for \ref{eq:WLD} based on a two-stage Runge-Kutta scheme in combination with a Metropolis adjustment, and give asymptotic guarantees for this discretisation.
A simpler but inaccurate discretisation of \ref{eq:WLD} can be defined by merely omitting the \(\nabla \cdot \metric{}^{-1}\) term in \ref{eq:WLA}, which turns \ref{eq:WLA} into the \emph{Preconditioned Langevin Algorithm} (\ref{eq:PLA}) defined below.
\begin{equation*}
    X_{k + 1} - X_{k} = -h \cdot \metric{X_{k}}^{-1}\gradpotential{X_{k}} + \sqrt{2h} \cdot \metric{X_{k}}^{-\nicefrac{1}{2}}~\xi_{k}~; \quad \xi_{k} \sim \gaussian{\bm{0}}{\rmI_{d \times d}} \tag{\textsf{PLA}}\label{eq:PLA}~.
\end{equation*}
It is important to note that even as \(h \to 0\), the stationary distribution of \ref{eq:PLA} is not\footnote{unless \(\metric{}\) is position-independent.} \(\targetdist\), and is also likely to carry a higher bias than \ref{eq:WLA} on account of being an unfaithful discretisation of \ref{eq:WLD}.
Despite this, a single step of \ref{eq:PLA} serves as a useful proposal distribution for a Metropolis-adjusted scheme as we demonstrate later.
Without \(\xi_{k}\) in \ref{eq:PLA}, this coincides with the natural gradient step with step size \(h\), and when \(\metric{} = \nabla^{2}\potential{}\), this resembles the Newton method.

When dealing with constrained distributions, special care has to be taken when working with discretisations of \ref{eq:LD} or its weighted variant \ref{eq:WLD}.
This is because iterates generated by \ref{eq:ULA} / \ref{eq:WLA} / \ref{eq:PLA} could escape \(\primalspace\) since the Gaussian random vector is unconstrained.
The projected Langevin algorithm \citep{bubeck2018sampling} modifies \ref{eq:ULA} by including an explicit projection onto \(\primalspace\), and is motivated by the projected gradient descent algorithm for optimising functions over a constrained feasibility set.
Another class of approaches for constrained sampling are mirror Langevin algorithm (which were briefly introduced in \cref{sec:intro}), and are based on discretisations of the mirror Langevin dynamics (\textsf{MLD}) \citep{zhang2020wasserstein,chewi2020exponential}.
The mirror Langevin dynamics are defined as follows; the mirror function \(\phi\) here is a Legendre type function.
\begin{gather*}
Y_{t} = \primtodual{X_{t}}; \quad \rmd Y_{t} = -\gradpotential{X_{t}}\rmd t + \sqrt{2} \cdot \hessmirror{X_{t}}^{\nicefrac{1}{2}}~\rmd B_{t}~.\tag{\textsf{MLD}\textsubscript{1}}\label{eq:MLD-primal-dual} \\
\rmd X_{t} = ((\nabla \cdot M)(X_{t}) - M(X_{t})\gradpotential{X_{t}}) \rmd t + \sqrt{2} \cdot M(X_{t})^{\nicefrac{1}{2}}\rmd B_{t}; \enskip M = \hessmirror{}^{-1}\tag{\textsf{MLD}\textsubscript{2}}\label{eq:MLD-primal}
~.
\end{gather*}
The assumption that \(\mirrorfunc{}\) is a Legendre type in conjunction with It\^{o}'s lemma establishes the equivalence of \ref{eq:MLD-primal-dual} and \ref{eq:MLD-primal} \citep{zhang2020wasserstein,jiang2021mirror}.
Through \ref{eq:MLD-primal}, we see that the mirror Langevin dynamics is a special case of \ref{eq:WLD} with \(\metric{} = \hessmirror{}\).
The two-step definition (\ref{eq:MLD-primal-dual}) is particularly amenable to discretisation as this automatically ensures feasibility of iterates when \(\primalspace\) is compact\footnote{
This is due to the fact that the range of \(\dualtoprim{}\) is \(\bbR^{d}\) when \(\primalspace\) is compact \citep[Corr. 13.3.1]{rockafellar1997convex}.
}.
The Euler-Maruyama discretisation of \ref{eq:MLD-primal-dual} is the mirror Langevin algorithm (\textsf{MLA}) \citep{zhang2020wasserstein,li2022mirror}.
Prior works \citep{ahn2021efficient,jiang2021mirror} propose other novel discretisations of \ref{eq:MLD-primal-dual} which are not practically feasible as they involve simulating a SDE.
These discretisations including \textsf{MLA} are \emph{biased}; for fixed \(h > 0\), the law of the iterate \(X_{k}\) does not converge to \(\targetdist\), and this is also a feature of \ref{eq:ULA} for unconstrained sampling, and projected \ref{eq:ULA} for constrained sampling.
\citet{srinivasan2024fast} propose augmenting a Metropolis-Hastings filter to the Markov chain induced by a single step of \textsf{MLA} (\MAMLA{}, discussed previously in \cref{sec:intro}), which eliminates the bias that \textsf{MLA} carries.
Implicitly, both \textsf{MLA} and \MAMLA{} make two key assumptions: that \(\primalspace\) is compact, and that \(\dualtoprim{}\) can be computed exactly.
The compactness of \(\primalspace\) ensures that \(\primalspace^{\star} = \bbR^{d}\), which is necessary to generate a proposal.
When \(\primalspace\) is not compact, \(\primalspace^{\star}\) is not necessarily \(\bbR^{d}\), and an explicit projection onto \(\primalspace^{\star}\) is required for the proposal to be well-defined.
For \MAMLA{}, it suffices to induce an explicit rejection of dual proposals that don't belong in \(\primalspace^{\star}\), but this requires a membership oracle for \(\primalspace^{\star}\) instead, which is non-trivial to obtain.
With regards to computing \(\dualtoprim{}\), while closed form expression for \(\mirrorfuncdual{}\) exist in certain special cases \(\mirrorfunc{}\), it generally does cannot be derived. In such cases, \(\dualtoprim{}\) can only be computed approximately by solving the convex problem which defines the convex conjugate to within a certain tolerance.

In addition to proposing \textsf{ManifoldMALA}, \citet{girolami2011riemann} also adapt Hamiltonian Monte Carlo to the geometry of the state space and propose the \emph{Riemannian} Hamiltonian Monte Carlo algorithm (\textsf{RHMC}).
The primary challenge with \textsf{RHMC} pertains to implementation, as its discretisation is not as straightforward as \ref{eq:WLA} for \ref{eq:WLD}.
\citet{lee2018convergence,gatmiry2024sampling} derive guarantees for \textsf{RHMC} given an ideal integrator, while \citet{kook2023condition,noble2023unbiased} investigate discretised integrators for \textsf{RHMC}.
Recent advances in proximal methods for sampling \citep{chen2022improved} have also led to novel samplers for sampling over non-Euclidean spaces using the log-Laplace transform \citep{gopi2023algorithmic}, and for uniform sampling from compact \(\primalspace\) under minimal assumptions \citep{kook2024and}.
Another class of first-order approaches for constrained sampling \citep{brosse2017sampling,gurbuzbalaban2024penalized} propose obtaining a good unconstrained approximation \(\widetilde{\targetdist}\) of the target \(\targetdist\), but this does not eliminate the likelihood of iterates lying outside \(\primalspace\).
This enables running the simpler unadjusted Langevin algorithm over \(\widetilde{\targetdist}\) to obtain approximate samples from \(\targetdist\).
Recent work by \citet{bonet2024mirror} translate mirror and natural gradient methods for minimising functions over Euclidean spaces to instead minimise functionals over space of probability measures endowed with the Wasserstein metric.

\section{The Metropolis-adjusted Preconditioned Langevin Algorithm}
\label{sec:algo}
Here, we formally introduce the Metropolis-adjusted Preconditioned Langevin Algorithm (\nameref{alg:mapla}).
Prior to this, we first present some notation not introduced previously.

\textbf{Notation}~~The set of \(d \times d\) symmetric positive definite matrices is denoted by \(\bbS_{+}^{d}\).
For \(A \in \bbS_{+}^{d}\), and \(x, y \in \bbR^{d}\), we use \(\langle x, y\rangle_{A}\) to denote \(\langle x, Ay\rangle\) and \(\|x\|_{A} = \sqrt{\langle x, x\rangle_{A}}\).
When the subscript \(A\) is omitted as in \(\langle x, y\rangle\), then this corresponds to the \(A = \rmI_{d \times d}\).
The \(\ell_{p}\) norm of \(x\) is denoted by \(\|x\|_{p}\).
For a distribution \(\nu\), \(d\nu(x)\) denotes its density at \(x\) (if it exists) unless otherwise specified.

We first give an introduction to Markov chains and the Metropolis adjustment more generally, expanding on the brief description in \Cref{sec:intro}.
A (time-homogeneous) Markov chain over domain \(\primalspace\) is defined as a collection of transition probability measures \(\bfP = \{\calP_{x} : x \in \primalspace\}\), and we refer to \(\calP_{x}\) as the \emph{one-step} distribution associated with \(x\).
Assume that for every \(x \in \primalspace\), the proposal distribution \(\calP_{x}\) has density function \(p_{x}\).
The transition operator \(\bbT_{\bfP}\) on the space of probability measures is defined as
\begin{equation*}
    (\bbT_{\bfP}\mu)(S) = \int_{\calK} \calP_{y}(S) \cdot \rmd\mu(y)~\qquad \text{and } (\bbT_{\bfP}^{k}\mu)(S) = \left(\bbT_{\bfP}(\bbT_{\bfP}^{k-1}\mu)\right)(S)~.
\end{equation*}
A probability measure \(\nu\) is a \emph{stationary} measure of a Markov chain \(\bfP\) if \(\bbT_{\bfP}\nu = \nu\).
If \(\bfP\) is \emph{ergodic}, then \(\nu\) is the unique stationary measure.
A related but stronger notion is \emph{reversibility}; a Markov chain \(\bfP\) is reversible with respect to \(\nu\) if for any measurable subsets of \(A, B\) of \(\primalspace\),
\begin{equation*}
    \int_{A}\calP_{x}(B) \cdot \rmd\nu(x) = \int_{B}\calP_{y}(A) \cdot \rmd\nu(y)~.
\end{equation*}
If \(\bfP\) is ergodic and reversible with respect to \(\nu\), then \(\nu\) is the stationary measure of \(\bfP\).

The Metropolis adjustment is one technique to produce a new Markov chain that is reversible with respect to a certain distribution.
Formally, suppose we intend to adjust (i.e., make reversible) \(\bfP\) with respect to a distribution \(\targetdist\) whose density is \(\targetdens\).
The Metropolis adjustment works by including an explicit accept-reject step called the Metropolis-Hastings (MH) filter.
The implementation of the MH filter involves computing an \emph{acceptance ratio} between feasible points \(X\) and \(Z\) defined as
\begin{equation}
\label{eq:MH-accept-def}
    p_{\mathrm{accept}}(Z; X) \defeq \min\left\{1, \frac{\targetdens(Z) \cdot p_{Z}(X)}{\targetdens(X) \cdot p_{X}(Z)}\right\} = \min\left\{1, \frac{e^{-\potential{Z}}}{e^{-\potential{X}}} \cdot \frac{p_{Z}(X)}{p_{X}(Z)}\right\}~.
\end{equation}
Notably, this only requires the unnormalised target density function \(e^{-\potential{}}\) unlike rejection sampling.
The Markov chain formed by the Metropolis adjustment is also known to be optimal in a certain sense \citep{billera2001geometric}, and is generally preferred to other similar techniques like the Barker correction which uses \(\frac{A}{1 + A}\) in lieu of \(\min\{1, A\}\) in \cref{eq:MH-accept-def}.

Consider the Markov chain \(\bfP\) defined by a single step of \ref{eq:PLA}.
From any \(x \in \interior{\primalspace}\)\footnote{Since \(\primalspace\) is convex, \(\partial\primalspace\) is a Lebesgue null set, and hence points on the boundary can be safely disregarded.}, \ref{eq:PLA} returns
\begin{equation*}
    x' = x - h \cdot \metric{x}^{-1}\gradpotential{x} + \sqrt{2h} \cdot \metric{x}^{-\nicefrac{1}{2}}~\xi
\end{equation*}
where \(\xi\) is independently drawn from \(\gaussian{\bm{0}}{\rmI_{d \times d}}\).
Hence, for any \(x \in \interior{\primalspace}\), \(\calP_{x}\) is defined as
\begin{equation*}
    \calP_{x} = \gaussian{x - h \cdot \metric{x}^{-1}\gradpotential{x}}{2h \cdot \metric{x}^{-1}}~.
\end{equation*}
As introducted in \Cref{sec:intro}, \nameref{alg:mapla} results from performing a Metropolis-adjustment of the Markov chain \(\bfP\) with respect to \(\targetdist\).
We use \(\bfT = \{\calT_{x} : x \in \primalspace\}\) to denote this Metropolis-adjusted version of \(\bfP\), where \(\calT_{x}\) is the one-step distribution from \(x \in \primalspace\) per \nameref{alg:mapla}.

\begin{algorithm}[t]
\DontPrintSemicolon

\caption{Metropolis-adjusted Preconditioned Langevin Algorithm (\MAPLA{})}
\algotitle{\MAPLA{}}{alg:mapla}

\SetKwInOut{Input}{Input}
\SetKwInOut{Output}{Output}
\SetAlgoLined
\Input{Potential \(\potential{}\) of \(\targetdist\), convex support \(\primalspace \subset \bbR^{d}\), metric \(\metric{} : \interior{\primalspace} \to \bbS_{+}^{d}\), step size \(h > 0\), iterations \(K\), initial distribution \(\targetdist_{0}\)}

    Sample \(x_{0} \sim \Pi_{0}\).

    \For{\(k \leftarrow 0\) \KwTo \(K - 1\)}{
    Sample a random vector \(\xi_{k} \sim \gaussian{\bm{0}}{\rmI_{d \times d}}\).

    Generate proposal \(z = x_{k} - h \cdot \metric{x_{k}}^{-1}\gradpotential{x_{k}} + \sqrt{2h} \cdot \metric{x_{k}}^{-\nicefrac{1}{2}}~\xi_{k}\).\label{alg:proposal-step}

    \eIf{\(z \not\in \primalspace\)}{
        Set \(x_{k + 1} = x_{k}\).
    }{
        Compute acceptance ratio \(p_{\mathrm{accept}}(z; x_{k})\) defined in \cref{eq:MH-accept-def}, where \(p_{y}\) is the density of \(\gaussian{y - h \cdot \metric{y}^{-1}\gradpotential{y}}{2h \cdot \metric{y}^{-1}}\).\label{alg:acceptance-ratio-step}

        Obtain \(U \sim \mathrm{Unif}([0, 1])\).

        \eIf{\(U \leq p_{\mathrm{accept}}(z; x_{k})\)}{
            Set \(x_{k + 1} = z\).
        }{
            Set \(x_{k + 1} = x_{k}\).
        }
    }
}

\Output{\(x_{K}\)}
\end{algorithm}

\subsection{Examples of metric \(\metric{}\) for certain domains}
\label{sec:algo:examples}

Let \(\primalspace\) be a polytope of the form \(\{x : x \in \bbR^{d},~a_{i}^{\top}x \leq b_{i} ~\forall ~i \in [m]\}\) where \(m \geq d\).
A natural choice of \(\metric{}\) for this domain is the Hessian of its log-barrier function, and can be generalised to incorporate weights for each linear constraint as
\begin{equation*}
    \metric{}_{\{Ax \leq b\}}(x) = A_{x}^{\top}WA_{x} \text{ where } [A_{x}]_{i}^{\top} = \frac{a_{i}}{(b_{i} - a_{i}^{\top}x)} \in \bbR^{d} \text{ and } W = \mathrm{diag}(\bm{w})~.
\end{equation*}
The vector \(\bm{w} \in \bbR^{m}\) contains non-negative entries, and may depend on \(x\).
Special cases include
\begin{enumerate}[itemsep=0pt, leftmargin=*]
\item the Vaidya metric \citep{vaidya1996new}, where \(w_{i} = \frac{d}{m} + a_{i}^{\top}(A_{x}^{\top}A_{x})^{-1}a_{i}\),
\item the John metric \citep{gustafson2018john}, where \(\bm{w}\) is defined as
\begin{equation*}
    \bm{w} = \argmin_{w \in \bbR^{m}_{+}} \log \det A_{x}^{\top}WA_{x} \quad \text{such that } \bm{1}^{\top}w = m~.
\end{equation*}
\item the \(p\)-Lewis-weights metric \citep{lee2019solving} which modifies the John metric to ensure \(\bm{w}\) varies smoothly as a function \(x\).
More precisely, it is defined as
\begin{equation*}
    \bm{w} = \argmax_{w \in \bbR^{m}_{+}} -\log \det (A_{x}^{\top}W^{1 - \nicefrac{2}{p}}A_{x}) + (1 - \nicefrac{2}{p})\cdot \bm{1}^{\top}w
\end{equation*}
\end{enumerate}

Polytopes are regions defined by intersection of half-spaces defined by a collection of linear inequalities.
Moving past linear inequalities, we consider quadratic inequalities, which define ellipsoids.
An ellipsoid is defined as \(\left\{x : x \in \bbR^{d},~\|x - c\|_{D}^{2} \leq 1\right\}\) where \(D \in \bbR^{d \times d}\) is a symmetric positive definite matrix and \(c \in \bbR^{d}\).
Analogous to the polytope, a natural choice of \(\metric{}\) here is the Hessian of its log-barrier function.
An ellipsoid can also be interpreted at the \(1\)-sublevel set of a quadratic function.
Relatedly, the epigraph of this quadratic function is the indexed union of all its \(t\)-sublevel sets for \(t \geq 0\), and this arises when lifting the Gaussian distribution as remarked previously in \cref{sec:intro}.
An example of a metric for this domain is
\begin{equation*}
    \metric{}_{\mathrm{Ellip}(c, D)}(x, t) = \nabla^{2}_{(x, t)} \varphi_{\mathrm{Ellip}(c, D)}(x, t)~; \qquad \varphi_{\mathrm{Ellip}(c, D)}(x, t) = -\log(t - \|x - c\|_{D}^{2})~.
\end{equation*}

Generalising from the \(\ell_{1}\)-ball (a polytope formed by \(m = 2^{d}\) constraints), and the \(\ell_{2}\)-ball (an ellipsoid with \(c = \bm{0}\) and \(D = \rmI_{d \times d}\)), we have the \(\ell_{p}\)-ball for an arbitrary \(p \geq 1\) defined as \(\{x : x \in \bbR^{d}, ~\|x\|_{p}^{p} \leq 1\}\).
Using the Hessian of the log-barrier of \(-\log(1 - \|x\|_{p}^{p})\) as the metric for this domain poses non-trivial difficulties theoretically.
To circumvent this, we use the following equivalence
\begin{equation*}
    \|x\|_{p} \leq 1 \Leftrightarrow \exists ~v \in \bbR^{d} \text{ such that } |x_{i}|^{p} \leq v_{i}~\forall~ i \in \{1, \ldots, d\} \text{ and }\sum_{i=1}^{d}v_{i} \leq 1~,
\end{equation*}
which suggests considering an extended domain in \(\bbR^{2d}\).
This extended domain is the intersection of the halfspace \(\{v : v \in \bbR^{d}, \bm{1}^{\top}v \leq 1\}\) and the product of \(d\) subsets of \(\bbR^{2}\) given as \(\prod_{i=1}^{d}\{(x_{i}, v_{i}) : |x_{i}|^{p} \leq v_{i}\}\).
Therefore, we can define a metric for the extended domain given a metric \(\metric{}_{2,p}\) for the set \(\{(y, t) : |y|^{2} \leq t\} \subset \bbR^{2}\).
A popular choice for \(\metric{}_{2,p}(y, t)\) is \(\nabla_{(y, t)}^{2}\varphi_{\ell_{p}}(y, t)\) where \(\varphi_{\ell_{p}}(y, t) = -\log(t) - \log(t^{\nicefrac{2}{p}} - y^{2})\), with which we can be define the metric for the extended domain as
\begin{equation*}
    \metric{}_{p}(x, v) = \frakP \left(\bigoplus_{i=1}^{d}\metric{}_{2,p}(x_{i}, v_{i})\right) \frakP^{\top} +
    \bm{0}_{d \times d} \oplus \metric{}_{\{\bm{1}^{\top}v \leq 1\}}(v)~.
\end{equation*}
The operation \(A \oplus B\) creates a block diagonal matrix with \(A\) and \(B\) on the diagonal\footnote{also referred to as the \emph{direct sum}.}, and \(\frakP\) is a permutation matrix that ensures consistency with respect to the ordering of inputs (\(x\) first, \(v\) second).
This extension procedure is more generally applicable to level sets of separable functions \citep[\S 5.4.8]{nesterov2018lectures}.
Another example is the entropic ball defined as \(\{x : x \in \bbR^{d}_{+},~\sum_{i=1}^{d}x_{i}\log x_{i} \leq 1\}\).
Similar to the \(\ell_{p}\)-ball, the metric \(\metric{}_{\mathsf{ent}}\) for the extended entropic ball is dependent on the metric \(\metric{}_{2,\mathsf{ent}}(y, t)\) for the 2-dimensional set \(\{(y, t) : y\log y \leq t\}\).
A viable option for \(\metric{}_{2,\mathsf{ent}}(y, t)\) is \(\nabla_{(y, t)}^{2}\varphi_{\mathsf{ent}}(y, t)\) where \(\varphi_{\mathsf{ent}}(y, t) = -\log y - \log(t - y \log y)\).

We refer the reader to \citet[\S 5.2]{nesterov1994interior} for a general calculus to combine metrics.
Their specific focus is the self-concordance of these combinations, which is central to the analysis of interior-point methods for optimisation.
\citet{kook2024gaussian} investigate other properties (``self-concordance\textsubscript{++}'') of such metrics which are pertinent for constrained sampling, and are useful for deriving mixing time guarantees for \Dikin{} and as we show, for \nameref{alg:mapla}.

\subsection{Implementing \nameref{alg:mapla}}

Given a potential \(\potential{}\) and metric \(\metric{}\) and a black-box membership oracle for \(\primalspace\) after any pre-processing, the key computational steps in each iteration of \nameref{alg:mapla} are (1) generating the proposal \(z\) (line \ref{alg:proposal-step}), and (2) computing the acceptance ratio \(p_{\mathrm{accept}}(z; x_{k})\) (line \ref{alg:acceptance-ratio-step}) (if \(z \in \primalspace\)).
In this discussion, we give practically applicable approaches for the aforementioned key steps as follows.

\paragraph{Generating a proposal}
This step can be decomposed into three steps.
\begin{equation*}
    \underbrace{\widetilde{\xi_{k}} \sim \gaussian{\bm{0}}{\metric{x_{k}}^{-1}}}_{\text{Step \ref{alg:proposal-step}.1}}\quad ; \quad \underbrace{v_{x_{k}} = \metric{x_{k}}^{-1}\gradpotential{x_{k}}}_{\text{Step \ref{alg:proposal-step}.2}} \quad ; \quad \underbrace{z = x_{k} - h \cdot v_{x_{k}} + \sqrt{2h} \cdot \widetilde{\xi_{k}}}_{\text{Step \ref{alg:proposal-step}.3}}~.
\end{equation*}
\begin{description}
\item [Step \ref{alg:proposal-step}.1:] In the general case, the most computationally efficient and numerically stable method to sample from \(\gaussian{\bm{0}}{\metric{x_{k}}^{-1}}\) is to (1) compute the Cholesky decomposition of \(\metric{x_{k}}\), and (2) perform a triangular solve with \(\xi_{k} \sim \gaussian{\bm{0}}{\rmI_{d}}\).
The correctness of this method is due to the following: let \(L_{x_{k}}\) be the lower-triangular Cholesky factor satisfyin \(L_{x_{k}}L_{x_{k}}^{\top} = \metric{x_{k}}\).
Then, \(L_{x_{k}}^{-\top}\xi_{k}\) has distribution \(\gaussian{\bm{0}}{L_{x_{k}}^{-\top}L_{x_{k}}^{-1}} = \gaussian{\bm{0}}{\metric{x_{k}}^{-1}}\).

\item [Step \ref{alg:proposal-step}.2:] The vector \(v_{x_{k}}\) is the solution to the system \(\metric{x_{k}}v_{x_{k}} = \gradpotential{x_{k}}\) and with the Cholesky decomposition can be equivalently written as \(L_{x_{k}}L_{x_{k}}^{\top}v_{x_{k}} = \gradpotential{x_{k}}\).
Thus, \(v_{x_{k}}\) can be computed through two triangular solves.

\item [Step \ref{alg:proposal-step}.3:] This is a combination of scalar multiplication and vector addition operations.
\end{description}

\paragraph{Computing the acceptance ratio}

We instead work with log acceptance ratio as it is more numerically stable.
By definition,
\begin{multline*}
    \log \frac{\targetdens(z)p_{z}(x_{k})}{\targetdens(x_{k})p_{x_{k}}(z)} = f(x_{k}) - f(z) + \frac{1}{2}\log \det \metric{z}\metric{x_{k}}^{-1} \\
    + \frac{1}{4h} \cdot \left(\|z - x_{k} + h \cdot \metric{x_{k}}^{-1}\gradpotential{x_{k}}\|_{\metric{x_{k}}}^{2} - \|x_{k} - z + h \cdot \metric{z}^{-1}\gradpotential{z}\|_{\metric{z}}^{2}\right)~.
\end{multline*}

From the expression on the right hand side, there are two important quantities
\begin{equation*}
    \underbrace{\log \det \metric{z} \metric{x_{k}}^{-1}}_{\text{Step \ref{alg:acceptance-ratio-step}.1}}\quad ; \quad \underbrace{\|z - x_{k} + h \cdot \metric{x_{k}}^{-1}\gradpotential{x_{k}}\|_{\metric{x_{k}}}^{2} - \|x_{k} - z + h \cdot \metric{z}^{-1}\gradpotential{z}\|_{\metric{z}}^{2}}_{\text{Step \ref{alg:acceptance-ratio-step}.2}}~.
\end{equation*}

Let the Cholesky factor of \(\metric{z}\) be \(L_{z}\)\footnote{lower triangular by convention.} which satisfies \(L_{z}L_{z}^{\top} = \metric{z}\).
This Cholesky factor aids with both computational steps as we demonstrate below.
\begin{description}
    \item [Step \ref{alg:acceptance-ratio-step}.1:] Since \(\det A^{-1} = \frac{1}{\det A}\) and \(\det A = \det A^{\top}\), we have
    \begin{equation*}
        \frac{1}{2}\log \det \metric{z}\metric{x_{k}}^{-1} = \frac{1}{2}\log \det \metric{z} - \frac{1}{2}\log \det \metric{x_{k}} = \log \det L_{z} - \log \det L_{x_{k}}~.
    \end{equation*}
    Since \(L_{x_{k}}\) and \(L_{z}\) are triangular, their log determinants can be easily computed by summing the log of the values on their diagonals.

    \item [Step \ref{alg:acceptance-ratio-step}.2:] 
    Recall that \(z = x_{k} - h \cdot \metric{x_{k}}^{-1}\gradpotential{x_{k}} + \sqrt{2h} \cdot \metric{x_{k}}^{-\nicefrac{1}{2}}~\xi_{k}\), where \(\xi_{k}\) is from line 3.
    This leads to two essential simplifications.
    \begin{align*}
        \|z - x_{k} + h \cdot \metric{x_{k}}^{-1}\gradpotential{x_{k}}\|_{\metric{x_{k}}}^{2} &= 2h \cdot \|\xi_{k}\|^{2}~, \\
        \|x_{k} - z + h \cdot \metric{z}^{-1}\gradpotential{z}\|_{\metric{z}}^{2} &= \|h \cdot (\metric{z}^{-1}\gradpotential{z} + \metric{x_{k}}^{-1}\gradpotential{x_{k}}) - \sqrt{2h} \cdot \widetilde{\xi_{k}}\|_{\metric{z}}^{2} \\
        &= \|L_{z}\{h \cdot (v_{z} + v_{x_{k}}) - \sqrt{2h} \cdot \widetilde{\xi_{k}}\}\|^{2}~.
    \end{align*}
    In the second equation, \(\widetilde{\xi_{k}}\) is from Step \ref{alg:proposal-step}.1.
    Hence for this step, we require computing \(v_{z} = \metric{z}^{-1}\gradpotential{z}\), which can be obtained in the same manner as Step \ref{alg:proposal-step}.2, and an additional triangular matrix-vector product to compute \(L_{z}\{h \cdot (v_{z} + v_{x_{k}}) - \sqrt{2h} \cdot \widetilde{\xi_{k}}\}\).
    The squared norm of a vector is the sum of squares of its entries, which completes the computation.
\end{description}

Notwithstanding elementwise operations, this procedure suggests that the cost of lines \ref{alg:proposal-step} and \ref{alg:acceptance-ratio-step} together scale as \(5 \cdot \mathrm{Cost}_{\mathrm{TriSol}} + 2 \cdot \mathrm{Cost}_{\mathrm{CholDec}} + \mathrm{Cost}_{\mathrm{TriMatVec}}\).
This naive analysis disregards the likelihood of \(z\) being accepted, and rejection of \(z\) could either be due to the proposal escaping \(\primalspace\) (which eliminates the acceptance ratio calculation), or due to the Metropolis-Hastings accept-reject step.
If \(z\) is accepted, then \(L_{z}\) and \(v_{z}\) would not have to be recomputed in the following iteration as they can be cached.
We can draw a comparison to \Dikin{}, which is similar to \nameref{alg:mapla} as elucidated in \cref{sec:intro}.
The key difference here is that \Dikin{} avoids having to compute \(v_{z}\) and \(v_{x_{k}}\) which are based on triangular solves.
Nonetheless, the computation complexity of both \nameref{alg:mapla} and \Dikin{} is dominated by the Cholesky decomposition used to compute the difference in log determinants in the acceptance ratio.
We can also compare \MAMLA{} with mirror function \(\phi\) and \nameref{alg:mapla} with \(\metric{} = \hessmirror{}\).
The computational complexity of both these methods scale similar notwithstanding the cost of \(\dualtoprim{}\) in \MAMLA{}.
However, as highlighted previously in \cref{sec:back-rel-work}, it might not be possible to compute \(\dualtoprim{}\) in the general setting, and hence the complexity of this operation is difficult to ascertain.
We believe that in scenarios where a closed-form expression for \(\dualtoprim{}\) does not exist, this would be the computationally dominating step in \MAMLA{}.

\section{Mixing time guarantees for \nameref{alg:mapla}}
\label{sec:mixing-guarantees}
In this section, we state our main theorems that provide upper bounds on the mixing time for \nameref{alg:mapla} under certain sufficient conditions as previously alluded to in \cref{tab:mixing-time-scalings}.
We begin with certain preliminaries that include the definitions of the conditions that we assume on the metric \(\metric{}\) and potential \(\potential{}\) to provide mixing time guarantees.

\subsection{Preliminaries}
\label{sec:mixing-guarantees:prelims}
We first present some additional notation that will be used henceforth in this work.

\textbf{Additional Notation~~}
Let \(B\) be a \(d\times d\) matrix.
The operator and Frobenius norms of \(B\) are denoted by \(\|B\|_{\op}\) and \(\|B\|_{\frob}\) respectively.
For a smooth map \(g\) and \(x\) in its domain, the directional derivative of \(g\) at \(x\) in direction \(v\) is denoted by \(\rmD g(x)[v]\).
The (second-order) directional derivative of \(g\) at \(x\) in directions \(v, w\) is denoted by \(\rmD^{2} g(x)[v, w]\), which is equal to both \(\rmD(\rmD g(x)[v])[w]\) and \(\rmD(\rmD g(x)[w])[v]\) as \(g\) is smooth.
Given a set \(\calA\), the set of all of its measurable subsets is denoted by \(\calF(\calA)\), and the interior and boundary of \(\calA\) are denoted by \(\interior{\calA}\) and \(\partial\calA\) respectively.

\subsubsection{Classes of metrics}
\label{sec:prelims:precond}

In the rest of this paper, we assume the following regularity conditions about the metric: \(\metric{}\) is only defined on \(\interior{\primalspace}\) and becomes unbounded as it approaches the boundary i.e., \(\|\metric{x_{k}}\|_{\op} \to \infty\) for any sequence \(\{x_{k}\} \to \partial\primalspace\), and is twice differentiable.
The first two conditions ensures that the solution to the continuous time dynamics (\ref{eq:WLD}) stays within \(\primalspace\).
The examples of the metrics for various domains given in \cref{sec:algo} satisfy these regularity conditions.

\paragraph{Self-concordance}
This classical property is key in the analysis of interior points methods for constrained optimisation \citep{nesterov1994interior}, and quantifies the rate of change of a matrix-valued function in a certain sense as defined below.

\begin{definition}
\label{def:SC}
The metric \(\metric{} : \interior{\primalspace} \to \bbS_{+}^{d}\) is \emph{self-concordant} if for all \(x \in \interior{\primalspace}\) and \(v \in \bbR^{d}\)
\begin{equation*}
    |\rmD \metric{x}[v, v, v]| \leq 2 \cdot \|v\|_{\metric{x}}^{3}~.
\end{equation*}
\end{definition}

The design of \Dikin{} and \MAPLA{} implicitly assume the invertibility of the metric \(\metric{}\).
Notably, when \(\metric{}\) is self-concordant and its domain \(\interior{\primalspace}\) contains no straight lines, then \(\metric{}\) is always invertible; we give a proof of this assertion in \cref{sec:app:conseq-SC}.
The definition above is equivalent to the following \citep[Corr. 5.1.1]{nesterov2018lectures}.
\begin{equation*}
    \forall~x \in \interior{\primalspace},~v \in \bbR^{d}, \quad \|\metric{x}^{-\nicefrac{1}{2}}\rmD\metric{x}[v]\metric{x}^{-\nicefrac{1}{2}}\|_{\op} \leq 2 \cdot \|v\|_{\metric{x}}~.
\end{equation*}

\paragraph{Strong self-concordance}
This property \citep{laddha2020strong} and replaces the operator norm in the equivalent characterisation of self-concordance by the Frobenius norm.
This is a stronger notion than self-concordance due to the fact that \(\|A\|_{\op} \leq \|A\|_{\frob}\) for any matrix \(A\).

\begin{definition}
\label{def:SSC}
The metric \(\metric{} : \interior{\primalspace} \to \bbS_{+}^{d}\) is \emph{strongly self-concordant} if for all \(x \in \interior{\primalspace}\) and \(v \in \bbR^{d}\),
\begin{equation*}
    \|\metric{x}^{-\nicefrac{1}{2}}\rmD \metric{x}[v]\metric{x}^{-\nicefrac{1}{2}}\|_{\frob} \leq 2 \cdot \|v\|_{\metric{x}}~.
\end{equation*}
\end{definition}

\paragraph{Symmetry}
The role of the symmetrised set \(\primalspace \cap 2x - \primalspace\) for \(x \in \interior{\primalspace}\) was originally observed by \citet{gustafson2018john} in their study of the John walk, which was separately isolated by \citet{laddha2020strong} as property of metrics.
This property yields an isoperimetric inequality that results in mixing time bounds for several constrained sampling algorithms discussed previously.

For any \(x \in \interior{\primalspace}\), and \(r > 0\), the Dikin ellipsoid of radius \(r\) (denoted by \(\calE_{x}^{\metric{}}(r)\)) is defined as
\begin{equation}
\label{eq:dikin-ellipsoid}
    \calE_{x}^{\metric{}}(r) = \{y : \|y - x\|_{\metric{}} < r\}~.
\end{equation}

\begin{definition}
\label{def:symm}
The metric \(\metric{} : \interior{\primalspace} \to \bbS_{+}^{d}\) is said to be \emph{symmetric} with parameter \(\nu \geq 1\) if for any \(x \in \interior{\primalspace}\),
\begin{equation*}
    \calE_{x}^{\metric{}}(1) \subseteq \primalspace \cap (2x - \primalspace) \subseteq \calE_{x}^{\metric{}}(\sqrt{\nu})~.
\end{equation*}
\end{definition}

\paragraph{Lower trace and average self-concordance}

The properties were recently proposed in \citet{kook2024gaussian} and are abstractions from prior analyses of \Dikin{} \citep{sachdeva2016mixing,narayanan2017efficient}.
Specifically, these analyses were catered to the setting where \(\primalspace\) is a polytope and \(\metric{} = \hessmirror{}\) for \(\mirrorfunc{}\) being the log-barrier function of the polytope.
The salient features of these analyses that were abstracted by \citet{kook2024gaussian} for general metrics are: (1) a lower bound on the curvature of the function \(x \mapsto \log \det \metric{x}\), and (2) an upper bound on the likelihood of \(\|x - z\|_{\metric{x}}^{2} - \|x - z\|_{\metric{z}}^{2}\) being large for a Dikin proposal \(z\) from \(x\).
Lower trace self-concordance (along with strong self-concordance) yields the first property, and average self-concordance yields the second property, which are defined below.

\begin{definition}
\label{def:LTSC}
The metric \(\metric{} : \interior{\primalspace} \to \bbS_{+}^{d}\) is said to be \emph{lower trace self-concordant} with parameter \(\alpha \geq 0\) if for all \(x \in \interior{\primalspace}\) and \(v \in \bbR^{d}\),
\begin{equation*}
    \trace(\metric{x}^{-1}\rmD^{2}\metric{x}[v, v]) \geq -\alpha \cdot \|v\|_{\metric{x}}^{2}~.
\end{equation*}
\end{definition}

\begin{definition}
\label{def:ASC}
The metric \(\metric{} : \interior{\primalspace} \to \bbS_{+}^{d}\) is said to be \emph{average self-concordant} if for any \(x \in \interior{\primalspace}\) and \(\varepsilon > 0\), there exists \(r_{\varepsilon} > 0\) such that for any \(h \in (0, \frac{r_{\varepsilon}^{2}}{2d}]\),
\begin{equation*}
    \bbP_{\xi \sim \calN(x, 2h \cdot \metric{x}^{-1})}\left(\|\xi - x\|_{\metric{\xi}}^{2} - \|\xi - x\|_{\metric{x}}^{2} \leq 4h \cdot \varepsilon \right) \geq 1 - \varepsilon
\end{equation*}
\end{definition}

When the metric \(\metric{}\) satisfies strong, lower-trace and average self-concordance, it is said to satisfy \emph{self-concordant\textsubscript{++}} as stated in \cref{tab:mixing-time-scalings}.
The examples of metrics for various domains discussed previously in \cref{sec:algo:examples} are self-concordant\textsubscript{++} and symmetric.
A more comprehensive discussion about these properties are given in \citet[Sec. 3.4]{kook2024gaussian}.

\subsubsection{Function classes}
\label{sec:prelims:func}

\paragraph{Curvature lower and upper-boundedness}
This is a generalisation of the standard second-order definitions of smoothness and convexity where \(\metric{} = \rmI_{d \times d}\).
These are related to relative convexity and smoothness \citep{bauschke2017descent,lu2018relatively}.
Specifically, if \(\potential{}\) is \(\lambda\)-relatively smooth (or \(\mu\)-relatively convex) with respect to \(\psi\), then \(\potential{}\) satisfies a \((\lambda, \nabla^{2}\psi)\)-curvature upper bound (or \((\mu, \nabla^{2}\psi)\)-curvature lower bound) respectively.
When the metric \(\metric{}\) is self-concordant and the potential \(\potential{}\) satisfies \((\mu, \metric{})\)-curvature lower bound for \(\mu > 0\), \(\targetdist\) satisfies an isoperimetric inequality which is also key for our mixing time guarantees.

\begin{definition}
\label{def:curv-lower-bdd}
Given a metric \(\metric{} : \interior{\primalspace} \to \bbS_{+}^{d}\), the potential \(\potential{} : \interior{\primalspace} \to \bbR\) satisfies a \emph{\((\mu, \metric{})\)-curvature lower bound} where \(\mu \geq 0\) if for any \(x \in \interior{\primalspace}\),
\begin{equation*}
    \hesspotential{x} \succeq \mu \cdot \metric{x}~.
\end{equation*}
\end{definition}

\begin{definition}
\label{def:curv-upper-bdd}
Given a metric \(\metric{} : \interior{\primalspace} \to \bbS_{+}^{d}\), the potential \(\potential{} : \interior{\primalspace} \to \bbR\) satisfies a \emph{\((\lambda, \metric{})\)-curvature upper bound} where \(\lambda \geq 0\) if for any \(x \in \interior{\primalspace}\),
\begin{equation*}
    \hesspotential{x} \preceq \lambda \cdot \metric{x}~.
\end{equation*}
\end{definition}

\paragraph{Gradient upper bound}
This class of functions can be viewed as an extension of the standard notion of Lipschitz continuity for differentiable functions to take into account the metric.
This property has been useful in the analysis of algorithms based on the mirror Langevin dynamics \citep{ahn2021efficient,srinivasan2024fast} which is equivalent to setting \(\metric{} = \hessmirror{}\).

\begin{definition}
\label{def:rel-lips}
Given a metric \(\metric{} : \interior{\primalspace} \to \bbS_{+}^{d}\), the potential \(\potential{} : \interior{\primalspace} \to \bbR\) satisfies a \emph{\((\beta, \metric{})\)-gradient upper bound} where \(\beta \geq 0\) if for any \(x \in \interior{\primalspace}\),
\begin{equation*}
    \|\gradpotential{x}\|_{\metric{x}^{-1}} \leq \beta~.
\end{equation*}
\end{definition}

\subsubsection{Conductance and mixing time}
\label{sec:prelims:mixing}

Let \(s \in (0, \nicefrac{1}{2})\).
The \(s\)-conductance of a Markov chain \(\bfP = \{\calP_{x} : x \in \primalspace\}\) with stationary distribution \(\nu\) supported on \(\primalspace\) is defined as
\begin{equation*}
    \Phi_{\bfP}^{s} = \inf_{\substack{A \in \calF(\primalspace) \\
    \nu(A) \in (s, 1 - s)}} \frac{1}{\min\{\nu(A) - s,~1 - \nu(A) - s\}} \cdot \int_{A} \calP_{x}(\primalspace \setminus A) \cdot \rmd\nu(x)~.
\end{equation*}
The (ordinary) conductance \(\Phi_{\bfP}\) of \(\bfP\) is the limit of \(\Phi_{\bfP}^{s}\) as \(s \to 0\).

We use the total variation (TV) distance to quantify the mixing time of the Markov chain to its stationary distribution, which between two distribution \(\mu, \nu\) with support \(\primalspace\) is defined as
\begin{equation*}
    \TVdist(\mu, \nu) = \sup_{A \in \calF(\primalspace)} \mu(A) - \nu(A)~.
\end{equation*}

To obtain mixing time guarantees for \nameref{alg:mapla}, we use a strategy which assumes that the initial distribution \(\Pi_{0}\) is \emph{warm} relative to the stationary measure \(\targetdist\).
Two notions of warmness that are relevant to this technique are listed below \citep[Sec. 3]{vempala2005geometric}.
\begin{itemize}[leftmargin=*, itemsep=0pt]
    \item a distribution \(\mu_{0}\) is \((L_{\infty}, M)\)-warm w.r.t. \(\nu\) if \(\sup\limits_{A \in \calF(\primalspace)} \frac{\mu_{0}(A)}{\nu(A)} \leq M\).

    The set of all such \(\mu_{0}\) is denoted by \(\mathsf{Warm}(L_{\infty}, M, \nu)\).
    \item a distribution \(\mu_{0}\) is \((L_{1}, M)\)-warm w.r.t. \(\nu\) if \(\left\|\frac{\mu_{0}}{\nu}\right\|_{L^{1}(\mu_{0})} = \bbE_{\mu_{0}}\left[\frac{d\mu_{0}(x)}{d\nu(x)}\right] = M\).

    The set of all such \(\mu_{0}\) is denoted by \(\mathsf{Warm}(L_{1}, M, \nu)\).
\end{itemize}

For \(\delta \in (0, 1)\), the \(\delta\)-mixing time in TV distance of a Markov chain \(\bfP\) with stationary distribution \(\nu\) starting from a distribution \(\mu_{0}\) is the least number of applications of the operator \(\bbT_{\bfP}\) to \(\mu_{0}\) that achieves a distribution that is at most \(\delta\) away from \(\nu\).
For a class of distributions \(\calC\), we have
\begin{equation*}
    \mixingtime{\delta; \bfP, \calC} \defeq \sup_{\pi_{0} \in \calC} \inf \{k \geq 0 : \TVdist(\bbT_{\bfP}^{k}\pi_{0}, \nu) \leq \delta\}~.
\end{equation*}

\subsection{Main results}

Now, we state our main mixing time guarantees for \nameref{alg:mapla}.
For clarity, we state the assumptions made on the potential \(\potential{}\) of \(\targetdist\) separately below.

\begin{assumplist}
\item \label{assump:curv-lower-bdd} \(\potential{}\) satisfies \((\mu, \metric{})\)-curvature lower bound (\Cref{def:curv-lower-bdd}) and let \(\widetilde{\mu} = \frac{\mu}{8 + 4\sqrt{\mu}}\).
\item \label{assump:curv-upper-bdd} \(\potential{}\) satisfies \((\lambda, \metric{})\)-curvature upper bound (\Cref{def:curv-upper-bdd}).
\item \label{assump:grad-upper-bdd} \(\potential{}\) is \(\beta\)-relatively Lipschitz continuous w.r.t. \(\metric{}\) (\Cref{def:rel-lips}).
\end{assumplist}

\paragraph{Warmup: working with self-concordance of \(\metric{}\)}
Here, we assume that the metric \(\metric{}\) satisfies self-concordance (\cref{def:SC}), which is the most basic notion of self-concordance defined previously. 

\begin{theorem}
\label{thm:just-SC-mapcla}
Consider a distribution \(\targetdist\) supported over \(\primalspace\) that is a closed, convex subset of \(\bbR^{d}\) whose density is \(\targetdens(x) \propto e^{-\potential{x}}\).
Let the metric \(\metric{} : \interior{\primalspace} \to \bbS_{+}^{d}\) be self-concordant and \(\nu\)-symmetric, and assume that the potential \(\potential{} : \interior{\primalspace} \to \bbR\) satisfies \emph{\ref{assump:curv-lower-bdd}, \ref{assump:curv-upper-bdd}} and \emph{\ref{assump:grad-upper-bdd}}.
Define the quantity \(b_{\textsf{SC}}(d, \lambda, \beta)\)
\begin{equation*}
    b_{\textsf{SC}}(d, \lambda, \beta) \defeq c_{1} \cdot \min\left\{\frac{1}{d^{3}},~\frac{1}{d \cdot \lambda},~\frac{1}{\beta^{2}}~,\frac{1}{\beta^{\nicefrac{2}{3}}}~,\frac{1}{(\beta \cdot \lambda)^{\nicefrac{2}{3}}}\right\}~.
\end{equation*}
For precision \(\delta \in (0, \nicefrac{1}{2})\) and warmness parameter \(M \geq 1\), if the step size \(h\) is bounded as \(0 < h \leq b_{\textsf{SC}}(d, \lambda, \beta)\), \nameref{alg:mapla} satisfies for \(\calC \in \{\mathsf{Warm}(L_{\infty}, M, \targetdist), \mathsf{Warm}(L_{1}, M, \targetdist)\}\) that
\begin{equation*}
    \mixingtime{\delta; \bfT, \calC} = \frac{c_{2}}{h} \cdot \max\left\{1, \min\left\{\frac{1}{\widetilde{\mu}^{2}}, \nu\right\}\right\} \cdot \log\left(\frac{\frakM_{\calC}}{\delta}\right)~, \quad
    \frakM_{\calC} = \begin{cases}
        M^{\nicefrac{1}{2}}, & \calC = \mathsf{Warm}(L_{\infty}, M, \targetdist) \\
        M^{\nicefrac{1}{3}}, & \calC = \mathsf{Warm}(L_{1}, M, \targetdist)
    \end{cases}
\end{equation*}
where \(c_{1}, c_{2}\) are universal positive constants.
\end{theorem}

\paragraph{Beyond standard self-concordance of \(\metric{}\)}
Now, we assume that the metric \(\metric{}\) satisfies \emph{self-concordance\textsubscript{++}}, which is a combination of the stronger notions of self-concordance (\cref{def:SSC,def:LTSC,def:ASC}) as described previously.
Self-concordance\textsubscript{++} enables a larger bound on the step size \(h\) than \(b_{1}\) in \cref{thm:just-SC-mapcla}, which consequently yields better mixing time guarantees.

\begin{theorem}
\label{thm:more-than-SC-mapcla}
Consider a distribution \(\targetdist\) supported over \(\primalspace\) that is a closed, convex subset of \(\bbR^{d}\) whose density is \(\targetdens(x) \propto e^{-\potential{x}}\).
Let the metric \(\metric{} : \interior{\primalspace} \to \bbS_{+}^{d}\) be strongly, \(\alpha\)-lower trace, and average self-concordant and \(\nu\)-symmetric, and assume that the potential \(\potential{} : \interior{\primalspace} \to \bbR\) satisfies \emph{\ref{assump:curv-lower-bdd}, \ref{assump:curv-upper-bdd}} and \emph{\ref{assump:grad-upper-bdd}}.
Define the quantity \(b_{\textsf{SC\textsubscript{++}}}(d, \lambda, \alpha, \beta)\)
\begin{equation*}
    b_{\textsf{SC\textsubscript{++}}}(d, \lambda, \beta, \alpha) \defeq c_{1} \cdot \min\left\{\frac{1}{d \cdot \beta},~\frac{1}{d \cdot \lambda},~\frac{1}{d \cdot (\alpha + 4)},~\frac{1}{\beta^{2}},~\frac{1}{(\beta \cdot (\alpha + 4))^{\nicefrac{2}{3}}},~\frac{1}{(\beta \cdot \lambda)^{\nicefrac{2}{3}}}\right\}~.
\end{equation*}
For precision \(\delta \in (0, \nicefrac{1}{2})\) and warmness parameter \(M \geq 1\), if the step size \(h\) is bounded as \(0 < h \leq b_{\textsf{SC\textsubscript{++}}}(d, \lambda, \beta, \alpha)\), \nameref{alg:mapla} satisfies for \(\calC \in \{\mathsf{Warm}(L_{\infty}, M, \targetdist), \mathsf{Warm}(L_{1}, M, \targetdist)\}\) that
\begin{equation*}
    \mixingtime{\delta; \bfT, \calC} = \frac{c_{2}}{h} \cdot \max\left\{1, \min\left\{\frac{1}{\widetilde{\mu}^{2}}, \nu\right\}\right\} \cdot \log\left(\frac{\frakM_{\calC}}{\delta}\right)~, \quad
    \frakM_{\calC} = \begin{cases}
        M^{\nicefrac{1}{2}}, & \calC = \mathsf{Warm}(L_{\infty}, M, \targetdist) \\
        M^{\nicefrac{1}{3}}, & \calC = \mathsf{Warm}(L_{1}, M, \targetdist)
    \end{cases}
\end{equation*}
where \(c_{1}, c_{2}\) are universal positive constants.
\end{theorem}

\paragraph{Handling linear \(\potential{}\)}
Here, we discuss the setting where \(\potential{x} = \left.\sigma^{\top}x\right|_{\primalspace}\) for \(\sigma \neq \bm{0}\).
Recall that both \Cref{thm:just-SC-mapcla,thm:more-than-SC-mapcla} assume that the potential \(\potential{}\) satisfies \((\beta, \metric{})\)-gradient upper bound (\ref{assump:grad-upper-bdd}), and hence directly invoking these theorems for this setting would result in a dependence on
\begin{equation*}
\beta(\sigma) \defeq \sup\limits_{x \in \interior{\primalspace}} \|\sigma\|_{\metric{x}^{-1}}~.
\end{equation*}
This uncovers two issues.
First, for any \(\frakc \in \bbR\), \(\beta(\frakc \cdot \sigma) = |\frakc| \cdot \beta(\sigma)\), which would imply that the scale of \(\sigma\) affects the mixing time guarantee.
Second, suppose \(\sigma\) is normalised (i.e., \(\|\sigma\| = 1\)) and \(\metric{}\) is self-concordant.
The quantity \(\beta(\sigma)\) could still depend on the size of \(\primalspace\) since
\begin{equation*}
    \sup_{x \in \interior{\primalspace}} \|\sigma\|_{\metric{x}^{-1}} = \sup_{\substack{x \in \interior{\primalspace} \\ \|v\|_{\metric{x}} \leq 1}} \langle \sigma, v\rangle = \sup_{\substack{x \in \interior{\primalspace} \\ y \in \calE_{x}^{\metric{}}(1)}} \langle \sigma, y - x\rangle \leq \sup_{x, y \in \interior{\primalspace}}\|y - x\|~.
\end{equation*}
Hence, it is crucial that the mixing time guarantees in this setting is independent of \(\beta(\sigma)\).
In the following theorem, we derive such a \emph{scale-independent} guarantee for \nameref{alg:mapla} which uses properties of densities whose potential \(\potential{}\) is linear \citep{kook2023condition}.

\begin{theorem}
\label{thm:exp-densities-mapcla}
Consider a distribution \(\targetdist\) supported over \(\primalspace\) that is a closed, convex subset of \(\bbR^{d}\) whose density is \(\targetdens(x) \propto e^{-\sigma^{\top}x}\).
Let the metric \(\metric{} : \interior{\primalspace} \to \bbS_{+}^{d}\) be strongly and average self-concordant and \(\nu\)-symmetric, and assume that it also satisfies \(\rmD^{2}\metric{x}[v, v] \succeq \bm{0}\) for all \(x \in \interior{\primalspace}\) and \(v \in \bbR^{d}\).
Define the quantity \(b_{\textsf{Exp}}(d, M, \delta)\)
\begin{equation*}
    b_{\textsf{Exp}}(d, M, \delta) \defeq c_{1} \cdot \frac{1}{d^{2} \log^{2}(\frac{M}{\delta})}~.
\end{equation*}
For precision \(\delta \in (0, \nicefrac{1}{2})\) and warmness parameter \(M \geq 1\), if the step size \(h\) is bounded as \(0 < h \leq b_{\textsf{Exp}}(d, M, \delta)\), \nameref{alg:mapla} satisfies
\begin{equation*}
    \mixingtime{\delta; \bfT, \mathsf{Warm}(L_{\infty}, M, \targetdist)} = \frac{c_{2}}{h} \cdot \max\left\{1, \nu\right\} \cdot \log\left(\frac{M}{\delta}\right)
\end{equation*}
where \(c_{1}, c_{2}\) are universal positive constants.
\end{theorem}

\subsubsection{A discussion of the results}
\label{sec:mixing-guarantees:discussion}

The underlying technique used to prove \cref{thm:just-SC-mapcla,thm:more-than-SC-mapcla,thm:exp-densities-mapcla} is due to is due to \citet{lovasz1999hit}.
Given a Markov chain \(\bfT\), \citeauthor{lovasz1999hit}'s technique involves showing that for any two points \(x, y\) close enough in a sufficiently large subset of \(\primalspace\), the TV distance between the one-step distributions \(\calT_{x}, \calT_{y}\) is uniformly bounded away from \(1\), and is hence referred to as the \emph{one-step overlap} technique.
Given this one-step overlap, we rely on isoperimetric inequalities which lead to a lower bound on the \(s\)-conductance / conductance of \(\bfT\), which results in mixing time guarantees by the classical result of \citet{lovasz1993random}.
The \(\nu\)-symmetry of the metric \(\metric{}\) results in an isoperimetric inequality for log-concave distributions \citep{laddha2020strong}, which can be complemented by another isoperimetric inequality when the potential satisfies a \((\mu, \metric{})\)-curvature lower bound (\ref{assump:curv-lower-bdd}) for a self-concordant \(\metric{}\).
One of our contributions is deriving the latter isoperimetric inequality, which is a generalisation of prior results by \citet[Lem. 7]{gopi2023algorithmic}.
Other conditions placed on \(\metric{}\) and the potential \(\potential{}\) namely self-concordance or self-concordance\textsubscript{++}, and the \((\lambda, \metric{})\)-curvature upper bound (\ref{assump:curv-upper-bdd}) and \((\beta, \metric{})\)-gradient upper bound (\ref{assump:grad-upper-bdd}) conditions on \(\potential{}\) yield bounds on the step size which ensures that the expected acceptance rate is away from \(0\) as \(d\) increases, which is related to the one-step overlap.
Our analysis reveals that \ref{assump:grad-upper-bdd} is not necessary to derive mixing time guarantees for \nameref{alg:mapla}, and also extends to the analysis of \MAMLA{} \citep{srinivasan2024fast} where \ref{assump:grad-upper-bdd} is also considered for the potential \(\potential{}\) following \citet{ahn2021efficient}.
More precisely, we find that it is sufficient if the function \(x \mapsto \|\gradpotential{x}\|_{\metric{x}^{-1}}\) is uniformly bounded in a sufficiently large convex subset of \(\primalspace\).
This weaker sufficient condition is satisfied in two cases which our theorems cover: (1) when \(\potential{}\) satisfies the \((\beta, \metric{})\)-gradient upper bound condition (as assumed in \cref{thm:just-SC-mapcla,thm:more-than-SC-mapcla}), and (2) when \(\potential{}\) is linear (\cref{thm:exp-densities-mapcla}).
Finally, the convexity-style assumption on the metric \(\metric{}\) in \cref{thm:exp-densities-mapcla} implies that \(\metric{}\) is \(0\)-lower trace self-concordance, which is used to guarantee that the large enough subset identified for the weaker sufficient condition described above is convex.

While we use self-concordance\textsubscript{++} in \cref{thm:more-than-SC-mapcla}, the proofs to establish mixing time guarantees for \Dikin{} in \citet{kook2024gaussian} and \nameref{alg:mapla} in this work differ in the details, specifically in how the one-step overlap is established.
In the analysis of \Dikin{} presented by \citet{kook2024gaussian}, they work with an exact analytical expression for the TV distance between the one-step distributions induced by an iteration of \Dikin{}.
More precisely, for a Markov chain \(\bfT\) formed by the Metropolis adjustment of \(\bfP\) w.r.t. \(\targetdist\) and \(p_{\mathrm{accept}}\) as defined in \cref{eq:MH-accept-def}, we have
\begin{align*}
    \TVdist(\calT_{x}, \calT_{y}) &= \frac{1}{2}(r_{x} + r_{y}) + \frac{1}{2}\int_{z \in \primalspace} \left|p_{\mathrm{accept}}(z; x)\cdot p_{x}(z) -  p_{\mathrm{accept}}(z; y)\cdot p_{y}(z)\right| \rmd z~, \\
    r_{x} &= 1 - \bbE_{z \sim \calP_{x}}[p_{\mathrm{accept}}(z; x) \cdot \bm{1}\{z \in \primalspace\}]~.
\end{align*}
\citet{kook2024gaussian} remark that working with this is essential to analyse \Dikin{}, and use sophisticated techniques to obtain a bound for the above quantity where \(\calP_{x} = \calN(x, 2h \cdot \metric{x}^{-1})\) and \(\bfP = \{\calP_{x} : x \in \primalspace\}\).
Due to the drift correction in \nameref{alg:mapla}, we are able to take a relatively simpler approach that involves giving bounds on \(\TVdist(\calT_{x}, \calP_{x})\) and \(\TVdist(\calP_{x}, \calP_{y})\) and noting that
\begin{equation*}
    \TVdist(\calT_{x}, \calT_{y}) \leq \TVdist(\calT_{x}, \calP_{x}) + \TVdist(\calP_{x}, \calP_{y}) + \TVdist(\calT_{y}, \calP_{y})~.
\end{equation*}
This simpler approach is known to yield a vacuous bound \(\TVdist(\calT_{x}, \calT_{y})\) for the analysis of \Dikin{}; however this is not the case for the proposal Markov chain in \nameref{alg:mapla} due to the inclusion of the drift correction \(-h \cdot \metric{}^{-1}\gradpotential{}\) to the proposal distribution of \Dikin{} that \nameref{alg:mapla} is based on, and results in a crisper and less complicated proof for \nameref{alg:mapla}.

\section{Numerical Experiments}
\label{sec:expts}
In this section, we discuss numerical experiments\footnote{Code for these experiments can be found at \url{https://github.com/vishwakftw/conspacesampler}.} performed to evaluate \nameref{alg:mapla} on a collection of problems.
Through these experiments, we aim to provide a comparison \nameref{alg:mapla} to \Dikin{} for these problems, 
and also verify our theoretical guarantees in \cref{sec:mixing-guarantees}.
The problems that we consider are approximate sampling from (1) Dirichlet distributions, and (2) posterior distributions arising from Bayesian logistic regression with constrained priors.

\subsection{Sampling from Dirichlet distributions}

\newcommand{\empwd}[1]{\widetilde{W_{2}^{2}}(\widehat{\bbT}^{k}_{\textsf{#1}}, \widehat{\targetdist})}
\newcommand{\emped}[1]{\mathrm{ED}(\widehat{\bbT}^{k}_{\textsf{#1}}, \widehat{\targetdist})}

The Dirichlet distribution is the multi-dimensional form of the Beta distribution which is supported on the simplex \(\Delta_{d + 1}\) defined as \(\Delta_{d + 1} \defeq \{x \in \bbR_{+}^{d} : \bm{1}^{\top}x \leq 1\}\).
This representation treats the simplex as a convex subset of \(\bbR^{d}\) with a non-empty interior and is equivalent to the canonical definition.
The Dirichlet distribution is parameterised by a concentration parameter \(\bm{a} \in \bbR^{d + 1}\) where \(a_{i} > -1\), and its density \(\targetdens\) is defined as follows
\begin{equation*}
    \pi(x) \propto \exp(-\potential{x}); \quad \potential{x} = -\sum_{i=1}^{d} a_{i}\log x_{i} - a_{d + 1}\log\left(1 - \sum_{i=1}^{d}x_{i}\right)~.
\end{equation*}
The Dirichlet distribution is log-concave when \(a_{i} \geq 0\) for all \(i \in [d + 1]\).
The experiments in this subsection are conducted with the metric \(\metric{}\) given by the Hessian of the log-barrier of \(\Delta_{d + 1}\).
This metric is natural for the Dirichlet potential \(\potential{}\) as \(\potential{}\) satisfies \((\bm{a}_{\min}, \metric{})\)-curvature lower bound, \((\bm{a}_{\max}, \metric{})\)-curvature upper bound, and \((\|\bm{a}\|, \metric{})\)-gradient upper bound for this metric.

Given \(N\) initial points drawn independently from an initial distribution, both algorithms return a collection of \(N\) independent samples after every iteration.
For \(\textsf{alg} \in \{\MAPLA{}, \Dikin{}\}\), let \(\widehat{\bbT}^{k}_{\textsf{alg}}\) be the empirical distribution of samples obtained after running \textsf{alg} for \(k\) iterations.
Let \(\widehat{\targetdist}\) be the empirical distribution formed by \(N\) samples from a Dirichlet distribution; several scientific computing packages provide functionality to obtain this.
We use two measures to assess the similarity between \(\widehat{\bbT}^{k}_{\textsf{alg}}\) and \(\widehat{\targetdist}\).
The first is the empirical 2-Wasserstein distance denoted by \(\empwd{alg}\), which does not have a closed form and is computed using the Sinkhorn-Knopp solver \citep{cuturi2013sinkhorn} with regularisation \(0.001\).
The second measure is the (empirical) energy distance \citep{szekely2013energy} which is based on E-statistics, and is defined as
\begin{equation*}
    \emped{alg} = \frac{2}{N^{2}}\sum_{i,j=1}^{N}\|X_{i} - Y_{j}\| - \frac{1}{N^{2}}\sum_{i,j=1}^{N}\|X_{i} - X_{j}\| - \frac{1}{N^{2}}\sum_{i,j=1}^{N}\|Y_{i} - Y_{j}\|
\end{equation*}
where \(\{X_{i}\}_{i=1}^{N}\) and \(\{Y_{i}\}_{i=1}^{N}\) are the supports of \(\widehat{\bbT}^{k}_{\textsf{alg}}\) and \(\widehat{\Pi}\) respectively.

\subsubsection{Comparison of mixing behaviours of \MAPLA{} and \Dikin{}}

Here, we aim to provide a comparison between the mixing behaviours of \nameref{alg:mapla} and \Dikin{} to the true Dirichlet distribution.
For a threshold \(\delta > 0\) and measure \(\textsf{dist} \in \{\widetilde{W_{2}^{2}}, \mathrm{ED}\}\), define the empirical mixing time as \(
    \widehat{\tau}_{\mathrm{mix}}^{\textsf{alg}}(\delta; \textsf{dist}) = \inf \{k \geq 1 : \textsf{dist}(\widehat{\bbT}^{k}_{\textsf{alg}}, \widehat{\Pi}) \leq \delta\}.\)
For a given \(d\), we perform \(20\) independent simulations with \(N = 2000\) initial points.
In each simulation, we set \(\bm{a}\) as 
\begin{equation*}
    a_{i} = \bm{a}_{\min} + \frac{(i - 1)}{d} \cdot (\bm{a}_{\max} - \bm{a}_{\min}) \text{ for } i \in [d + 1]~,\quad
    \bm{a}_{\min} = 1, \bm{a}_{\max} = 3~,
\end{equation*}
and run both methods with step size \(h = \frac{C_{h}}{\bm{a}_{\max} \cdot d}\) where \(C_{h} \in \{0.1, 0.2\}\).
In \cref{fig:mixing-time-dirichlet}, we plot the results of these simulations.

\begin{figure}[t]
    \centering
    \begin{subfigure}{0.48\linewidth}
        \centering
        \includegraphics[width=0.8\linewidth]{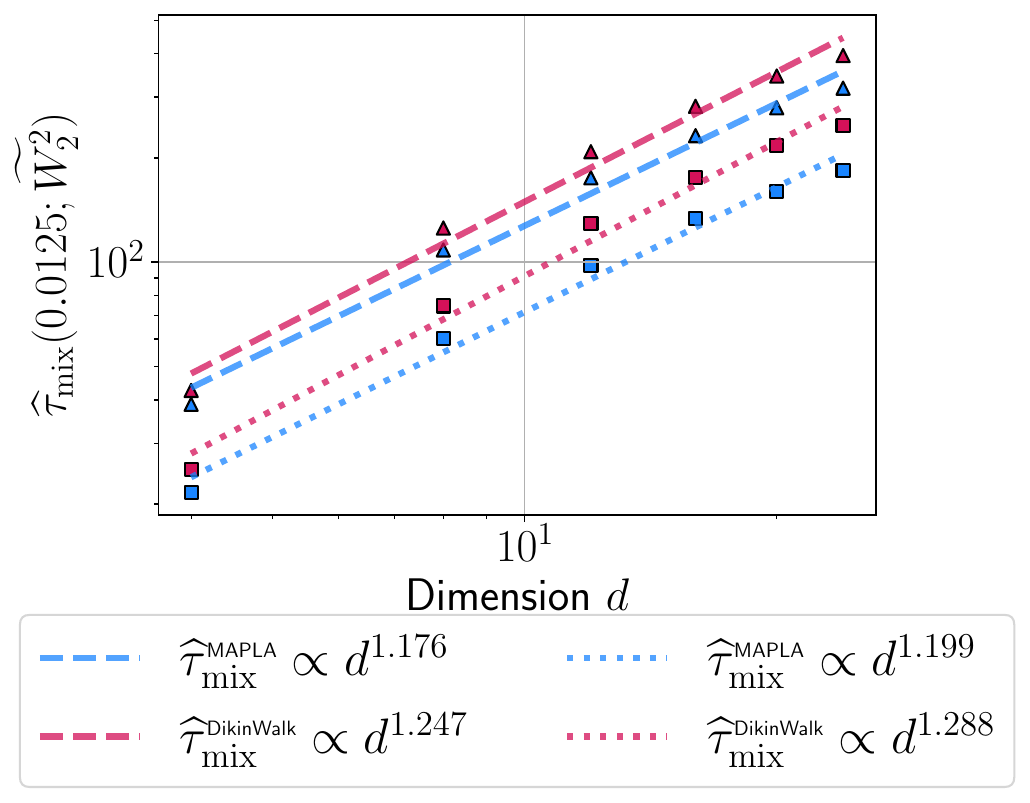}
    \end{subfigure}
    \hfill
    \begin{subfigure}{0.48\linewidth}
        \centering
        \includegraphics[width=0.8\linewidth]{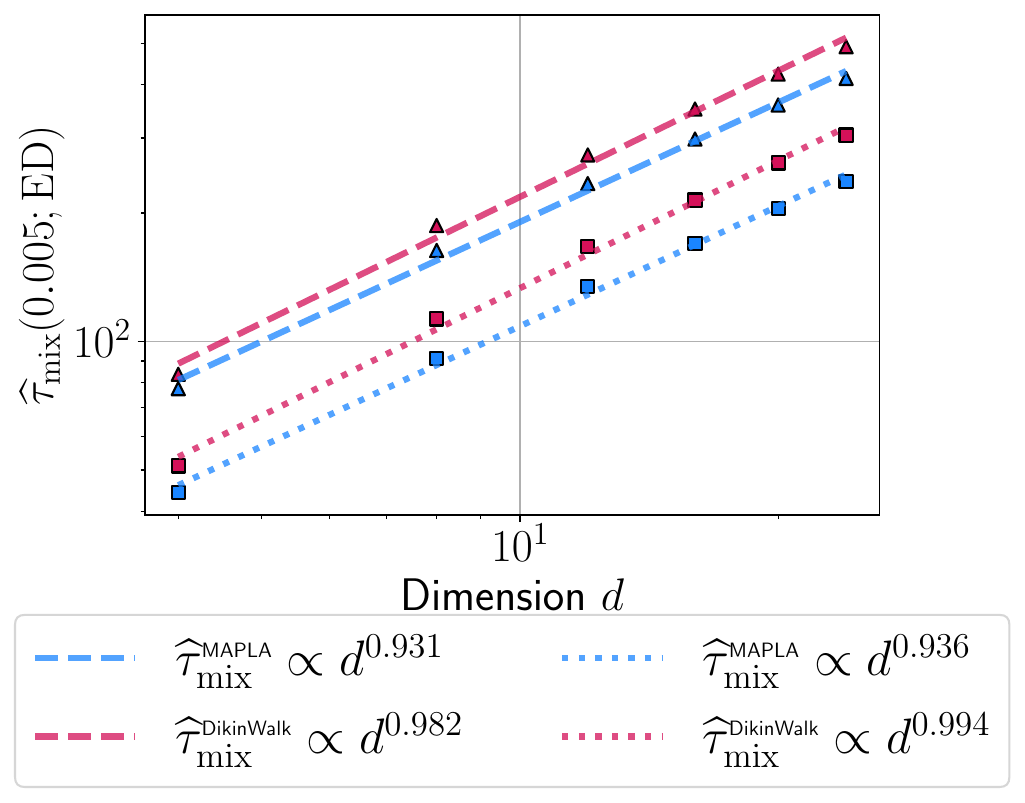}
    \end{subfigure}
    \caption{Variation of empirical mixing time, computed with \(\widetilde{W_{2}^{2}}\) (left) and \(\mathrm{ED}\) (right) for both \MAPLA{} and \Dikin{}.
    The dashed and dotted lines correspond to \(C_{h} = 0.1\) and \(0.2\) respectively.
    The ordinates of the markers indicate the average empirical mixing time over 20 simulations.}
    \label{fig:mixing-time-dirichlet}
\end{figure}

A few direct inferences can be made from \cref{fig:mixing-time-dirichlet}.
Irrespective of \(C_{h}\) and the measure \(\textsf{dist}\), the variation of the empirical mixing time with dimension is almost linear for both methods.
For both \nameref{alg:mapla} and \Dikin{}, we see that a larger \(C_{h}\) leads too a faster mixing, which is to be expected.
Most notably, we observe that for a fixed \(C_{h}\), \nameref{alg:mapla} mixes faster than \Dikin{}.
Intuitively, this is highly likely due to \nameref{alg:mapla} using more information about the potential \(\potential{}\) through its gradient than \Dikin{} in its proposal, which is a geometric random walk.

\begin{figure}[t]
    \centering
    \begin{subfigure}{0.49\linewidth}
        \centering
        \includegraphics[width=\linewidth]{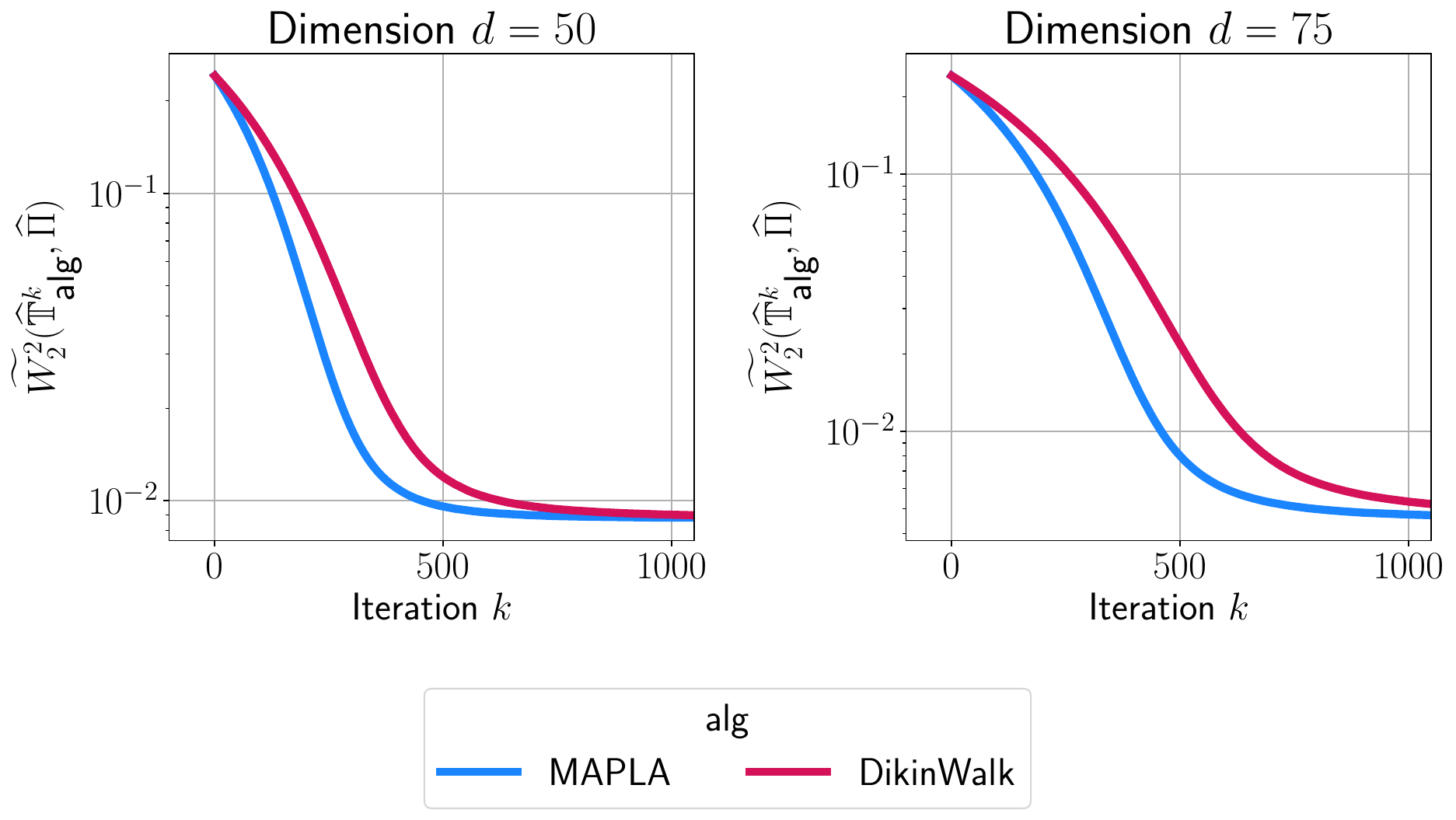}
    \end{subfigure}
    \hfill
    \begin{subfigure}{0.49\linewidth}
        \centering
        \includegraphics[width=\linewidth]{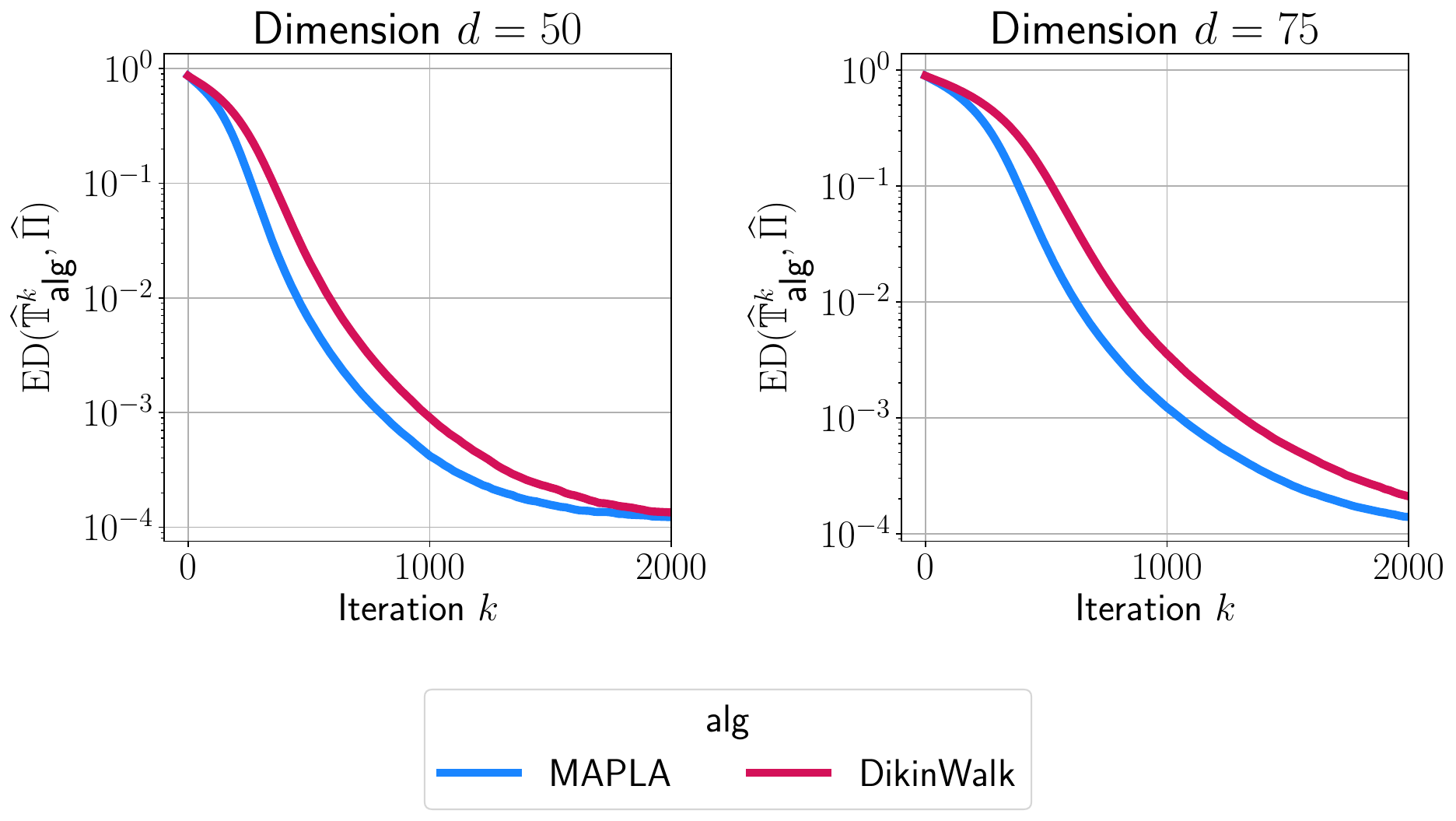}
    \end{subfigure}
    \caption{Variation of \(\textsf{dist}(\widehat{\bbT}^{k}_{\textsf{alg}}, \widehat{\Pi})\) for \(\textsf{dist} = \widetilde{W_{2}^{2}}\) (left) and \(\mathrm{ED}\) (right) with iteration \(k\).
    For the plots showing the variation of \(\empwd{}\), the \(x\)-axis is truncated to \(1000\) as the values converged.}
    \label{fig:dist-variation-d-50-75}
\end{figure}

To complement \cref{fig:mixing-time-dirichlet} which showcases the empirical mixing time, we also plot the variation in \(\textsf{dist}(\widehat{\bbT}^{k}_{\mathsf{alg}}, \widehat{\Pi})\) with the iteration \(k\) for larger values of \(d\) in \cref{fig:dist-variation-d-50-75}.
Here, we set the concentration parameter to \(\bm{2}_{d}\), and the step size \(h = \frac{1}{10 \bm{a}_{\max} \cdot d}\), and run both methods for \(2000\) iterations.
These plots reveal a discernible difference in the mixing behaviour of both methods, and this difference appears to widen as \(d\) increases.

\subsubsection{Step size scaling and accceptance rate}

The choice of metric \(\metric{}\) here satisfies self-concordance\textsubscript{++}, and therefore \cref{thm:more-than-SC-mapcla} suggests that the maximum step size for \nameref{alg:mapla} scales as \((d \cdot \|\bm{a}\|)^{-1}\).
This is due to the fact that \(\potential{}\) satisfies \((\|\bm{a}\|, \metric{})\)-gradient upper bound as stated previously.
When \(a_{i} = a\) for all \(i \in [d + 1]\), \(\|\bm{a}\| = a \cdot \sqrt{d + 1}\), and this leads to the maximum step size scaling with dimension \(d\) as \(d^{-\nicefrac{3}{2}}\).
We note that this bound on the step size is based on a \emph{uniform} bound on the local norm \(\|\gradpotential{}\|_{\metric{}^{-1}}\), and does not take into account other properties of the distribution which might be useful.
Hence, it would be useful to verify if this bound is indeed tight and if step sizes that scale better with dimension (say as \(d^{-\gamma}\) for \(\gamma < 1.5\)) yield non-vanishing acceptance rates.
This has implications for the mixing time guarantee of \nameref{alg:mapla}; in essence, if \(\bar{h}\) is the maximum step size that results in non-vanishing acceptance rates for \nameref{alg:mapla}, then its mixing time is proportional to \(\bar{h}^{-1}\), which in turn scales as \(d^{\gamma}\).

\begin{minipage}{0.5\linewidth}
Our diagnostic for this is the (empirical) average acceptance rate, which was previously used by \citet{dwivedi2018log} and \citet{srinivasan2024fast} to similarly verify guarantees for \textsf{MALA} and \MAMLA{} respectively.
For a step size \(h\) and algorithm \textsf{alg}, this (empirical) average acceptance rate \(\widehat{R}_{\mathrm{accept}}\) is defined as the average proportion of accepted proposals (out of \(N = 2000\) particles) over \(4500\) iterations of \textsf{alg} with step size \(h\) after a burn-in period of \(500\) iterations.
The concentration parameter \(\bm{a}\) is set to \(\bm{2}_{d}\).
For both \nameref{alg:mapla} and \Dikin{}, we set the step size as \(h = (10 \cdot d^{\gamma})^{-1}\) and \(\gamma \in \{0.75, 1, 1.5\}\).
We plot the average variation (over 10 independent simulations) of \(\widehat{R}_{\mathrm{accept}}\) with the dimension \(d\) in \cref{fig:accept-rate}.
\end{minipage}
\hfill
\begin{minipage}{0.47\linewidth}
\centering
\begin{figure}[H]
\raggedleft
\includegraphics[width=0.8\linewidth]{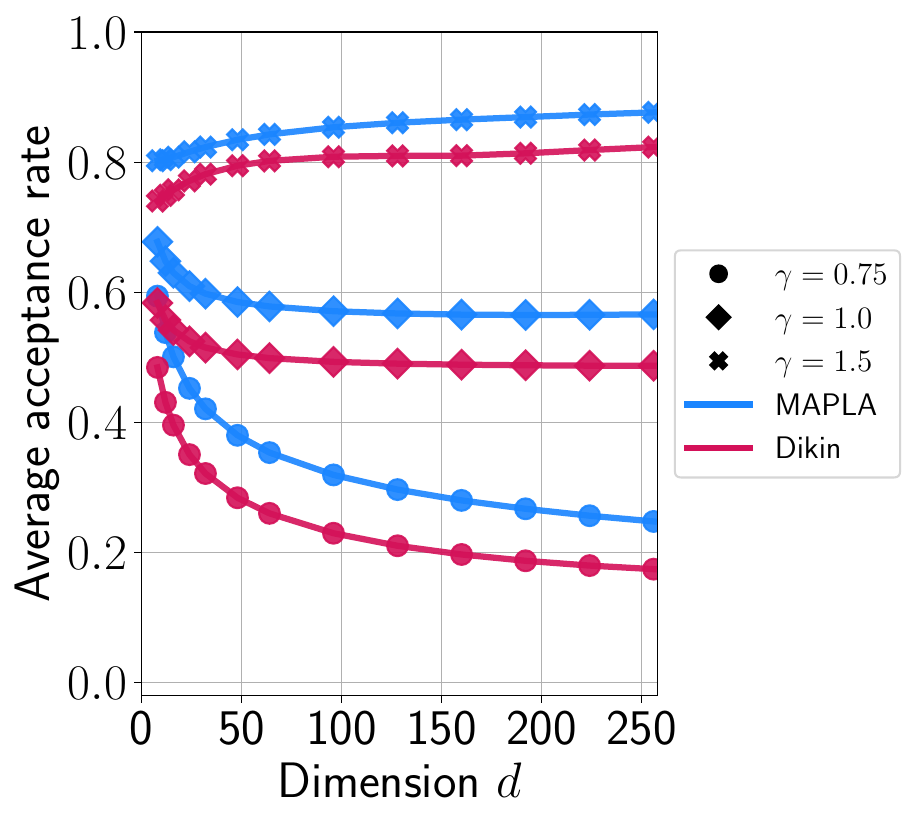}
\caption{Variation of \(\widehat{R}_{\mathrm{accept}}\) with dimension \(d\) when stepsize \(h \propto d^{-\gamma}\).}
\label{fig:accept-rate}
\end{figure}
\end{minipage}

From \cref{fig:accept-rate}, we observe similar trends for both \nameref{alg:mapla} and \Dikin{} for all values of \(\gamma\), and that the average acceptance rate is consistently higher for \nameref{alg:mapla}.
We see that \(\widehat{R}_{\mathrm{accept}}\) increases as the dimension \(d\) increases when \(\gamma = 1.5\), which suggests that this scaling is indeed moderately conservative.
More interestingly, we also observe that \(\gamma = 1.0\) yields non-vanishing acceptance rates like \Dikin{}, but when \(\gamma = 0.75\), we see that the average acceptance rate for both methods decrease steadily as the dimension \(d\) increases.

\subsection{Application to Bayesian logistic regression posteriors}

In logistic regression, we are given a collection of pairs \(D_{n} = \{(X^{(i)}, y^{(i)})\}_{i=1}^{n}\), where \(X^{(i)} \in \bbR^{d}\), and \(y^{(i)} \in \{0, 1\}\) for all \(i \in [n]\).
It is assumed that there exists a parameter space \(\Theta \subseteq \bbR^{d}\) and \(\theta^{\star} \in \Theta\) such that for each \(i \in [n]\), \(y^{(i)}\) is drawn according to a Bernoulli distribution with probability \(p_{i} = (1 + \exp(-\langle\theta^{\star}, X^{(i)}\rangle))^{-1}\).
The goal of (Bayesian) logistic regression is to obtain an estimate \(\widehat{\theta}_{n} \in \Theta\) of \(\theta^{\star}\) based on \(D_{n}\) alone, with high probability.
A natural approach is to obtain the maximum likelihood estimate of \(\theta^{\star}\), and this is consistent.
However, in low-data regimes, additional information about \(\theta^{\star}\) is beneficial as it imposes more structure on the problem.

In the Bayesian framework, this additional information is provided in the form of a prior distribution over \(\Theta\) with a density \(\nu(\theta)\).
This allows us to obtain the posterior density of \(\theta\) given \(D_{n}\), which is useful in making inferences about \(\theta^{\star}\) by drawing samples from the posterior due to the Bernstein von-Mises theorem when \(\nu(\theta^{\star}) > 0\).
The posterior density function \(\tilde{\nu}\) is defined as
\begin{equation}
\label{eq:potential-bayesian-logistic-regression}
    \tilde{\nu}(\theta) \propto \exp(-\potential{\theta; D_{n}} + \log \nu(\theta)); \quad \potential{\theta; D_{n}} = -\sum_{i = 1}^{n}\left(y^{(i)}\langle\theta, X^{(i)}\rangle - \log(1 + \exp(\langle\theta, X^{(i)}\rangle))\right).
\end{equation}

For the empirical study in this section, \(\Theta\) is a polytope formed by \(2d\) linear constraints.
More specifically, this is obtained performing by \(\lfloor \nicefrac{d}{2} \rfloor\) random rotations of the \(d\)-dimensional box \([-2, 2]^{d}\) followed by a fixed translation.
The prior \(\nu\) is a uniform prior over \(\Theta\) and hence, the posterior \(\tilde{\nu}\) is a constrained distribution over \(\Theta\).
For both \nameref{alg:mapla} and \Dikin{}, we choose the metric \(\metric{}\) as the Hessian of the log-barrier of \(\Theta\).
Let \(\overline{\metric{}}\) denote the Hessian of the log-barrier of \([-2, 2]^{d}\) given by
\begin{equation*}
    \overline{\metric{}}(\theta) = \sum_{i=1}^{d} \left(\frac{1}{(2 - \theta_{i})^{2}} + \frac{1}{(2 + \theta_{i})^{2}}\right)e_{i}e_{i}^{\top} \succeq \frac{1}{2}\rmI_{d \times d}~.
\end{equation*}
The metric \(\metric{}\) is related to \(\overline{\metric{}}\) as \(\metric{} = R^{\top}\overline{\metric{}}R\), where \(R\) is a product of \(\lfloor \nicefrac{d}{2}\rfloor\) Givens rotation matrices (which are orthonormal), and hence for any \(\theta \in \Theta\), \(\metric{\theta} \succeq \frac{1}{2}\rmI_{d \times d}\).
Let \(\lambda_{\max}\) be the maximum eigenvalue of \(\sum_{i=1}^{n} X^{(i)}{X^{(i)}}^{\top}\).
Since \(\frac{\exp(t)}{(1 + \exp(t))^{2}} \leq \frac{1}{4}\) for \(t \in \bbR\), we have
\begin{equation*}
    \hesspotential{\theta} = \sum_{i=1}^{n} \frac{\exp(\langle \theta, X^{(i)}\rangle)}{(1 + \exp(\langle \theta, X^{(i)}\rangle))} X^{(i)}{X^{(i)}}^{\top} \preceq \frac{\lambda_{\max}}{2} \cdot \frac{1}{2}\rmI_{d \times d} \preceq \frac{\lambda_{\max}}{2} \cdot \metric{\theta}~.
\end{equation*}

We initialise both algorithms with \(N\) initial points, and at each iteration both algorithms return a collection of \(N\) independent samples.
Let \(\{\theta^{\textsf{alg}}_{k, j}\}_{j=1}^{N}\) be the samples obtained at iteration \(k\) from algorithm \(\textsf{alg} \in \{\MAPLA{}, \Dikin{}\}\).
To assess the mixing behaviours, we use two measures -- the error between the sample mean of these samples and the true parameter \(\theta^{\star}\) denoted by \(\widehat{\mathrm{Err}}_{k}\), and the variation in the average log-likelihood of the data \(D_{n}\) denoted by \(\widehat{\mathrm{LL}}_{k}\).
\begin{equation*}
    \widehat{\mathrm{Err}}_{k} = \frac{1}{d} \cdot \left\|\widehat{\theta}_{k}^{\textsf{alg}} - \theta^{\star}\right\|_{1}~,\quad  \widehat{\mathrm{NLL}}_{k} = \frac{1}{n} \cdot f(\widehat{\theta}^{\textsf{alg}}_{k}; D_{n})~\quad \text{where }\widehat{\theta}_{k}^{\textsf{alg}} = \frac{1}{N}\sum_{j=1}^{N}\theta_{k, j}^{\textsf{alg}}~.
\end{equation*}

For a given dimension \(d\), we consider \(10\) independent simulations.
In each simulation, we generate a dataset \(D_{n}\), where \(n = 20d\).
The covariate \(X\) for \((X, y) \in D_{n}\) is a \(d\)-dimensional vector, whose entries are drawn independently and with equal propability from the set \(\{-\nicefrac{1}{\sqrt{d}}, \nicefrac{1}{\sqrt{d}}\}\).
The true parameter is set as \(\theta^{\star} = \bm{1}_{d}\), and the corresponding response \(y \in \{0, 1\}\) is drawn according to a Bernoulli distribution where the probability of \(y = 1\) is given by \((1 + \exp(-\langle \theta^{\star}, X\rangle))^{-1}\).
We run \nameref{alg:mapla} and \Dikin{} for 3000 iterations with stepsize \(h = \frac{C_{h}}{\lambda_{\max} \cdot d}\).
Under this covariate model, \(\lambda_{\max}\) is almost a constant for large enough \(d\) as Bai-Yin's law states that \(\lambda_{\max} \to \left(\sqrt{20} + 1 + o(1)\right)^{2}\) as \(d\) increases.
For \(\textsf{meas} \in \{\widehat{\mathrm{Err}}, \widehat{\mathrm{NLL}}\}\), we plot the average variation of \(\textsf{meas}_{k}\) with the iteration \(k\), where the average is taken over simulations.
\cref{fig:meas-variation-d-32-64} showcases this variation for \(d \in \{32, 64\}\).

\begin{figure}[t]
    \centering
    \begin{subfigure}{0.49\linewidth}
        \centering
        \includegraphics[width=\linewidth]{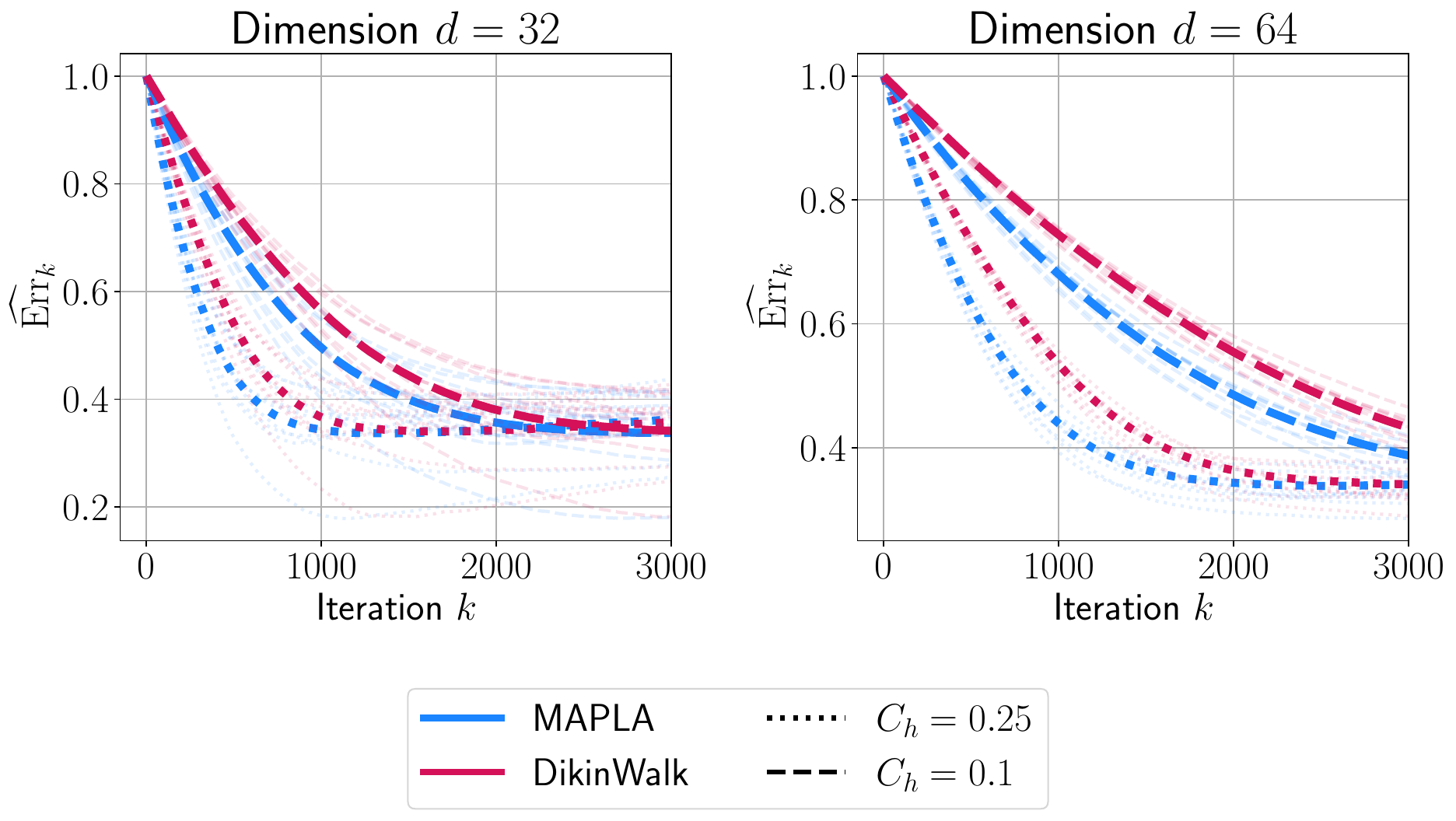}
    \end{subfigure}
    \hfill
    \begin{subfigure}{0.49\linewidth}
        \centering
        \includegraphics[width=\linewidth]{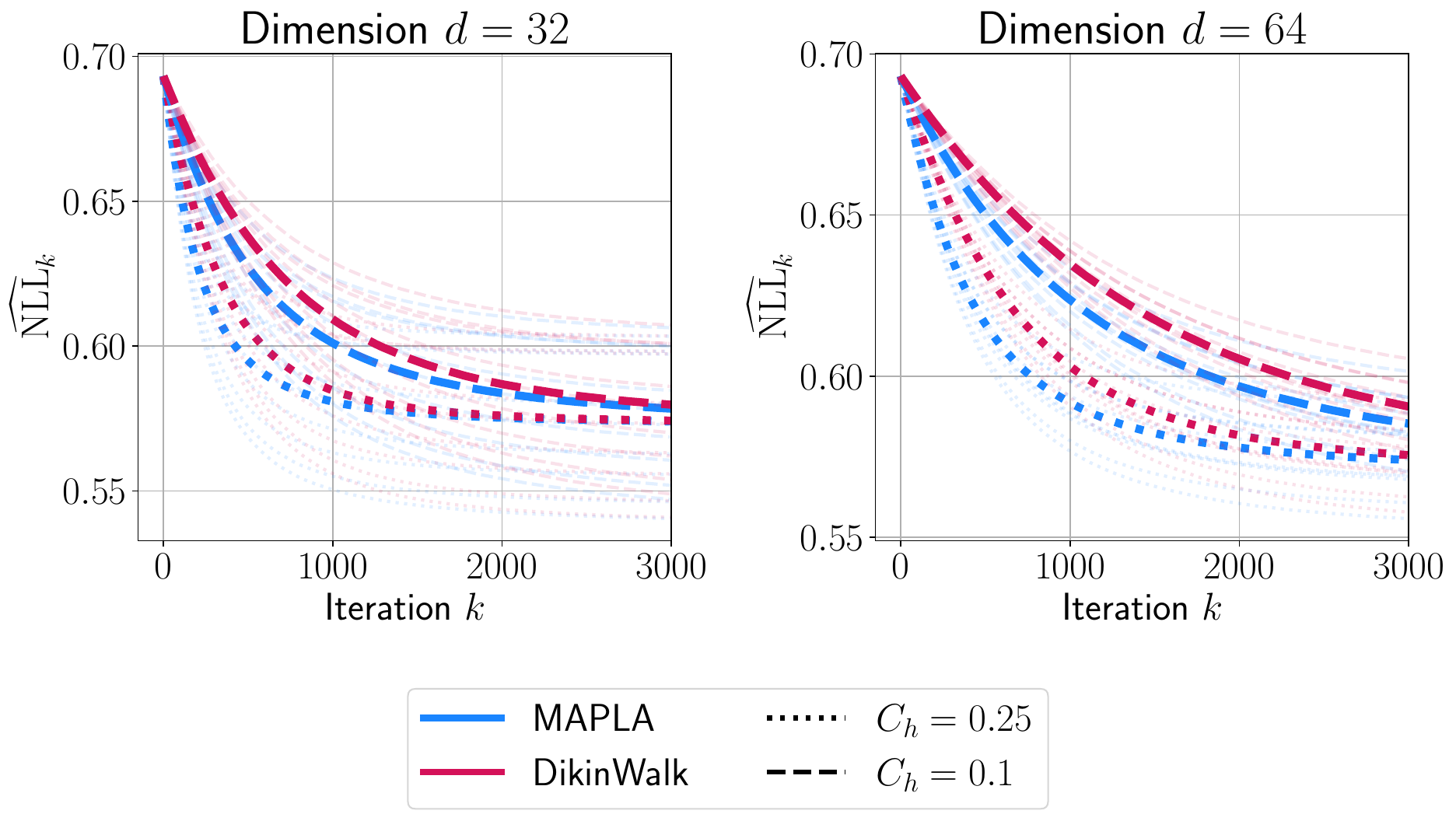}
    \end{subfigure}
    \caption{Average variation of \(\textsf{meas}_{k}\) for \(\textsf{meas} = \widehat{\mathrm{Err}}\) (left) and \(\widehat{\mathrm{NLL}}\) (right) with iteration \(k\).
    The faint lines shown depicts the variation over iterations per simulation.}
    \label{fig:meas-variation-d-32-64}
\end{figure}

We make the following observations from \cref{fig:meas-variation-d-32-64}, which mirror those made in the Dirichlet sampling setup.
First, for both algorithms, a larger value of \(C_{h}\) results in faster decrease in both measures \(\widehat{\mathrm{Err}}\) and \(\widehat{\mathrm{NLL}}\).
Second, for a fixed \(C_{h}\), we see that the rate of decrease of either measure is faster for \nameref{alg:mapla} than \Dikin{}.
This is fundamentally due to the fact that \nameref{alg:mapla} uses gradient information of \(\potential{\theta; D_{n}}\), thus demonstrating the utility of first-order information in addition to the unnormalised density.
Third, we see a larger difference in the rates of decrease between the methods for either measure when \(d = 64\) in comparison to when \(d = 32\).

\begin{wrapfigure}{r}{0.52\linewidth}
\centering
\includegraphics[width=\linewidth]{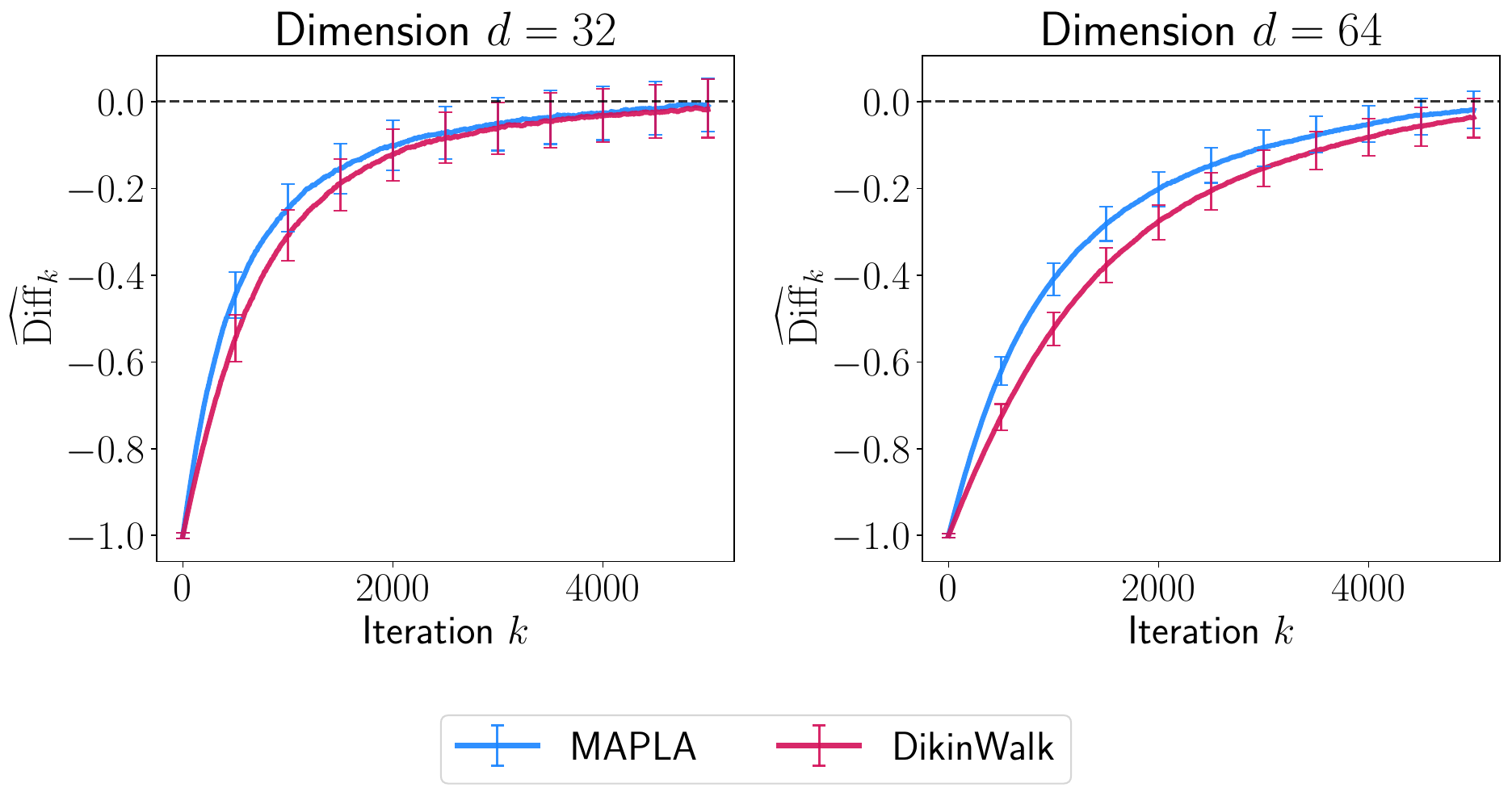}
\caption{Variation in the IQR of \(\widehat{\mathrm{Diff}}_{k}\) with iteration \(k\). The upper and lower bars indicate the 75\textsuperscript{th} and 25\textsuperscript{th} percentile of \(\widehat{\mathrm{Diff}}_{k}\) respectively.}
\label{fig:diff-quantile}
\end{wrapfigure}
In addition to the variation in \(\textsf{meas}_{k}\) with iteration \(k\), we also study the inter-quartile range (IQR) for the (dimension-wise average) error in the parameter following \citet{dwivedi2018log}.
In particular, given the set of samples \(\{\theta_{k,j}^{\textsf{alg}}\}_{j=1}^{N}\) generated by running \(\textsf{alg} \in \{\MAPLA{}, \Dikin{}\}\), we plot the variation in the IQR of \(\widehat{\mathrm{Diff}}_{k} := \left\{\frac{\bm{1}^{\top}(\theta_{k, j}^{\textsf{alg}} - \theta^{\star})}{d}\right\}_{j=1}^{N}\) with iteration \(k\).
For \(d \in \{32, 64\}\), we fix a randomly generated dataset \(D_{n}\) with \(n = 20d\), and run both methods with step size \(h = (10 \lambda_{\max} \cdot d)^{-1}\).
\cref{fig:diff-quantile} portrays a clear distinction in the ranges generated by \nameref{alg:mapla} and \Dikin{}, and this difference is again accentuated for a larger value of \(d\).

\section{Proofs}
\label{sec:proofs}
In this section, we give the proofs of the theorems in \cref{sec:mixing-guarantees}.
\cref{sec:proofs:key-lemmas} states the key lemmas which form the proofs of the theorems in \cref{sec:mixing-guarantees}, and we state these proofs in \cref{sec:proofs:complete-proof}.
Next in \cref{sec:proofs:proofs-of-key-lemmas}, we focus on the proofs of these key lemmas which involve the new isoperimetric inequality that we mentioned earlier, whose proof concludes this section in \cref{sec:proofs:proofs-isoperimetric-lemmas}.

\subsection{A pathway to obtain mixing time guarantees}
\label{sec:proofs:key-lemmas}

Let \(\bfQ = \{\calQ_{x} : x \in \primalspace\}\) be a Markov chain that is reversible with respect to \(\rho\) supported over \(\primalspace\).
The conductance / \(s\)-conductance of \(\bfQ\) (defined in \cref{sec:mixing-guarantees:prelims}) quantifies the likelihood of escaping a set in a worst-case sense, and hence intuitively determines how quickly a \(\bfQ\) mixes to its stationary distribution.
This intuition is made more precise in the following result from \citet{vempala2005geometric} (a collection of special cases arising from the classical analysis of \citet{lovasz1993random}) that relates the conductance / \(s\)-conductance of \(\bfQ\) to its mixing time.
Recall that \(\bbT_{\bfQ}^{k}\) is the transition operator defined by \(\bfQ\) applied \(k\) times.

\begin{proposition}[{\citet[Corr. 3.5]{vempala2005geometric}}]
\label{prop:conductance-to-mixing}
Let \(\bfQ = \{\calQ_{x} : x \in \primalspace\}\) be a lazy, reversible Markov chain with stationary distribution \(\rho\), and let \(\rho_{0}\) be a distribution whose support is contained in \(\primalspace\).
\begin{enumerate}[leftmargin=*]
\item If \(\rho_{0} \in \mathsf{Warm}(L_{\infty}, M, \rho)\), then
\begin{equation*}
    \TVdist(\bbT_{\bfQ}^{k}\rho_{0}, \rho) \leq \sqrt{M} \cdot \left(1 - \frac{{\Phi_{\bfQ}}^{2}}{2}\right)^{k}~.
\end{equation*}
\item If \(\rho_{0} \in \mathsf{Warm}(L_{1}, M, \rho)\), then for any \(\gamma > 0\)
\begin{equation*}
    \TVdist(\bbT_{\bfQ}^{k}\rho_{0}, \rho) \leq \gamma + \sqrt{\frac{M}{\gamma}} \cdot \left(1 - \frac{{\Phi_{\bfQ}}^{2}}{2}\right)^{k}~.
\end{equation*}
\item If \(\rho_{0} \in \mathsf{Warm}(L_{\infty}, M, \rho)\), then for any \(s \in (0, \nicefrac{1}{2})\)
\begin{equation*}
    \TVdist(\bbT_{\bfQ}^{k}\rho_{0}, \rho) \leq M \cdot s + M \cdot \left(1 - \frac{{\Phi_{\bfQ}^{s}}^{2}}{2}\right)^{k}~.
\end{equation*}
\end{enumerate}
\end{proposition}

Therefore, given a lower bound on the conductance / \(s\)-conductance of \(\bfQ\), we can obtain non-asymptotic mixing time guarantees for the Markov chain \(\bfQ\) from a warm initial distribution.
Indeed, these lower bounds have to be bounded away from \(0\) for a meaningful mixing time guarantee.

As discussed briefly previously in \cref{sec:mixing-guarantees:discussion} we use the one-step overlap technique pioneered by \citet{lovasz1999hit} to obtain lower bounds on the conductance / \(s\)-conductance.
The following lemma formally states how the one-step overlap yields the likelihood of escaping a set.

\begin{lemma}
\label{lem:conductance-lower-isoperimetry}
Consider a Markov chain \(\bfT = \{\calT_{x} : x \in \primalspace\}\) that is reversible with respect to a log-concave distribution \(\targetdist\) supported on \(\primalspace\).
Let the metric \(\metric{} : \interior{\primalspace} \to \bbS_{+}^{d}\) be self-concordant and \(\nu\)-symmetric.
Assume that there exists a convex subset \(\calS\) of \(\interior{\primalspace}\) such that for any \(x, y \in \calS\),
\begin{equation*}
    \|x - y\|_{\metric{y}} \leq \Delta \Rightarrow \TVdist(\calT_{x}, \calT_{y}) \leq \frac{1}{4}~
\end{equation*}
where \(\Delta \leq \frac{1}{2}\).
If the potential \(\potential{}\) of \(\targetdist\) satisfies a \((\mu, \metric{})\)-curvature lower bound, then for any measurable partition \(\{A_{1}, A_{2}\}\) of \(\primalspace\),
\begin{equation*}
    \int_{A_{1}} \calT_{x}(A_{2})~\targetdens(x)\rmd x \geq \frac{3 \Delta}{16} \cdot \min\left\{1,\max\left\{\frac{\widetilde{\mu}}{4},~\frac{1}{8\sqrt{\nu}}\right\}\right\} \cdot \min\left\{\targetdist(A_{1} \cap \calS), \targetdist(A_{2} \cap \calS)\right\}; \enskip \widetilde{\mu} = \frac{\mu}{8 + 4\sqrt{\mu}}~.
\end{equation*}
\end{lemma}

We give a proof of \Cref{lem:conductance-lower-isoperimetry} in \Cref{sec:proof:lem:conductance-lower-isoperimetry}.
This technique has been widely employed to obtain mixing time upper bounds for several MCMC algorithms that induce Markov chains which are reversible with respect to the target distribution which can either be constrained or unconstrained.

We recall that our focus is on obtaining lower bounds on the conductance / \(s\)-conductance of Markov chain \(\bfT\) induced by \nameref{alg:mapla}, which results from the Metropolis adjustment of the Markov chain \(\bfP\) defined by a single step of \ref{eq:PLA}.
To use \cref{lem:conductance-lower-isoperimetry} for \(\bfT\), we require checking the existence of a convex subset \(\calS\) where the one-step overlap assumed in its statement holds.
We identify \(\calS\) based on properties of the \(\potential{}\) of the target \(\targetdist\) and state these in the proofs of the theorems.
Our strategy to establish the one-step overlap in this subset is to bound \(\TVdist(\calP_{x}, \calP_{y})\) and \(\TVdist(\calT_{x}, \calP_{x})\) for any \(x, y \in \calS\), and then use the triangle inequality for the TV distance like so:
\begin{equation*}
    \TVdist(\calT_{x}, \calT_{y}) \leq \TVdist(\calT_{x}, \calP_{x}) + \TVdist(\calP_{x}, \calP_{y}) + \TVdist(\calT_{y}, \calP_{y})~.
\end{equation*}
Note that \(\bfT\) is defined based on the potential \(\potential{}\) of \(\targetdist\), the metric \(\metric{}\), and the step size \(h > 0\).

\begin{lemma}
\label{lem:px-py-small}
Let the metric \(\metric{} : \interior{\primalspace} \to \bbS_{+}^{d}\) be self-concordant.
Assume that there exists a convex set \(\calS \subseteq \interior{\primalspace}\) and \(\sfU_{\potential{}, \calS} \geq 0\) such that for any \(x \in \calS\)
\begin{equation*}
    \sup_{z :~ \|z - x\|_{\metric{x}} \leq \frac{1}{2}} \|\gradpotential{z}\|_{\metric{z}^{-1}} \leq \sfU_{\potential{}, \calS}~.
\end{equation*}
Then, for any \(x, y \in \calS\) such that \(\|x - y\|_{\metric{y}} \leq \frac{\sqrt{h}}{10}\) with \(h < 1\),
\begin{equation*}
    \TVdist(\calP_{x}, \calP_{y}) \leq \frac{1}{2} \cdot \sqrt{\frac{h \cdot d}{20} + \frac{1}{80} + \frac{9h \cdot \sfU_{\potential{}, \calS}^{2}}{2}}~.
\end{equation*}
\end{lemma}

\begin{lemma}
\label{lem:tx-px-small-SC}
Let the metric \(\metric{} : \interior{\primalspace} \to \bbS_{+}^{d}\) be self-concordant, and the potential \(\potential{} : \interior{\primalspace} \to \bbR\) satisfy \((\lambda, \metric{})\)-curvature upper bound.
Assume that there exists a convex set \(\calS \subseteq \interior{\primalspace}\) and \(\sfU_{\potential{}, \calS} \geq 0\) such that for any \(x \in \calS\),
\begin{equation*}
    \sup_{z :~ \|z - x\|_{\metric{x}} \leq \frac{1}{2}} \|\gradpotential{z}\|_{\metric{z}^{-1}} \leq \sfU_{\potential{}, \calS}~.
\end{equation*}
If the step size \(h > 0\) is bounded from above as
\begin{equation*}
    h \leq c \cdot \min\left\{\frac{1}{\sfU_{\potential{}, \calS}^{2}},~\frac{1}{\sfU_{\potential{}, \calS}^{\nicefrac{2}{3}}},~\frac{1}{(\sfU_{\potential{}, \calS} \cdot \lambda)^{\nicefrac{2}{3}}},~\frac{1}{d \cdot \lambda},~\frac{1}{d^{3}}\right\}; \quad c \leq \frac{1}{20}
\end{equation*}
then for any \(x \in \calS\)
\begin{equation*}
    \TVdist(\calT_{x}, \calP_{x}) \leq \frac{1}{16}~.
\end{equation*}
\end{lemma}

\begin{lemma}
\label{lem:tx-px-small-beyond-SC}
Let the metric \(\metric{} : \interior{\primalspace} \to \bbS_{+}^{d}\) be strongly, \(\alpha\)-lower trace, and average self-concordant, and the potential \(\potential{} : \interior{\primalspace} \to \bbR\) satisfy \((\lambda, \metric{})\)-curvature upper bound.
Assume that there exists a convex set \(\calS \subseteq \interior{\primalspace}\) and \(\sfU_{\potential{}, \calS} \geq 0\) such that for any \(x \in \calS\),
\begin{equation*}
    \sup_{z :~ \|z - x\|_{\metric{x}} \leq \frac{1}{2}} \|\gradpotential{z}\|_{\metric{z}^{-1}} \leq \sfU_{\potential{}, \calS}~.
\end{equation*}
If the step size \(h > 0\) is bounded from above as
\begin{equation*}
    h \leq c \cdot \min\left\{\frac{1}{\sfU_{\potential{}, \calS}^{2}},~\frac{1}{(\sfU_{\potential{}, \calS} \cdot (\alpha + 4))^{\nicefrac{2}{3}}},~\frac{1}{(\sfU_{\potential{}, \calS} \cdot \lambda)^{\nicefrac{2}{3}}},~\frac{1}{d \cdot \lambda},~\frac{1}{d \cdot \sfU_{\potential{}, \calS}},~\frac{1}{d \cdot (\alpha + 4)}\right\}; \quad c \leq \frac{1}{20}~,
\end{equation*}
then for any \(x \in \calS\)
\begin{equation*}
    \TVdist(\calT_{x}, \calP_{x}) \leq \frac{1}{16}~.
\end{equation*}
\end{lemma}

\subsection{Complete proofs of \Cref{thm:just-SC-mapcla,thm:more-than-SC-mapcla,thm:exp-densities-mapcla}}
\label{sec:proofs:complete-proof}

With the key lemmas \cref{lem:conductance-lower-isoperimetry,lem:px-py-small,lem:tx-px-small-SC,lem:tx-px-small-beyond-SC}, we can now prove the main theorems.

\begin{proof}[Proofs of \Cref{thm:just-SC-mapcla,thm:more-than-SC-mapcla}]

In both theorems, the potential is assumed to satisfy the \((\beta,\metric{})\)-gradient upper bound condition, which implies by definition that
\begin{equation*}
    \sup_{x \in \interior{\primalspace}} \|\gradpotential{x}\|_{\metric{x}^{-1}} \leq \beta~.
\end{equation*}
Thus, we can now use \cref{lem:tx-px-small-SC,lem:tx-px-small-beyond-SC} with \(\calS \leftarrow \interior{\primalspace}\), and \(\sfU_{\potential{}, \calS} \leftarrow \beta\).
The resulting upper bounds on the step size from these lemmas are exactly the bounds \(b_{\mathsf{SC}}(d, \lambda, \beta)\) and \(b_{\textsf{SC\textsubscript{++}}}\) in \cref{thm:just-SC-mapcla,thm:more-than-SC-mapcla} respectively.
Hence, for any \(x \in \interior{\primalspace}\), we have
\begin{equation*}
    \TVdist(\calT_{x}, \calP_{x}) \leq \frac{1}{16}~.
\end{equation*}

In \cref{thm:just-SC-mapcla}, the metric \(\metric{}\) is assumed to be self-concordant, and in \cref{thm:more-than-SC-mapcla}, the metric \(\metric{}\) is assumed to be strongly self-concordant which implies that it is self-concordant as well.
Due to this, we can use \cref{lem:px-py-small} with \(\calS \leftarrow \interior{\primalspace}\) and \(\sfU_{\potential{}, \calS} \leftarrow \beta\) to obtain for any \(x, y \in \interior{\primalspace}\) that satisfy \(\|x - y\|_{\metric{y}} \leq \frac{\sqrt{h}}{10}\) for \(h < 1\),
\begin{equation*}
    \TVdist(\calP_{x}, \calP_{y}) \leq \frac{1}{2} \cdot \sqrt{\frac{h \cdot d}{20} + \frac{1}{80} + \frac{9h \cdot \beta^{2}}{2}}~.
\end{equation*}
The bounds \(b_{\textsf{SC}}(d, \lambda, \beta)\) and \(b_{\textsf{SC\textsubscript{++}}}\) ensure that \(h \leq \frac{1}{20\beta^{2}}\) and \(h \leq \frac{1}{20d}\), which yields
\begin{equation*}
    \TVdist(\calP_{x}, \calP_{y}) \leq \frac{1}{2} \cdot \sqrt{\frac{h \cdot d}{20} + \frac{1}{80} + \frac{9h \cdot \beta^{2}}{2}} \leq \frac{1}{2} \cdot \sqrt{\frac{1}{400} + \frac{1}{80} + \frac{9}{40}} \leq \frac{1}{8}~.
\end{equation*}

Combining the bounds on \(\TVdist(\calT_{x}, \calP_{x})\) and \(\TVdist(\calP_{x}, \calP_{y})\), we obtain that for any \(x, y \in \interior{\primalspace}\) satisfying \(\|x - y\|_{\metric{y}} \leq \frac{\sqrt{h}}{10}\),
\begin{equation*}
    \TVdist(\calT_{x}, \calT_{y}) \leq \TVdist(\calT_{x}, \calP_{x}) + \TVdist(\calP_{x}, \calP_{y}) + \TVdist(\calT_{y}, \calP_{y}) \leq \frac{1}{16} + \frac{1}{8} + \frac{1}{16} = \frac{1}{4}~.
\end{equation*}
With this result, we apply \cref{lem:conductance-lower-isoperimetry} with \(\Delta \leftarrow \frac{\sqrt{h}}{10}\), \(\calS \leftarrow \interior{\primalspace}\) as the potential \(\potential{}\) is assumed to satisfy a \((\mu, \metric{})\)-curvature lower bound and the metric \(\metric{}\) is assumed to be \(\nu\)-symmetric.
From this, we know that there exists a universal constant \(C > 0\) such that for any \(A \subseteq \primalspace\),
\begin{equation*}
    \int_{A} \calT_{x}(\primalspace \setminus A)\targetdens(x) ~\rmd x \geq C \cdot \sqrt{h} \cdot \min\left\{1, \max\left\{\widetilde{\mu},~\frac{1}{\sqrt{\nu}}\right\}\right\} \cdot \min\{\targetdist(A), \targetdist(\primalspace \setminus A)\}~.
\end{equation*}
Above, we have used the fact that \(\primalspace\) is a convex subset of \(\bbR^{d}\) and \(\targetdist\) is log-concave, and hence
\begin{equation*}
    \targetdist(A \cap \interior{\primalspace}) = \targetdist(A); \quad \targetdist((\primalspace \setminus A) \cap \interior{\primalspace}) = \targetdist(\primalspace \setminus A)~.
\end{equation*}
This directly implies by definition that
\begin{equation}
\label{eq:conductance-lower-1}
    \Phi_{\bfT} \geq C \cdot \sqrt{h} \cdot \min\left\{1, \max\left\{\widetilde{\mu},~ \frac{1}{\sqrt{\nu}}\right\}\right\}~.
\end{equation}

For the final step of this proof, we invoke \cref{prop:conductance-to-mixing}.
We have the following cases.
\begin{description}
    \item [For \(\Pi_{0} \in \calC := \mathsf{Warm}(L_{\infty}, M, \targetdist)\):]
    \begin{equation*}
        \TVdist(\bbT_{\bfT}^{k}\targetdist_{0}, \targetdist) \leq \sqrt{M} \cdot \exp\left(-\frac{k \cdot {\Phi_{\bfT}}^{2}}{2}\right)~.
    \end{equation*}
    By setting the right hand side to be less than \(\delta\) for \(\delta \in (0, \nicefrac{1}{2})\) and using \Cref{eq:conductance-lower-1}, we get
    \begin{equation*}
        \mixingtime{\delta; \bfT, \calC} = \frac{2}{{\Phi_{\bfT}}^{2}} \cdot \log\left(\frac{\sqrt{M}}{\delta}\right) \leq \frac{C'}{h} \cdot \max\left\{1, \min\left\{\frac{1}{\widetilde{\mu}^{2}},~\nu\right\}\right\} \cdot \log\left(\frac{\sqrt{M}}{\delta}\right)
    \end{equation*}
    for some universal positive constant \(C'\).

    \item [For \(\Pi_{0} \in \calC := \mathsf{Warm}(L_{1}, M, \targetdist)\):]
    \begin{equation*}
        \TVdist(\bbT_{\bfT}^{k}\Pi_{0}, \targetdist) \leq \gamma + \sqrt{\frac{M}{\gamma}} \cdot \exp\left(-\frac{k \cdot {\Phi_{\bfT}}^{2}}{2}\right)
    \end{equation*}
    for any \(\gamma > 0\).
    For \(\delta \in (0, \nicefrac{1}{2})\), by setting \(\gamma = \frac{\delta}{2}\) and the second term on the right hand side to be less than \(\frac{\delta}{2}\), the upper bound is at most \(\delta\).
    With \Cref{eq:conductance-lower-1}, we get
    \begin{equation*}
        \mixingtime{\delta; \bfT, \calC} = \frac{6}{{\Phi_{\bfT}}^{2}} \cdot \log\left(\frac{(2M)^{\nicefrac{1}{3}}}{\delta}\right) \leq \frac{C''}{h} \cdot \max\left\{1, \min\left\{\frac{1}{\widetilde{\mu}^{2}}, \nu\right\}\right\} \cdot \log\left(\frac{M^{\nicefrac{1}{3}}}{\delta}\right)
    \end{equation*}
    for some universal positive constant \(C''\).
\end{description}
\end{proof}

\begin{proof}[Proof of \Cref{thm:exp-densities-mapcla}]
For \(s \in (0, \nicefrac{1}{2})\),
\begin{equation*}
    \primalspace_{s} \defeq \left\{x \in \interior{\primalspace} : \|\sigma\|^{2}_{\metric{x}^{-1}} \leq 25d^{2} \cdot \log^{2}\left(\frac{1}{s}\right)\right\}~,
\end{equation*}
where \(\sigma = \gradpotential{x}\) for all \(x \in \interior{\primalspace}\).
Equivalently,
\begin{equation*}
    \sup_{x \in \calC_{s}} \|\gradpotential{x}\|_{\metric{x}^{-1}} \leq 5d \cdot \log\left(\frac{1}{s}\right)~.
\end{equation*}
In \citet{kook2023condition}, we have the following properties of \(\primalspace_{s}\):
\begin{description}[itemsep=0pt]
    \item [Lem. 43:] If \(\metric{}\) is self-concordant, \(\targetdist(\primalspace_{s}) \geq 1 - s\).
    \item [Lem. 42:] If \(\metric{}\) satisfies \(\rmD^{2}\metric{x}[v, v] \succeq \bm{0}\) for all \(x \in \interior{\primalspace}\) and \(v \in \bbR^{d}\), \(\primalspace_{s}\) is convex.
\end{description}
In the setting of \cref{thm:exp-densities-mapcla}, the metric \(\metric{}\) is assumed to satisfy both preconditions in the properties above.
Additionally, since the metric \(\metric{}\) is self-concordant, from \cref{lem:dikin1-in-domain-self-concordant-loewner} we have for any \(z \in \calE_{x}^{\metric{}}(1)\) (defined in \cref{eq:dikin-ellipsoid}) and \(u \in \bbR^{d}\) that
\begin{equation*}
    (1 - \|z - x\|_{\metric{x}})^{2} \cdot \metric{x} \preceq \metric{z} \Leftrightarrow \|u\|_{\metric{x}} \leq \frac{\|u\|_{\metric{z}}}{(1 - \|z - x\|_{\metric{x}})}~.
\end{equation*}
As a corollary, if \(\|z - x\|_{\metric{x}} \leq \frac{1}{2}\) and \(\|u\|_{\metric{z}} \leq 1\), then \(\|u\|_{\metric{x}} \leq \frac{1}{(1 - \|z - x\|_{\metric{x}})} \leq 2\), and
\begin{align*}
    \|\gradpotential{z}\|_{\metric{z}^{-1}} = \|\sigma\|_{\metric{z}^{-1}} &= \sup_{u :~\|u\|_{\metric{z}} \leq 1} \langle \sigma, u\rangle \\
    &\leq \sup_{u :~\|u\|_{\metric{x}} \leq 2} \langle \sigma, u\rangle \\
    &=2 \cdot \|\sigma\|_{\metric{x}^{-1}} = 2 \cdot \|\gradpotential{x}\|_{\metric{x}^{-1}}~.\\
\end{align*}
Therefore, for any \(x \in \primalspace_{s}\),
\begin{equation*}
    \sup_{z:~\|z - x\|_{\metric{x}} \leq \frac{1}{2}} \|\gradpotential{z}\|_{\metric{z}^{-1}} \leq 2 \cdot 5d \cdot \log\left(\frac{1}{s}\right) = 10d \cdot \log\left(\frac{1}{s}\right)~.
\end{equation*}

The property \(\rmD^{2}\metric{x}[v, v] \succeq \bm{0}\) for all \(x \in \interior{\primalspace}\) and \(v \in \bbR^{d}\) implies that \(\metric{}\) is \(0\)-lower trace self-concordant, due to the following calculation.
\begin{equation*}
    \trace(\metric{x}^{-\nicefrac{1}{2}}\rmD^{2}\metric{x}[v, v]\metric{x}^{-\nicefrac{1}{2}}) = \sum_{i=1}^{d}\rmD^{2}\metric{x}[v, v, \metric{x}^{-\nicefrac{1}{2}}e_{i}, \metric{x}^{-\nicefrac{1}{2}}e_{i}] \geq 0~.
\end{equation*}
For warmness parameter \(M \geq 1\) and error tolerance \(\delta \in (0, \nicefrac{1}{2})\), we pick \(s = \frac{\delta}{2M} \leq \frac{1}{4}\).
Since \(\potential{x} = \sigma^{\top}x\), it satisfies a \((0,\metric{})\)-curvature lower and upper bounds.
The step size bound \(b_{\textsf{Exp}}(d, M, \delta)\) in the theorem exactly matches the bound arising from applying \cref{lem:tx-px-small-beyond-SC} with \(\calS \leftarrow \primalspace_{\nicefrac{\delta}{2M}}\) and \(\sfU_{\potential{}, \calS} = 10d \cdot \log(\nicefrac{2M}{\delta})\), and we have for any \(x \in \calS\) that
\begin{equation*}
    \TVdist(\calT_{x}, \calP_{x}) \leq \frac{1}{16}~.
\end{equation*}
We use \cref{lem:px-py-small} with \(\calS \leftarrow \primalspace_{\nicefrac{\delta}{2M}}\) and \(\sfU_{\potential{}, \calS} \leftarrow 10d \cdot \log(\nicefrac{2M}{\delta})\) to obtain for any \(x, y \in \calS\) such that \(\|x - y\|_{\metric{y}} \leq \frac{\sqrt{h}}{10}\),
\begin{equation*}
    \TVdist(\calP_{x}, \calP_{y}) \leq \frac{1}{2} \cdot \sqrt{\frac{h \cdot d}{20} + \frac{1}{80} + \frac{9h^{2} \cdot d^{2} \cdot \log^{2}(\nicefrac{2M}{\delta})}{2}}~.
\end{equation*}
The bound \(b_{\textsf{Exp}}(d, M, \delta)\) ensures that \(h \cdot d \leq h \cdot d^{2} \leq \frac{1}{2000 \cdot \log^{2}(\nicefrac{2M}{\delta})}\)
and therefore
\begin{equation*}
    \TVdist(\calP_{x}, \calP_{y}) \leq \frac{1}{2} \cdot \sqrt{\frac{h \cdot d}{20} + \frac{1}{80} + \frac{9h \cdot d^{2}}{2}} \leq \frac{1}{2} \cdot \sqrt{\frac{1}{400} + \frac{1}{80} + \frac{9}{40}} \leq \frac{1}{8}~.
\end{equation*}
Therefore, for any \(x, y \in \primalspace_{\nicefrac{\delta}{2M}}\) such that \(\|x - y\|_{\metric{y}} \leq \frac{\sqrt{h}}{10}\),
\begin{equation*}
    \TVdist(\calT_{x}, \calT_{y}) \leq \TVdist(\calT_{x}, \calP_{x}) + \TVdist(\calP_{x}, \calP_{y}) + \TVdist(\calT_{y}, \calP_{y}) \leq \frac{1}{16} + \frac{1}{8} + \frac{1}{16} = \frac{1}{4}~.
\end{equation*}

With these results, we can obtain a lower bound on the \(s\)-conductance of \(\bfT\) for \(s = \frac{\delta}{2M}\) as indicated previously.
Using \cref{lem:conductance-lower-isoperimetry} with \(\Delta \leftarrow \frac{\sqrt{h}}{10}\) and \(\calS \leftarrow \primalspace_{s}\), we have for any \(A \subseteq \primalspace\)
\begin{equation}
\label{eq:exp-density-conductance-lower}
    \int_{A} \calT_{x}(\primalspace \setminus A) \targetdens(x)~\rmd x \geq C \cdot \sqrt{h} \cdot \min\left\{1, \frac{1}{\sqrt{\nu}}\right\} \cdot \min\{\targetdist(A \cap \primalspace_{s}), \targetdist((\primalspace \setminus A) \cap \primalspace_{s})\}
\end{equation}
for a universal constant \(C\).
Note that for any subset \(B\) of \(\primalspace\), it holds that \(\targetdist(B) = \targetdist(B \cap \primalspace_{s}) + \targetdist(B \cap (\primalspace \setminus \primalspace_{s}))\) as \(\targetdist\) is supported on \(\primalspace\).
Since \(\targetdist(\primalspace_{s}) \geq 1 - s\), we get
\begin{gather*}
    \targetdist(B \cap \primalspace_{s}) = \targetdist(B) - \targetdist(B \cap (\primalspace \setminus \primalspace_{s})) \geq \targetdist(B) - \targetdist(\primalspace \setminus \primalspace_{s}) \geq \targetdist(B) - s~.
\end{gather*}
We apply this fact for \(B \leftarrow A\) and \(B \leftarrow (\primalspace \setminus A)\) for \(A \subseteq \primalspace\), which gives
\begin{equation*}
    \targetdist(A \cap \primalspace_{s}) \geq \targetdist(A) - s, \quad \targetdist((\primalspace \setminus A) \cap \primalspace_{s}) \geq \targetdist(\primalspace \setminus A) - s~.
\end{equation*}

Substituting these inequalities in \cref{eq:exp-density-conductance-lower}, we get for any \(A \subset \primalspace\) satisfying \(\targetdist(A) \in (s, 1 - s)\) that
\begin{equation*}
    \int_{A}\calT_{x}(\primalspace \cap A)\targetdens(x) ~\rmd x \geq C \cdot \sqrt{h} \cdot \min\left\{1, \frac{1}{\sqrt{\nu}}\right\} \cdot \min\left\{\targetdist(A) - s, \targetdist(\primalspace \setminus A) - s\right\}~,
\end{equation*}
which implies by definition that
\begin{equation}
\label{eq:s-conductance-lower}
    \Phi_{\bfT}^{s} \geq C \cdot \sqrt{h} \cdot \min\left\{1, \frac{1}{\sqrt{\nu}}\right\}~, \quad s = \frac{\delta}{2M}~.
\end{equation}

Finally, we call \cref{prop:conductance-to-mixing} which states that for \(\targetdist_{0} \in \calC := \mathsf{Warm}(L_{\infty}, M, \targetdist)\),
\begin{equation*}
    \TVdist(\bbT_{\bfT}^{k}\targetdist_{0}, \targetdist) \leq M \cdot s + M \cdot \left(1 - \frac{{\Phi_{\bfT}^{s}}^{2}}{2}\right)^{k} \leq \frac{\delta}{2} + M \cdot \exp\left(-\frac{k \cdot {\Phi_{\bfT}^{s}}^{2}}{2}\right)~.
\end{equation*}
Setting the second term on the right hand side to be less than \(\frac{\delta}{2}\), and using \Cref{eq:s-conductance-lower}, we obtain
\begin{equation*}
    \mixingtime{\delta; \bfT, \calC} = \frac{2}{{\Phi_{\bfT}^{s}}^{2}} \cdot \log\left(\frac{2M}{\delta}\right) \leq \frac{C'''}{h} \cdot \max\left\{1, \nu\right\} \cdot \log\left(\frac{M}{\delta}\right)
\end{equation*}
for some universal positive constant \(C'''\).
\end{proof}

\subsection{Proving the key lemmas in \cref{sec:proofs:key-lemmas}}
\label{sec:proofs:proofs-of-key-lemmas}

\subsubsection{Proof of \Cref{lem:conductance-lower-isoperimetry}}
\label{sec:proof:lem:conductance-lower-isoperimetry}

\cref{lem:conductance-lower-isoperimetry} is derived from two isoperimetric inequalities.
The first inequality is a consequence of log-concavity of \(\targetdist\) and \(\nu\)-symmetry of the metric \(\metric{}\), and the other holds when the potential \(\potential{}\) of \(\targetdist\) satisfies \((\mu,\metric{})\)-curvature lower bound for a self-concordant metric \(\metric{}\).

\begin{lemma}
\label{lem:log-concave-isoperimetry}
Let \(\targetdist\) be a log-concave distribution whose support is \(\primalspace\), and consider a metric \(\metric{} : \interior{\primalspace} \to \bbS_{+}^{d}\) that is \(\nu\)-symmetric.
Let \(\targetdist_{\calA}\) denote the restriction of \(\targetdist\) to \(\calA \subseteq \primalspace\).
For any partition \(\{S_{1}, S_{2}, S_{3}\}\) of a convex subset \(\calS\) of \(\primalspace\), we have
\begin{equation*}
    \targetdist_{\calS}(S_{3}) \geq \frac{1}{\sqrt{\nu}} \cdot \inf_{y \in S_{2}, x \in S_{1}} \|x - y\|_{\metric{y}} \cdot \targetdist_{\calS}(S_{2}) \cdot \targetdist_{\calS}(S_{1})~.
\end{equation*}
\end{lemma}

For the following lemma, we introduce some additional notation.
The geodesic distance between \(x, y \in \interior{\primalspace}\) with respect to the metric \(\metric{}\) is denoted by \(d_{\metric{}}(x, y)\).

\begin{lemma}
\label{lem:curv-lower-bdd-isoperimetry}
Consider a log-concave distribution \(\targetdist\) whose support \(\primalspace\) and a metric \(\metric{} : \interior{\primalspace} \to \bbS_{+}^{d}\).
If \(\metric{}\) is self-concordant and the potential \(\potential{}\) of \(\targetdist\) satisfies \((\mu,\metric{})\)-curvature lower bound, then for any partition \(\{S_{1}, S_{2}, S_{3}\}\) of \(\primalspace\),
\begin{equation*}
    \targetdist(S_{3}) \geq \frac{\mu}{8 + 4\sqrt{\mu}} \cdot \inf_{y \in S_{2}, x \in S_{1}} d_{\metric{}}(x, y) \cdot \min\left\{\targetdist(S_{2}), \targetdist(S_{1})\right\}~.
\end{equation*}
\end{lemma}

\cref{lem:log-concave-isoperimetry} is the combination of a isoperimetric inequality for log-concave distributions \citep[Thm. 2.2]{lovasz2003hit} and \(\nu\)-symmetry \citep[Lem. 2.3]{laddha2020strong}.
\cref{lem:curv-lower-bdd-isoperimetry} is a new isoperimetric inequality, and is a generalisation of the result in \citet[Lem. 9]{gopi2023algorithmic} and \citet[Lem. B.7]{kook2024gaussian}.
These prior results assume \(\mu\)-relative convexity with respect to a self-concordant function \(\psi\) whose Hessian is approximately the metric \(\metric{}\).
It is interesting to note that when \(\mu\) is small, \(\frac{\mu}{8 + 4\sqrt{\mu}}\) scales as \(\mu\), whereas when \(\mu\) is large, this ratio scales as \(\sqrt{\mu}\) instead.
To work with \cref{lem:curv-lower-bdd-isoperimetry}, we require relating \(d_{\metric{}}(x, y)\) to \(\|x - y\|_{\metric{}}\), and this is possible when \(d_{\metric{}}(x, y)\) is sufficiently small as given in the following lemma.

\begin{lemma}[{\citet[Lem. 3.1]{nesterov2002riemannian}}]
\label{lem:nesterov-todd-result}
Let the metric \(\metric{} : \interior{\primalspace} \to \bbS_{+}^{d}\) be self-concordant.
For \(x, y \in \interior{\primalspace}\) and \(\kappa \in [0, 1)\), if \(d_{\metric{}}(x, y) \leq \kappa - \frac{\kappa^{2}}{2}\), then \(\|x - y\|_{\metric{x}} \leq \kappa < 1\).
\end{lemma}

Equipped with the above lemmas that we will prove later, we now state the proof of \Cref{lem:conductance-lower-isoperimetry}.

\begin{proof}
Our proof follows the same structure as prior results based on the one-step overlap.
In particular, here we follow \citet[Proof of Lem. 2]{dwivedi2018log}.
Due to the reversibility of \(\bfT\),
\begin{equation*}
    \int_{A_{1}} \calT_{x}(A_{2})\targetdens(x)~\rmd x = \int_{A_{2}} \calT_{y}(A_{1})\targetdens(y)~\rmd y~.
\end{equation*}

Define the sets
\begin{equation*}
    A_{1}' \defeq \left\{x \in A_{1} \cap \calS ~:~ \calT_{x}(A_{2}) < \frac{3}{8}\right\} \qquad A_{2}' \defeq \left\{y \in A_{2} \cap \calS ~:~ \calT_{y}(A_{1}) < \frac{3}{8}\right\}~.
\end{equation*}

Using the definitions above, and the fact that \(A_{i} \cap (\calS \setminus A_{i}') \subseteq A_{i}\),
\begin{align*}
    \int_{A_{1}} \calT_{x}(A_{2})\targetdens(x)~\rmd x &= \frac{1}{2}\left(\int_{A_{1}} \calT_{x}(A_{2})\targetdens(x)~\rmd x + \int_{A_{2}} \calT_{y}(A_{1})\targetdens(y) ~\rmd y\right) \\
    &\geq \frac{1}{2}\left(\int_{A_{1} \cap (\calS \setminus A_{1}')} \calT_{x}(A_{2})\targetdens(x)~\rmd x + \int_{A_{2} \cap (\calS \setminus A_{2}')} \calT_{y}(A_{1})\targetdens(y) ~\rmd y\right) \\
    &\geq \frac{3}{16} \cdot \left(\int_{A_{1} \cap (\calS \setminus A_{1}')}\targetdens(x)~\rmd x + \int_{A_{2} \cap (\calS \setminus A_{2}')}\targetdens(y)~\rmd y\right) \\
    &= \frac{3}{16} \cdot (\targetdist(A_{1} \cap (\calS \setminus A_{1}')) + \targetdist(A_{2} \cap (\calS \setminus A_{2}')))~\numberthis\label{eq:conductance-numerator-lower-primitive}.
\end{align*}

If \(\targetdist(A_{i}') \leq \frac{\targetdist(A_{i} \cap \calS)}{2}\) for given \(i \in \{1, 2\}\),
\begin{gather*}
    \targetdist(A_{i} \cap (\calS \setminus A_{1}')) \geq \targetdist(A_{i} \cap \calS) - \targetdist(A_{i}') \geq \frac{1}{2} \cdot \targetdist(A_{i} \cap \calS)~.
\end{gather*}
Hence, if \(\targetdist(A_{1}') \leq \frac{\targetdist(A_{1} \cap \calS)}{2}\) or \(\targetdist(A_{2}') \leq \frac{\targetdist(A_{2} \cap \calS)}{2}\), continuing from \Cref{eq:conductance-numerator-lower-primitive},
\begin{equation}
\label{eq:conductance-numerator-lower-1}
    \int_{A_{1}}\calT_{x}(A_{2}) \targetdens(x) ~\rmd x \geq \frac{3}{16} \cdot \min\left\{\targetdist(A_{1} \cap \calS)~, \targetdist(A_{2} \cap \calS)\right\}~.
\end{equation}

An alternative lower bound continuing from \Cref{eq:conductance-numerator-lower-primitive} is
\begin{align*}
    \int_{A_{1}} \calT_{x}(A_{2})\targetdens(x)~\rmd x &\geq \frac{3}{16} \cdot (\targetdist(A_{1} \cap (\calS \setminus A_{1}')) + \targetdist(A_{2} \cap (\calS \setminus A_{2}'))) \\
    &= \frac{3}{16} \cdot \targetdist(\calS \setminus A_{1}' \setminus A_{2}') \\
    &= \frac{3}{16} \cdot \targetdist(\calS) \cdot \targetdist_{\calS}(\calS \setminus A_{1}' \setminus A_{2}')~\numberthis\label{eq:conductance-numerator-lower-2-prelim}
\end{align*}
where recall that \(\targetdist_{\calS}\) is the restriction of \(\targetdist\) to \(\calS\) with density \(\targetdens_{\calS}(x) = \frac{e^{-\potential{}_{|\calS}(x)}}{\targetdist(\calS)}\).
By definition of \(A_{1}'\) and \(A_{2}'\), for any \(x' \in A_{1}'\) and \(y' \in A_{2}'\),
\begin{equation*}
    \TVdist(\calT_{x'}, \calT_{y'}) \geq \calT_{x'}(A_{1}) - \calT_{y'}(A_{1}) = 1 - \calT_{x'}(A_{2}) - \calT_{y'}(A_{1}) > 1 - \frac{3}{8} - \frac{3}{8} = \frac{1}{4}~.
\end{equation*}
By the one-step overlap assumed in the statement, we have that \(\|x' - y'\|_{\metric{y'}} > \Delta\) for \(\Delta \leq \frac{1}{2}\).
The remainder of the proof deals with the case where \(\targetdist(A_{i}') \geq \frac{\targetdist(A_{i} \cap \calS)}{2}\) for both \(i = 1, 2\).

\paragraph{When \(\metric{}\) is \(\nu\)-symmetric}

Since \(x' \in A_{1}'\) and \(y' \in A_{2}'\) are arbitrary, we have
\begin{equation*}
    \inf_{x' \in A_{1}', y' \in A_{2}'} \|x' - y'\|_{\metric{y'}} \geq \Delta~.
\end{equation*}
Now, we use \Cref{lem:log-concave-isoperimetry} for the partition \(\{A_{1}', A_{2}', \calS \setminus A_{1}' \setminus A_{2}'\}\) of \(\calS\).
\begin{align*}
    \targetdist_{\calS}(\calS \setminus A_{1}' \setminus A_{2}') &\geq \frac{\Delta}{\sqrt{\nu}} \cdot \targetdist_{\calS}(A_{1}') \cdot \targetdist_{\calS}(A_{2}') \\
    &\geq \frac{\Delta}{\sqrt{\nu}} \cdot \frac{\targetdist(A_{1}')}{\targetdist(\calS)} \cdot \frac{\targetdist(A_{2}')}{\targetdist(\calS)} \\
    &\geq \frac{\Delta}{4\sqrt{\nu}} \cdot \frac{\targetdist(A_{1} \cap \calS)}{\targetdist(\calS)} \cdot \frac{\targetdist(A_{2} \cap \calS)}{\targetdist(\calS)}~,
\end{align*}
where the last inequality is due to the assumption that \(\targetdist(A_{i}') \geq \frac{\targetdist(A_{i} \cap \calS)}{2}\) for \(i \in \{1, 2\}\).
Since \(A_{1}\) and \(A_{2}\) form a partition of \(\primalspace\),
\begin{equation*}
    \targetdist(A_{1} \cap \calS) + \targetdist(A_{2} \cap \calS) = \targetdist(\calS) \Leftrightarrow
    \frac{\targetdist(A_{2} \cap \calS)}{\targetdist(\calS)} = 1 - \frac{\targetdist(A_{1} \cap \calS)}{\targetdist(\calS)}~.
\end{equation*}
Using the algebraic fact that \(t(1 - t) \geq \frac{1}{2} \min\{t, 1 - t\}\) for \(t \in [0, 1]\), we obtain the lower bound
\begin{equation}
\label{eq:conductance-numerator-lower-2-sym}
    \targetdist_{\calS}(\calS \setminus A_{1}' \setminus A_{2}') \geq \frac{\Delta}{8\sqrt{\nu}} \cdot \frac{\min\{\targetdist(A_{1} \cap \calS), \targetdist(A_{2} \cap \calS)\}}{\targetdist(\calS)}~.
\end{equation}

\paragraph{When \(\potential{}\) satisfies \((\mu, \metric{})\)-curvature lower bound}

Since \(\Delta \leq \frac{1}{2}\) and \(\|x' - y'\|_{\metric{y'}} \geq \Delta\), we get from \cref{lem:nesterov-todd-result} that
\begin{equation*}
    d_{\metric{}}(y', x') > \Delta - \frac{\Delta^{2}}{2} \Rightarrow \inf_{y' \in A_{2}', x' \in A_{1}'} d_{\metric{}}(y', x') > \Delta - \frac{\Delta^{2}}{2} \geq \frac{\Delta}{2}
\end{equation*}
where the last inequality uses the algebraic fact that \(t - \frac{t^{2}}{2} \geq \frac{t}{2}\) for \(t \in [0, 1]\).

By definition of \(\targetdist_{\calS}\), its potential \(f_{|\calS}\) satisfies \((\mu, \metric{})\)-curvature lower bound over \(\calS\).
This permits using \Cref{lem:curv-lower-bdd-isoperimetry} for \(\targetdist_{\calS}\) for the partition \(\{A_{1}', A_{2}', \calS \setminus A_{1}' \setminus A_{2}'\}\), which results in
\begin{align*}
    \targetdist_{\calS}(\calS \setminus A_{1}' \setminus A_{2}') &\geq \frac{\mu}{8 + 4\sqrt{\mu}} \cdot \inf_{y' \in A_{2'}, x' \in A_{1}'} d_{\metric{}}(y', x') \cdot \min\left\{\targetdist_{\calS}(A_{1}'), \targetdist_{\calS}(A_{2}')\right\}~\\
    &\geq \frac{\mu}{8 + 4\sqrt{\mu}} \cdot \frac{\Delta}{2} \cdot \min\left\{\targetdist_{\calS}(A_{1}'), \targetdist_{\calS}(A_{2}')\right\} \\
    &= \frac{\mu}{8 + 4\sqrt{\mu}} \cdot \frac{\Delta}{2} \cdot \frac{\min\{\targetdist(A_{1}'), \targetdist(A_{2}')\}}{\targetdist(\calS)} \\
    &\geq \frac{\mu}{8 + 4\sqrt{\mu}} \cdot \frac{\Delta}{4} \cdot \frac{\min\{\targetdist(A_{1} \cap \calS), \targetdist(A_{2} \cap \calS)\}}{\targetdist(\calS)}~,\numberthis\label{eq:conductance-numerator-lower-2-relconv}
\end{align*}
where the last inequality is due to the assumption that \(\targetdist(A_{i}') \geq \frac{\targetdist(A_{i} \cap \calS)}{2}\) for \(i \in \{1, 2\}\).

Combining \Cref{eq:conductance-numerator-lower-2-relconv} and \Cref{eq:conductance-numerator-lower-2-sym} in \Cref{eq:conductance-numerator-lower-2-prelim}, we get the net lower bound in conjunction with \Cref{eq:conductance-numerator-lower-1}
\begin{equation}
\label{eq:conductance-number-lower}
\int_{A_{1}} \calT_{x}(A_{2})\targetdens(x) ~\rmd x \geq \frac{3\Delta}{16} \cdot \min\left\{1, \max\left\{\frac{\widetilde{\mu}}{4},~ \frac{1}{8\sqrt{\nu}}\right\}\right\} \cdot \min\left\{\targetdist(A_{1} \cap \calS), \targetdist(A_{2} \cap \calS)\right\}~.
\end{equation}
\end{proof}

\subsubsection{Proof of \Cref{lem:px-py-small}}

\begin{proof}
The KL divergence between two distribution \(\rho_{1}\) and \(\rho_{2}\) whre \(\rho_{1}\) is absolutely continuous with respect to \(\rho_{2}\) is defined as
\begin{equation*}
    \KLdist(\rho_{1} \| \rho_{2}) = \int \rmd\rho_{1}(x) \log \frac{\rmd\rho_{1}(x)}{\rmd\rho_{2}(x)}~.
\end{equation*}
Pinsker's inequality yields a bound on the TV distance as
\begin{equation*}
    \TVdist(\gaussian{\mu_{1}}{\Sigma_{1}}, \gaussian{\mu_{2}}{\Sigma_{2}}) \leq \sqrt{\frac{1}{2} \cdot \KLdist(\gaussian{\mu_{1}}{\Sigma_{1}} ~\|~ \gaussian{\mu_{2}}{\Sigma_{2}})}~.
\end{equation*}
Hence, to obtain a bound on \(\TVdist(\calP_{y}, \calP_{x})\), it suffices to bound \(\KLdist(\calP_{y} \| \calP_{x})\).

When \(\rho_{1}\), \(\rho_{2}\) are Gaussian distributions, the KL divergence has a closed form that is stated below.
\begin{equation*}
    \KLdist(\gaussian{\mu_{1}}{\Sigma_{1}} ~\|~ \gaussian{\mu_{2}}{\Sigma_{2}}) = \frac{1}{2}\left(\trace(\Sigma_{1}\Sigma_{2}^{-1} - \rmI_{d \times d}) - \log \det \Sigma_{1}\Sigma_{2}^{-1} + \|\mu_{2} - \mu_{1}\|^{2}_{\Sigma_{2}^{-1}}\right)~.
\end{equation*}
Since \(\calP_{y}\) and \(\calP_{x}\) are Gaussian distributions, we can obtain \(\KLdist(\calP_{y} ~\|~ \calP_{x})\) by substituting
\begin{equation*}
    \mu_{1} \leftarrow y - h \cdot \metric{y}^{-1}\gradpotential{y}~, \enskip \mu_{2} \leftarrow x - h \cdot \metric{x}^{-1}\gradpotential{x}~, \enskip \Sigma_{1} \leftarrow 2h \cdot \metric{y}~, \enskip \Sigma_{2} \leftarrow 2h \cdot \metric{x}^{-1}~,
\end{equation*}
in the above formula, which gives
\begin{align*}
    \KLdist(\calP_{y} ~\|~ \calP_{x}) &= \frac{1}{2} \cdot \left(\trace(\metric{x}\metric{y}^{-1} - \rmI_{d \times d}) - \log \det \metric{x}\metric{y}^{-1}\right) \\
    &\quad + \frac{1}{4h} \cdot \left\|(x - h \cdot \metric{x}^{-1}\gradpotential{x}) - (y - h \cdot \metric{y}^{-1}\gradpotential{y})\right\|^{2}_{\metric{x}}~\\
    &= \frac{1}{2} \cdot \underbrace{\left(\trace(\metric{y}^{-\nicefrac{1}{2}}\metric{x}\metric{y}^{-\nicefrac{1}{2}} - \rmI_{d \times d}) - \log \det \metric{y}^{-\nicefrac{1}{2}}\metric{x}\metric{y}^{-\nicefrac{1}{2}}\right)}_{T_{1}^{P}} \\
    &\qquad + \frac{1}{4h} \cdot \underbrace{\left\|(x - h \cdot \metric{x}^{-1}\gradpotential{x}) - (y - h \cdot \metric{y}^{-1}\gradpotential{y})\right\|_{\metric{x}}^{2}}_{T_{2}^{P}}~.
\end{align*}

In the lemma, the metric \(\metric{}\) is assumed to be self-concordant, and \(x, y\) are such that \(\|x - y\|_{\metric{y}} \leq \frac{\sqrt{h}}{10} < 1\).
Hence, from \Cref{lem:dikin1-in-domain-self-concordant-loewner}(2) we have

\begin{subequations}
\begin{align}
    (1 - \|x - y\|_{\metric{y}})^{2} \cdot \metric{y} &\preceq \metric{x} \preceq \frac{1}{(1 - \|x - y\|_{\metric{y}})^{2}} \cdot \metric{y} \label{eq:loewner-1}\\
    \Leftrightarrow 
    (1 - \|x - y\|_{\metric{y}})^{2} \cdot \rmI_{d \times d} \preceq \metric{y}^{-\nicefrac{1}{2}}&\metric{x}\metric{y}^{-\nicefrac{1}{2}} \preceq \frac{1}{(1 - \|x - y\|_{\metric{y}})^{2}} \cdot \rmI_{d \times d}~.\label{eq:loewner-2}
\end{align}
\end{subequations}

For convenience, \(M\) to denote \(\metric{y}^{-\nicefrac{1}{2}}\metric{x}\metric{y}^{-\nicefrac{1}{2}}\), and \(r_{y}\) to denote \(\|x - y\|_{\metric{y}}\).
\cref{eq:loewner-2} asserts that all eigenvalues of \(M\) lie in the range \([(1 - r_{y})^{2}, (1 - r_{y})^{-2}]\).
\cref{auxlem:log-inequality,auxlem:small-quad} imply that
\begin{equation*}
    \forall~i \in [d]~,\quad \lambda_{i}(M) - 1 - \log \lambda_{i}(M) \leq \frac{((1 - r_{y})^{2} - 1)^{2}}{(1 - r_{y})^{2}} \leq \frac{4r_{y}^{2}}{(1 - r_{y})^{2}}~.
\end{equation*}
This leads to a bound for \(T_{1}^{P}\) as follows when \(h \leq 1\).
\begin{align*}
    T_{1}^{P} = \trace(M - I) - \log \det M &= \sum_{i=1}^{d} (\lambda_{i}(M) - 1 - \log \lambda_{i}(M)) \\
    &\leq d \cdot \max_{i \in [d]} 
    ~(\lambda_{i}(M) - 1 - \log \lambda_{i}(M)) \\
    &\leq d \cdot \frac{4r_{y}^{2}}{(1 - r_{y})^{2}}  \\
    &\leq d \cdot \frac{4 \cdot \frac{h}{100}}{(1 - \frac{\sqrt{h}}{10})^{2}} \leq \frac{h \cdot d}{20}~.
\end{align*}

Using \cref{eq:loewner-1} and the fact that \(x, y \in \calS\), we also obtain a bound for \(T_{2}^{P}\) as shown below.
\begin{align*}
    T_{2}^{P} &= \|(x - h \cdot \metric{x}^{-1}\gradpotential{x}) - (y - h \cdot \metric{y}^{-1}\gradpotential{y})\|_{\metric{x}}^{2} \\
    &\leq 2 \cdot \|x - y\|_{\metric{x}}^{2} + 2h^{2} \cdot \|\metric{y}^{-1}\gradpotential{y} - \metric{x}^{-1}\gradpotential{x}\|_{\metric{x}}^{2} \\
    &\leq 2 \cdot \|x - y\|_{\metric{x}}^{2} + 4h^{2} \cdot \|\metric{x}^{-1}\gradpotential{x}\|_{\metric{x}}^{2} + 4h^{2} \cdot \|\metric{y}^{-1}\gradpotential{y}\|_{\metric{x}}^{2} \\
    &\leq \frac{2 \cdot \|x - y\|_{\metric{y}}^{2}}{(1 - r_{y})^{2}} + 4h^{2} \cdot \|\gradpotential{x}\|_{\metric{x}^{-1}}^{2} + 4h^{2} \cdot \frac{\|\gradpotential{y}\|^{2}_{\metric{y}^{-1}}}{(1 - r_{y})^{2}}  \\
    &\leq 2 \cdot \frac{r_{y}^{2}}{(1 - r_{y})^{2}} + 4h^{2} \cdot \sfU_{\potential{}, \calS}^{2} + \frac{4h^{2} \cdot \sfU_{\potential{}, \calS}^{2}}{(1 - r_{y})^{2}}~.
\end{align*}
Since \(r_{y} \leq \frac{\sqrt{h}}{10}\) for \(h \leq 1\), \(\frac{1}{(1 - r_{y})^{2}} \leq \frac{5}{4}\).
Hence, we have the net bound
\begin{equation*}
    \KLdist(\calP_{y} ~\|~ \calP_{x}) = \frac{1}{2} \cdot T_{1}^{P} + \frac{1}{4h} \cdot T_{2}^{P} = \frac{1}{2} \cdot \left(\frac{h \cdot d}{20} + \frac{1}{80} + \frac{9h \cdot \sfU_{\potential{}, \calS}^{2}}{2}\right)
\end{equation*}
which results in
\begin{equation*}
    \TVdist(\calP_{x}, \calP_{y}) \leq \frac{1}{2}\sqrt{\frac{h \cdot d}{20} + \frac{1}{80} + \frac{9h \cdot \sfU_{\potential{}, \calS}^{2}}{2}}~.
\end{equation*}
\end{proof}

\subsubsection{Proofs of \Cref{lem:tx-px-small-SC,lem:tx-px-small-beyond-SC}}

First, we state a collection of facts about metrics that satisfy various notions of self-concordance defined in \cref{sec:mixing-guarantees:prelims}.
These will come in handy in proving \cref{lem:tx-px-small-SC,lem:tx-px-small-beyond-SC}.

\begin{lemma}
\label{lem:SC-lemma}
Let the metric \(\metric{}\) be self-concordant.
For any \(x, y \in \interior{\primalspace}\) such that \(\|y - x\|_{\metric{x}} \leq \frac{3}{10}\),
\begin{align*}
    \log \det \metric{y} - \log \det \metric{x} &\geq -3d \cdot \|y - x\|_{\metric{x}}~, \\
    \|y - x\|_{\metric{x}}^{2} - \|y - x\|_{\metric{y}}^{2} &\geq -6 \cdot \|y - x\|_{\metric{x}}^{3}~.
\end{align*}
\end{lemma}

\begin{lemma}
\label{lem:SSC-LTSC-lemma}
Let the metric \(\metric{}\) be strongly self-concordant.
Then, for any \(x \in \interior{\primalspace}\),
\begin{equation*}
    \|\metric{x}^{-\nicefrac{1}{2}}\nabla \log \det \metric{x}\| \leq 2\sqrt{d}~.
\end{equation*}
Additionally, if \(\metric{}\) is also \(\alpha\)-lower trace self-concordant, then for any \(y \in \calE_{x}^{\metric{}}(1)\),
\begin{equation*}
    \log \det \metric{y}\metric{x}^{-1} \geq \langle \nabla \log \det \metric{x}, y - x\rangle - \frac{17}{8} \cdot (\alpha + 4) \cdot \|y - x\|_{\metric{x}}^{2}~.
\end{equation*}
\end{lemma}

\begin{lemma}
\label{lem:sq-dist-diff-ASC-SC}
Let the metric \(\metric{}\) be self-concordant, and \(x \in \interior{\primalspace}\).
Consider \(y, w \in \calE_{x}^{\metric{}}(1)\), and define \(\Delta(u; x) = \|u - x\|^{2}_{\metric{u}} - \|u - x\|^{2}_{\metric{x}}\).
There exists \(t^{\star} \in (0, 1)\) such that
\begin{equation*}
    \Delta(w; x) - \Delta(y;x) \leq 6 \cdot \frac{\|y + t^{\star}(w - y) - x\|_{\metric{x}}^{2} \cdot \|w - y\|_{\metric{x}}}{(1 - \|y + t^{\star}(w - y) - x\|_{\metric{x}})^{3}}~.
\end{equation*}
\end{lemma}

The proofs of \Cref{lem:tx-px-small-SC,lem:tx-px-small-beyond-SC} follow a similar structure, and we will state them as one to avoid redundancy, and highlight the point at which they branch out.

\begin{proof}
We adopt the shorthand notation \(R_{x \to z}\) for the ratio \(\frac{\targetdens(z)p_{z}(x)}{\targetdens(x)p_{x}(z)}\) for the convenience.
The transition distribution \(\calT_{x}\) has an atom at \(x\) and satisfies
\begin{equation*}
    \calT_{x}(\{x\}) = 1 - \bbE_{z \sim \calP_{x}}\left[\min\left\{1, R_{x \to z} \cdot \bm{1}\{z \in \primalspace\}\right\}\right]~.
\end{equation*}
Consequently, the TV distance \(\TVdist(\calT_{x}, \calP_{x})\) for any \(x \in \interior{\primalspace}\) can be given as
\begin{align*}
    \TVdist(\calT_{x}, \calP_{x}) &= \frac{1}{2}\left(\calT_{x}(\{x\}) + \int_{\bbR^{d} \setminus \{x\}} \left(1 - \min\left\{1, R_{x \to z} \cdot \bm{1}\{z \in \primalspace\}\right\}\right) p_{x}(z)~\rmd z \right) \\
    &= 1 - \underbrace{\bbE_{z \sim \calP_{x}}\left[\min\left\{1, R_{x \to z} \cdot \bm{1}\{z \in \primalspace\}\right\}\right]}_{\calQ}~.\numberthis\label{eq:prelim-ub-TV-px-tx}
\end{align*}
Thus, a lower bound for \(\calQ\) implies an upper bound for \(\TVdist(\calT_{x}, \calP_{x})\).
One strategy to obtain a lower bound is using Markov inequality; for any \(\tau \in (0, 1)\),
\begin{align*}
    \bbE_{z \sim \calP_{x}}\left[\min\left\{1, R_{x \to z} \cdot \bm{1}\{z \in \primalspace\}\right\}\right] &\geq \tau \cdot \bbP_{z \sim \calP_{x}}\left(R_{x \to z} \cdot \bm{1}\{z \in \primalspace\} \geq \tau\right) \\
    &\geq \tau \cdot \bbP_{z \sim \calP_{x}}\left(\left.R_{x \to z} \cdot \bm{1}\{z \in \primalspace\} \geq \tau ~\right|~ \frakE\right) \cdot \bbP_{z \sim \calP_{x}}(\frakE)~.\numberthis\label{eq:markov-q2}
\end{align*}
Therefore, it suffices to identify an event \(\frakE\) that satisfies three key desiderata: (1) implies that \(z \in \primalspace\), (2) yields a lower bound \(\tau\) that is bounded away from \(0\), and (3) occurs with high probability.
From the definition of \(\calP_{x}\), any \(z \sim \calP_{x}\) is distributionally equivalent to the random vector
\begin{equation*}
    x - h \cdot \metric{x}^{-1}\gradpotential{x} + \sqrt{2h} \cdot \metric{x}^{-\nicefrac{1}{2}}~\gamma, \quad \gamma \sim \calN(0, \rmI_{d \times d})~.
\end{equation*}
Due to this, we can view \(z\) as a function of \(\gamma\), and henceforth we consider the underlying random variable to be \(\gamma\), and the event \(\frakE\) is defined in terms of this \(\gamma\).
Before defining \(\frakE\), we first provide a simplified lower bound for \(R_{x \to z}\) using certain properties of the potential \(\potential{}\) assumed in the lemmas.
\begin{align*}
    \log R_{x \to z} &= \overbrace{f(x) - f(z)}^{T_{F}} + \frac{1}{2}\log\det \metric{z}\metric{x}^{-1} \\
    & \qquad +~\underbrace{\frac{\|z - x + h \cdot \metric{x}^{-1}\gradpotential{x}\|_{\metric{x}}^{2} - \|x - z + h \cdot \metric{z}^{-1}\gradpotential{z}\|_{\metric{z}}^{2}}{4h}}_{T_{D}}~.
\end{align*}
First, we work with \(T_{F}\). We have \(z \in \primalspace\), otherwise these calculations would be irrelevant.
Let \(x_{t} = x + t \cdot (z - x)\) for \(t \in [0, 1]\).
\begin{align*}
    T_{F} :&= f(x) - f(z) \\
    &= \frac{1}{2}(f(x) - f(z)) + \frac{1}{2}(f(x) - f(z)) \\
    &\overset{(a)}\geq \frac{1}{2} \cdot \langle \gradpotential{z}, x- z\rangle - \frac{1}{2} \cdot \langle \gradpotential{x}, z - x\rangle - \frac{1}{4} \cdot \langle z - x, \nabla^{2}\potential{x_{t^{\star}}}(z - x)\rangle \\
    &\overset{(b)}\geq \frac{1}{2} \cdot \langle \gradpotential{z}, x- z\rangle - \frac{1}{2} \cdot \langle \gradpotential{x}, z - x\rangle - \frac{\lambda}{4} \cdot \langle z - x, \metric{x_{t^{\star}}}(z - x)\rangle~.
\end{align*}
Inequality \((a)\) uses the fact that \(\potential{}\) is convex, and consider a second-order Taylor expansion of \(\potential{}\) around \(x\).
Inequality \((b)\) uses the fact that \(\potential{}\) satisfies \((\lambda, \metric{})\)-curvature upper bound.
Next for \(T_{D}\),
\begin{align*}
    T_{D} :&= \frac{1}{4h} \cdot \left(\|z - x + h \cdot \metric{x}^{-1}\gradpotential{x}\|_{\metric{x}}^{2} - \|x - z + h \cdot \metric{z}^{-1}\gradpotential{z}\|_{\metric{z}}^{2}\right) \\
    &= \frac{1}{4h} \cdot \left(\|z - x\|_{\metric{x}}^{2} - \|z - x\|_{\metric{z}}^{2} + 2h \cdot \langle z - x, \gradpotential{x}\rangle + 2h \cdot \langle z - x, \gradpotential{z}\rangle\right. \\
    &\qquad \qquad \left.+~h^{2} \cdot \|\gradpotential{x}\|_{\metric{x}^{-1}}^{2} - h^{2} \cdot \|\gradpotential{z}\|_{\metric{z}^{-1}}^{2} \right) \\
    &\overset{(a)}\geq \frac{\|z - x\|_{\metric{x}}^{2} - \|z - x\|_{\metric{z}}^{2}}{4h} - \frac{h}{4} \cdot \|\gradpotential{z}\|_{\metric{z}^{-1}}^{2} \\
    &\qquad \qquad + \frac{1}{2} \cdot \langle z - x, \gradpotential{x}\rangle + \frac{1}{2} \cdot \langle z - x, \gradpotential{z}\rangle~.
\end{align*}
Inequality \((a)\) simply uses the fact that \(\|\gradpotential{x}\|_{\metric{x}^{-1}}^{2} \geq 0\)~.
This results in the lower bound
\begin{multline}
\label{eq:log-accept-ratio-initial-lower}
\log R_{x \to z} \geq
- \frac{h}{4} \cdot \underbrace{\|\gradpotential{z}\|_{\metric{z}^{-1}}^{2}}_{T_{0}^{A}} - \frac{\lambda}{4} \cdot \underbrace{\|z - x\|_{\metric{x_{t^{\star}}}}^{2}}_{T_{1}^{A}} \\+ \frac{1}{2} \cdot \underbrace{\log \det \metric{z}\metric{x}^{-1}}_{T_{2}^{A}} + \frac{1}{4h} \cdot \underbrace{\left(\|z - x\|_{\metric{x}}^{2} - \|z - x\|_{\metric{z}}^{2}\right)}_{T_{3}^{A}}~.
\end{multline}

For \(\gamma \sim \gaussian{\bm{0}}{\rmI_{d \times d}}\), define \(\xi  = x + \sqrt{2h} \cdot \metric{x}^{-\nicefrac{1}{2}}\gamma\).
The proposal \(z \sim \calP_{x}\) is distributionally equivalent to \(\xi - h \cdot \metric{x}^{-1}\gradpotential{x}\).
For \(\varepsilon > 0\), define
\begin{equation}
\label{eq:Neps-Ieps-def}
    \sfN_{\varepsilon} = 1 + 2\log \frac{1}{\varepsilon} + 2\sqrt{\log \frac{1}{\varepsilon}} \qquad \sfI_{\varepsilon} = \sqrt{2 \log \frac{1}{\varepsilon}}~
\end{equation}
and the events
\begin{center}
\begin{tabular}{c@{\hskip 1.2cm}c}
\(\frakE_{1} := \|\gamma\|^{2} \leq d \cdot \sfN_{\varepsilon}\) & \(\frakE_{2} := -\langle \gamma, \metric{x}^{-\nicefrac{1}{2}}\gradpotential{x}\rangle \leq \sfU_{\potential{}, \calS} \cdot \sfI_{\varepsilon}\) \\
[1mm] \\
\(\frakE_{3} := -\langle \gamma, \metric{x}^{-\nicefrac{1}{2}}\nabla \log \det \metric{x}\rangle \leq 2\sqrt{d} \cdot \sfI_{\varepsilon}\) & \(\frakE_{4} := \|\xi - x\|_{\metric{\xi}}^{2} - \|\xi - x\|_{\metric{x}}^{2} \leq 4h \cdot \varepsilon\)
\end{tabular}
\end{center}
Let \(\gamma \sim \gaussian{\bm{0}}{\rmI_{d \times d}}\).
Then, we have the following facts about \(\frakE_{i}\) for \(i \in [4]\).
\begin{itemize}[itemsep=0pt,leftmargin=*]
\item By \cref{auxlem:chisquared}, \(\bbP(\frakE_{1}) \geq 1 - \varepsilon\).
\item Since \(x \in \calS\) according to the statements of the lemmas,
\begin{equation*}
\|\metric{x}^{-\nicefrac{1}{2}}\gradpotential{x}\| = \|\gradpotential{x}\|_{\metric{x}^{-1}} \leq \sfU_{\potential{}, \calS} \Rightarrow \bbP(\frakE_{2}) \geq 1- \varepsilon
\end{equation*}
and the implication is due to \cref{auxlem:gaussian-inner}.
\item If \(\metric{}\) is strongly self-concordant (as assumed in \cref{lem:tx-px-small-beyond-SC}), then from \cref{lem:SSC-LTSC-lemma},
\begin{equation*}
    \|\metric{x}^{-\nicefrac{1}{2}}\nabla \log \det \metric{x}\| \leq 2\sqrt{d} \Rightarrow \bbP(\frakE_{3}) \geq 1 - \varepsilon
\end{equation*}
where the implication again follows from \cref{auxlem:gaussian-inner}.
\item When \(\metric{}\) is average self-concordant and \(h \leq \frac{r_{\varepsilon}^{2}}{2d}\), then by definition \(\bbP(\frakE_{4}) \geq 1 - \varepsilon\).
\end{itemize}

Conditioning on \(\frakE_{1}\) and \(\frakE_{2}\), we get
\begin{align*}
    \|z - x\|_{\metric{x}}^{2} &= \|x -h \cdot \metric{x}^{-1}\gradpotential{x} + \sqrt{2h} \cdot \metric{x}^{-\nicefrac{1}{2}}\gamma - x\|_{\metric{x}}^{2} \\
    &= h^{2} \cdot \|\gradpotential{x}\|_{\metric{x}^{-1}}^{2} - (2h)^{\nicefrac{3}{2}} \cdot \langle \metric{x}^{-\nicefrac{1}{2}}\gradpotential{x}, \gamma\rangle + 2h \cdot \|\gamma\|^{2} \\
    &\leq h^{2} \cdot \sfU_{\potential{}, \calS}^{2} + (2h)^{\nicefrac{3}{2}} \cdot \sfU_{\potential{}, \calS} \cdot \sfI_{\varepsilon} + 2h \cdot d \cdot \sfN_{\varepsilon}~\numberthis\label{eq:z-minus-x-local-sq}.
\end{align*}
When the step size \(h\) satisfies the bound
\begin{equation*}
h \leq \min\left\{\frac{1}{6\sfU_{\potential{}, \calS}},~\frac{3}{200\sfN_{\varepsilon} \cdot d},~\frac{1}{25\sfI_{\varepsilon}^{\nicefrac{2}{3}}\sfU_{\potential{}, \calS}^{\nicefrac{2}{3}}}\right\}\tag{\textsf{B\textsubscript{0}}}\label{eq:B0-bound}~,
\end{equation*}
we have \(\|z - x\|_{\metric{x}} \leq \frac{3}{10}\) which implies \(z \in \calE_{x}^{\metric{}}(1)\).
In the setting of both \cref{lem:tx-px-small-SC,lem:tx-px-small-beyond-SC}, the metric \(\metric{}\) is self-concordant, and by \Cref{lem:dikin1-in-domain-self-concordant-loewner}(1).
Moreover, this also shows that \(x_{t}\) also satisfies \(\|x_{t} - x\|_{\metric{x}} = t \cdot \|z - x\|_{\metric{x}} \leq \frac{3}{10}\) for \(t \in [0, 1]\).

All of the calculations henceforth in this proof are performed when conditioning on \(\frakE_{1}\) and \(\frakE_{2}\), and that the step size satisfies the bound in \ref{eq:B0-bound}, as this ensures that \(\|z - x\|_{\metric{x}} \leq \frac{3}{10}\).
Now, we work from \Cref{eq:log-accept-ratio-initial-lower} to find a suitable \(\tau\) in \Cref{eq:markov-q2}.

For \(T_{0}^{A}\), we directly have from the definition of \(\calS\) that
\begin{equation*}
    T_{0}^{A} = \|\gradpotential{z}\|_{\metric{z}^{-1}}^{2} \leq \sup_{z:~\|z - x\|_{\metric{x}} \leq \frac{1}{2}} \|\gradpotential{z}\|^{2}_{\metric{z}^{-1}} = \sfU_{\potential{}, \calS}^{2}~.
\end{equation*}
For \(T_{1}^{A}\), we use the self-concordance of the metric \(\metric{}\) (\cref{lem:dikin1-in-domain-self-concordant-loewner}(2)) to obtain
\begin{align*}
    T_{1}^{A} &= \|z - x\|_{\metric{x_{t^{\star}}}}^{2} \\
    &\leq \frac{1}{(1 - \|z - x\|_{\metric{x}})^{2}} \cdot \|z - x\|_{\metric{x}}^{2} \\
    &\leq \frac{17}{8} \cdot (h^{2} \cdot \sfU_{\potential{}, \calS}^{2} + (2h)^{\nicefrac{3}{2}} \cdot \sfU_{\potential{}, \calS} \cdot \sfI_{\varepsilon} + 2h \cdot d \cdot \sfN_{\varepsilon})~.
\end{align*}
The last inequality uses \cref{eq:z-minus-x-local-sq} and that \(\|z - x\|_{\metric{x}} \leq \frac{3}{10}\).
For \(T_{2}^{A}\) and \(T_{3}^{A}\) in \cref{eq:log-accept-ratio-initial-lower}, we use different properties of metric corresponding to the settings of \cref{lem:tx-px-small-SC,lem:tx-px-small-beyond-SC}.

\paragraph{When \(\metric{}\) is self-concordant (\Cref{lem:tx-px-small-SC})}

Since \(\|z - x\|_{\metric{x}} \leq \frac{3}{10}\), we use \cref{lem:SC-lemma} to get
\begin{align*}
    T_{2}^{A} &= \log \det \metric{z} - \log \det \metric{x} \\
    &\geq -3d \cdot \|x - z\|_{\metric{x}} \\
    &= -3d \cdot \sqrt{h^{2} \cdot \sfU_{\potential{}, \calS}^{2} + 2h \cdot d \cdot \sfN_{\varepsilon} + (2h)^{\nicefrac{3}{2}} \cdot \sfU_{\potential{}, \calS} \cdot \sfI_{\varepsilon}} \\
    &\overset{(a)}\geq -3 \cdot \left(h \cdot d \cdot \sfU_{\potential{}, \calS} + \sqrt{2h \cdot d^{3} \cdot \sfN_{\varepsilon}} + \sqrt{(2h)^{\nicefrac{3}{2}} \cdot d^{2} \cdot \sfU_{\potential{}, \calS} \cdot \sfI_{\varepsilon}}\right) \\
    \\
    T_{3}^{A} &= \|z - x\|_{\metric{x}}^{2} - \|z - x\|_{\metric{z}}^{2} \\
    &\geq -6 \cdot \|x - z\|_{\metric{x}}^{3} \\
    &= -6 \cdot \left(h^{2} \cdot \sfU_{\potential{}, \calS}^{2} + 2h \cdot d \cdot \sfN_{\varepsilon} + (2h)^{\nicefrac{3}{2}} \cdot \sfU_{\potential{}, \calS} \cdot \sfI_{\varepsilon}\right)^{\nicefrac{3}{2}} \\
    &\overset{(b)}\geq -6\sqrt{3} \cdot \left(h^{3} \cdot \sfU_{\potential{}, \calS}^{3} + (2h \cdot d \cdot \sfN_{\varepsilon})^{\nicefrac{3}{2}} + ((2h)^{\nicefrac{3}{2}} \cdot \sfU_{\potential{}, \calS} \cdot \sfI_{\varepsilon})^{\nicefrac{3}{2}} \right)~.
\end{align*}
Inequality \((a)\) uses the fact that \(\sqrt{a + b + c} \leq \sqrt{a} + \sqrt{b} + \sqrt{c}\) for \(a, b, c \geq 0\), and inequality \((b)\) uses the convexity of \(t \mapsto t^{\nicefrac{3}{2}}\).
Collecting the bounds for each \(T_{i}^{A}\), \(i \in \{0, 1, 2, 3\}\), we can give a lower bound for \(\log R_{x \to z}\).
Recall that this holds when conditioned on \(\frakE_{1}\) and \(\frakE_{2}\), and when the step size \(h\) satisfies the bound \ref{eq:B0-bound}.
\begin{align*}
    \log R_{x \to z} &= -\frac{h}{4} \cdot T_{0}^{A} - \frac{17\lambda}{32} \cdot T_{1}^{A} + \frac{1}{2} \cdot T_{2}^{A} + \frac{1}{4h} \cdot T_{3}^{A} \\
    &\geq - \frac{h \cdot \sfU_{\potential{}, \calS}^{2}}{4} - \frac{\lambda}{4} \cdot (h^{2} \cdot \sfU_{\potential{}, \calS}^{2} + 2h \cdot d \cdot \sfN_{\varepsilon} + (2h)^{\nicefrac{3}{2}} \cdot \sfU_{\potential{}, \calS} \cdot \sfI_{\varepsilon}) \\
    &\quad -\frac{3}{2} \cdot \left(h \cdot d \cdot \sfU_{\potential{}, \calS} + \sqrt{2h \cdot d^{3} \cdot \sfN_{\varepsilon}} + \sqrt{(2h)^{\nicefrac{3}{2}} \cdot d^{2} \cdot \sfU_{\potential{}, \calS} \cdot \sfI_{\varepsilon}}\right) \\
    &\qquad -\frac{3\sqrt{3}}{2} \cdot \left(h^{2} \cdot \sfU_{\potential{}, \calS}^{3} + (2\sfN_{\varepsilon})^{\nicefrac{3}{2}} \cdot \sqrt{h \cdot d^{3}} + (2^{\nicefrac{3}{2}}\sfI_{\varepsilon})^{\nicefrac{3}{2}} \cdot (h^{5} \cdot \sfU_{\potential{}, \calS}^{6})^{\nicefrac{1}{4}}\right)~.\numberthis\label{eq:SC-primitive-lower}
\end{align*}

Let \(C_{1}(\varepsilon) = \frac{\varepsilon^{2}}{486 \cdot \sfN_{\varepsilon}^{3}}\)~.
Define the function \(b_{1}(d, \sfU_{\potential{}, \calS}, \lambda)\) as
\begin{equation*}
b_{1}(d, \sfU_{\potential{}, \calS}, \lambda) \defeq C_{1}(\varepsilon) \cdot \min\left\{\frac{1}{\sfU_{\potential{}, \calS}^{2}}, ~\frac{1}{\sfU_{\potential{}, \calS}^{\nicefrac{2}{3}}}, ~\frac{1}{(\sfU_{\potential{}, \calS} \cdot \lambda)^{\nicefrac{2}{3}}}, ~\frac{1}{d^{3}}, ~\frac{1}{d \cdot \lambda}\right\}~.
\end{equation*}
If the step size satisfies \(h \leq b_{1}(d, \sfU_{\potential{}, \calS}, \lambda)\), then \ref{eq:B0-bound} holds since
\begin{equation*}
    C_{1}(\varepsilon) \leq \min\left\{\frac{1}{6},~\frac{3}{200\sfN_{\varepsilon}},~\frac{1}{25\sfI_{\varepsilon}^{\nicefrac{2}{3}}}\right\}~.
\end{equation*}
Moreover, when \(h \leq b_{1}(d, \sfU_{\potential{}, \calS}, \lambda)\), we have a lower bound for \(\log R_{x \to z}\) in terms of \(\varepsilon\) alone.
From \cref{auxlem:b1-implication} with \(x \leftarrow h, V \leftarrow C_{1}(\varepsilon), a \leftarrow \sfU_{\potential{}, \calS}, b \leftarrow \lambda\), we have 
\begin{equation*}
    \log R_{x \to z} \geq -4\varepsilon~. 
\end{equation*}

Therefore, choosing \(\tau = e^{-4\varepsilon}\) and \(\frakE = \frakE_{1} \cap \frakE_{2}\), we have for \(x \in \calS\) using \cref{eq:markov-q2} that
\begin{align*}
   \bbE_{z \sim \calP_{x}}[\min\{1, R_{x \to z} \cdot \bm{1}\{z \in \primalspace\}\}] &\geq e^{-4\varepsilon} \cdot (1 - 2\varepsilon) \\
    &\geq (1 - 4\varepsilon) \cdot (1 - 2\varepsilon) \geq 1 - 6\varepsilon~.
\end{align*}
From \cref{eq:prelim-ub-TV-px-tx}, we get \(\TVdist(\calT_{x}, \calP_{x}) \leq 1 - (1 - 6\varepsilon) = 6\varepsilon\).
We set \(\varepsilon = \frac{1}{96}\), and note that \(C_{1}(\nicefrac{1}{96}) \leq \frac{1}{20}\) to complete the proof of \cref{lem:tx-px-small-SC}.

\paragraph{When \(\metric{}\) is self-concordant\textsubscript{++} (\Cref{lem:tx-px-small-beyond-SC})}

Recall that self-concordant\textsubscript{++} is equivalent to strongly, \(\alpha\)-lower trace, and average self-concordant.
Since \(\|z - x\|_{\metric{x}} \leq \frac{3}{10}\), we have a lower bound for \(T_{2}^{A}\) through \cref{lem:SSC-LTSC-lemma}, and the definition of \(z\) in terms of \(\gamma \sim \gaussian{\bm{0}}{\rmI_{d \times d}}\).
\begin{align*}
    T_{2}^{A} &= \log \det \metric{z} - \log \det \metric{x} \\
    &\geq \langle \nabla \log \det \metric{x}, z - x\rangle - (\alpha + 4) \cdot \frac{\|z - x\|^{2}_{\metric{x}}}{(1 - \|z - x\|_{\metric{x}})^{2}} \\
    &= -h \cdot \langle \nabla \log \det \metric{x}, \metric{x}^{-1}\gradpotential{x}\rangle + \sqrt{2h} \cdot \langle \nabla \log \det \metric{x}, \metric{x}^{-\nicefrac{1}{2}}\gamma\rangle \\
    &\qquad - (\alpha + 4) \cdot \frac{\|z - x\|^{2}_{\metric{x}}}{(1 - \|z - x\|_{\metric{x}})^{2}} \\
    &\geq -h \cdot \|\metric{x}^{-\nicefrac{1}{2}}\nabla \log \det \metric{x}\| \cdot \|\metric{x}^{-\nicefrac{1}{2}}\gradpotential{x}\| + \sqrt{2h} \cdot \langle \nabla \log \det \metric{x}, \metric{x}^{-\nicefrac{1}{2}}\gamma\rangle \\
    &\qquad -\frac{17}{8} \cdot (\alpha + 4) \cdot \|z - x\|_{\metric{x}}^{2}~.
\end{align*}
The final inequality uses the Cauchy-Schwarz inequality.
Additionally, by conditioning on \(\frakE_{3}\) and using the fact that \(x \in \calS\),
\begin{equation*}
    T_{2}^{A} \geq -2h \cdot \sqrt{d} \cdot \sfU_{\potential{}, \calS} - 2\sqrt{2h} \cdot \sqrt{d} \cdot \sfI_{\varepsilon} - \frac{17}{8} \cdot (\alpha + 4) \cdot (h^{2} \cdot \sfU_{\potential{},\calS}^{2} + 2h \cdot d \cdot \sfN_{\varepsilon} + (2h)^{\nicefrac{3}{2}} \cdot \sfU_{\potential{}, \calS} \cdot \sfI_{\varepsilon})~.
\end{equation*}

For \(T_{3}^{A}\), we first note that a Dikin proposal \(\xi = x + \sqrt{2h} \cdot \metric{x}^{-\nicefrac{1}{2}}\gamma\) satisfies
\begin{equation*}
    \|\xi - x\|_{\metric{x}}^{2} = \|\sqrt{2h} \cdot \metric{x}^{-\nicefrac{1}{2}}\gamma\|^{2}_{\metric{x}} = 2h \cdot \|\gamma\|^{2}~.
\end{equation*}
Since the calculations here are considered when conditioning on \(\frakE_{1}\) and \(\frakE_{2}\), and when the step size satisfies the bound \ref{eq:B0-bound},
\begin{equation*}
    \|\xi - x\|_{\metric{x}}^{2} \leq 2h \cdot d \cdot \sfN_{\varepsilon} \leq \frac{6}{200}~,
\end{equation*}
which implies \(\|\xi - x\|_{\metric{x}} \leq \frac{3}{10}\).
\cref{lem:sq-dist-diff-ASC-SC} with \(w \leftarrow z\) and \(y \leftarrow \xi\) states that there exists \(\bar{t} \in (0, 1)\) such that
\begin{align*}
    T_{3}^{A} &= \|z - x\|_{\metric{x}}^{2} - \|z - x\|_{\metric{z}}^{2} \\
    &\geq \|\xi - x\|_{\metric{x}}^{2} - \|\xi - x\|_{\metric{\xi}}^{2} - 6 \cdot \frac{\|\xi + \bar{t}(z - \xi) - x\|_{\metric{x}}^{2} \cdot \|z - \xi\|_{\metric{x}}}{(1 - \|\xi + \bar{t}(z - \xi) - x\|_{\metric{x}})^{3}}~.
\end{align*}
Note that \(z - \xi = -h \cdot \metric{x}^{-1}\gradpotential{x}\), and therefore
\begin{align*}
    \|\xi + \bar{t}(z - \xi) - x\|_{\metric{x}}^{2} &\leq \max\{\|z - x\|_{\metric{x}}^{2}, \|\xi - x\|_{\metric{x}}^{2}\} \\
    &\leq h^{2} \cdot \sfU_{\potential{}, \calS}^{2} + 2h \cdot d \cdot \sfN_{\varepsilon} + (2h)^{\nicefrac{3}{2}} \cdot \sfU_{\potential{}, \calS} \cdot \sfI_{\varepsilon} \\
    \|z - \xi\|_{\metric{x}} &\leq h \cdot \sfU_{\potential{}, \calS}~.
\end{align*}
When conditioning on the event \(\frakE_{4}\) and assuming that the step size \(h\) additionally satisfies \(h \leq \frac{r_{\varepsilon}^{2}}{2d}\) for \(r_{\varepsilon}\) in the definition of average self-concordance
\begin{equation*}
    T_{3}^{A} \geq -18h \cdot \sfU_{\potential{}, \calS} \cdot (h^{2} \cdot \sfU_{\potential{}, \calS}^{2} + 2h \cdot d \cdot \sfN_{\varepsilon} + (2h)^{\nicefrac{3}{2}} \cdot \sfU_{\potential{}, \calS} \cdot \sfI_{\varepsilon}) -4h \cdot \varepsilon~.
\end{equation*}

In summary, let the step size \(h\) satisfy the bound \ref{eq:B0-bound} and \(h \leq \frac{r_{\varepsilon}^{2}}{2d}\).
When conditioned on the events \(\frakE_{i}\) for \(i \in [4]\), then continuing from \cref{eq:log-accept-ratio-initial-lower}
\begin{align*}
    \log R_{x \to z} &\geq -\frac{h}{4} \cdot T_{0}^{A} - \frac{\lambda}{4} \cdot T_{1}^{A} + \frac{1}{2} \cdot T_{2}^{A} + \frac{1}{4h} \cdot T_{3}^{A} \\
    &\geq -\frac{h \cdot \sfU_{\potential{}, \calS}^{2}}{4} - \frac{17\lambda}{32} \cdot (h^{2} \cdot \sfU_{\potential{}, \calS}^{2} + 2h \cdot d \cdot \sfN_{\varepsilon} + (2h)^{\nicefrac{3}{2}} \cdot \sfU_{\potential{}, \calS} \cdot \sfI_{\varepsilon}) \\
    &\quad -h \cdot \sqrt{d} \cdot \sfU_{\potential{}, \calS} - \sqrt{2h \cdot d \cdot \sfI_{\varepsilon}^{2}} - \frac{17}{16} \cdot (\alpha + 4) \cdot (h^{2} \cdot \sfU_{\potential{}, \calS}^{2} + 2h \cdot d \cdot \sfN_{\varepsilon} + (2h)^{\nicefrac{3}{2}} \cdot \sfU_{\potential{}, \calS} \cdot \sfI_{\varepsilon}) \\
    &\qquad -\varepsilon - \frac{9}{2} \cdot (h^{2} \cdot \sfU_{\potential{}, \calS}^{3} + 2h \cdot d \cdot \sfU_{\potential{}, \calS} \cdot \sfN_{\varepsilon} + (2h)^{\nicefrac{3}{2}} \cdot \sfU_{\potential{}, \calS}^{2} \cdot \sfI_{\varepsilon})~.\numberthis\label{eq:SC++-primitive-lower}
\end{align*}

Let \(C_{2}(\varepsilon) = \min\left\{\frac{\varepsilon^{2}}{100\cdot \sfI_{\varepsilon}^{2}}, r_{\varepsilon}^{2}\right\}\).
Define the function \(b_{2}(d, \sfU_{\potential{}, \calS}, \lambda, \alpha)\) as 
\begin{equation*}
b_{2}(d, \sfU_{\potential{}, \calS}, \lambda, \alpha) \defeq C_{2}(\varepsilon) \cdot \min\left\{\frac{1}{d \cdot (\alpha + 4)},\frac{1}{d \cdot \sfU_{\potential{}, \calS}},\frac{1}{d \cdot \lambda},\frac{1}{(\lambda \cdot \sfU_{\potential{}, \calS})^{\nicefrac{2}{3}}},\frac{1}{(\sfU_{\potential{}, \calS} \cdot (\alpha + 4))^{\nicefrac{2}{3}}},\frac{1}{\sfU_{\potential{}, \calS}^{2}}\right\}~.
\end{equation*}
If \(h \leq b_{2}(d, \sfU_{\potential{}, \calS}, \lambda, \alpha)\), then \ref{eq:B0-bound} holds since \(\alpha \geq 0\) and
\begin{equation*}
    C_{2}(\varepsilon) \leq \min\left\{\frac{1}{6}, \frac{3}{200\sfN_{\varepsilon}}, \frac{1}{25\sfI_{\varepsilon}^{\nicefrac{2}{3}}}\right\}~.
\end{equation*}
Also, when \(h \leq b_{2}(d, \sfU_{\potential{}, \calS}, \lambda, \alpha)\), we have a lower bound for \(\log R_{x \to z}\) solely in terms of \(\varepsilon\).
From \cref{auxlem:b2-implication} with \(x \leftarrow h, V \leftarrow C_{1}(\varepsilon), a \leftarrow \sfU_{\potential{}, \calS}, b \leftarrow \lambda, c \leftarrow (\alpha + 4)\), we have 
\begin{equation*}
    \log R_{x \to z} \geq -2\varepsilon~. 
\end{equation*}
Therefore, choosing \(\tau = e^{-2\varepsilon}\) and \(\frakE = \bigcap_{i=1}^{4}\frakE_{i}\), we have for \(x \in \calS\) using \cref{eq:markov-q2} that
\begin{align*}
    \bbE_{z \sim \calP_{x}}[\min\{1, R_{x \to z} \cdot \bm{1}\{z \in \primalspace\}\}] &\geq e^{-2\varepsilon} \cdot (1 - 4\varepsilon) \\
    &\geq (1 - 2\varepsilon) \cdot (1 - 4\varepsilon) \geq 1 - 6\varepsilon~.
\end{align*}
From \cref{eq:prelim-ub-TV-px-tx}, we get \(\TVdist(\calT_{x}, \calP_{x}) \leq 1 - (1 - 6\varepsilon) = 6\varepsilon\).
We set \(\varepsilon = \frac{1}{96}\), and note that \(C_{2}(\nicefrac{1}{96}) \leq \frac{1}{20}\) to complete the proof of \cref{lem:tx-px-small-beyond-SC}.
\end{proof}

\subsection{Proofs of isoperimetry lemmas in \Cref{sec:proofs:proofs-of-key-lemmas}}
\label{sec:proofs:proofs-isoperimetric-lemmas}

Here, we give the proofs of \cref{lem:log-concave-isoperimetry,lem:curv-lower-bdd-isoperimetry} which were used to give lower bounds on the conductance of the reversible Markov chain \(\bfT\) formed by an iteration of \nameref{alg:mapla}.

\subsubsection{Proof of \Cref{lem:log-concave-isoperimetry}}

The proof of this lemma follows the proof of \citet[Lem. B.6]{kook2024gaussian}, with care given to restrictions of the distribution to subsets of the support.

\begin{proof}
Let \(\bbB(0, r)\) be a ball of radius \(r > 0\) centered at the the origin, and define the set \(\calD_{|r} \defeq \calD \cap \bbB(0, r)\) for a set \(\calD \subseteq \primalspace\).
Consider the restriction \(\targetdist_{\calS_{|r}}\) of \(\targetdist\) over \(\calS_{|r}\).
Since \(\calS\) is convex, \(\targetdist_{\calS_{|r}}\) is log-concave.
For any partition \(\{A_{1}, A_{2}, A_{3}\}\) of \(\calS\), \(\{{A_{1}}_{|r}, {A_{2}}_{|r}, {A_{3}}_{|r}\}\) forms a partition of \(\calS_{|r}\), and by \citet[Thm. 2.2]{lovasz2003hit}, we have
\begin{equation*}
    \targetdist_{\calS_{|r}}({A_{3}}_{|r}) \geq \crossratio{{A_{2}}_{|r}}{{A_{1}}_{|r}}{\calS_{|r}} \cdot \targetdist_{\calS_{|r}}({A_{2}}_{|r}) \cdot \targetdist_{\calS_{|r}}({A_{1}}_{|r})~.
\end{equation*}
Above, \(\crossratio{y}{x}{\calD}\) is the cross-ratio between \(x\) and \(y\) in \(\calD\), and is defined as
\begin{equation*}
    \crossratio{y}{x}{\calD} = \frac{\|y - x\| \cdot \|p - q\|}{\|y - p\| \cdot \|x - q\|}.
\end{equation*}
where \(p\) and \(q\) are the end points of the extensions of the line segment between \(y\) and \(x\) to the boundary of \(\calD\) respectively.
The cross ratio between sets is defined as
\[\crossratio{\calD_{1}}{\calD_{2}}{\calD} = \inf_{y \in \calD_{1}, x \in \calD_{2}} \crossratio{y}{x}{\calD}~.\]
Since \(\metric{}\) is \(\nu\)-symmetric, from \citet[Lem. 2.3]{laddha2020strong}, we have for any \(x \in {A_{1}}_{|r}, y \in {A_{2}}_{|r}\) that
\begin{equation*}
    \crossratio{y}{x}{\calS_{|r}} \geq \frac{\|y - x\|_{\metric{y}}}{\sqrt{\nu}} \Rightarrow \crossratio{{A_{2}}_{|r}}{{A_{1}}_{|r}}{\calS_{|r}} \geq \inf_{y \in {A_{2}}_{|r}, x \in {A_{1}}_{|r}} \frac{\|y - x\|_{\metric{y}}}{\sqrt{\nu}}~.
\end{equation*}
Since \({A_{i}}_{|r} \subseteq A_{i}\), the lower bound above is at least \(\inf_{y \in A_{2}, x \in A_{1}} \frac{\|x - y\|_{\metric{y}}}{\sqrt{\nu}}\).
By setting \(r \to \infty\), the dominated convergence theorem implies that
\begin{equation*}
    \targetdist_{\calS}(A_{3}) \geq \inf_{y \in A_{2}, x \in A_{1}} \frac{\|y - x\|_{\metric{y}}}{\sqrt{\nu}} \cdot \targetdist_{\calS}(A_{2}) \cdot \targetdist_{\calS}(A_{1})
\end{equation*}
which is the statement of the lemma.
\end{proof}

\subsubsection{Proof of \Cref{lem:curv-lower-bdd-isoperimetry}}

We first state a one-dimensional inequality that the proof of \cref{lem:curv-lower-bdd-isoperimetry} we give relies on.
The proof of the below one-dimensional inequality is given in \cref{sec:proofs:proof-one-dim-isoperimetry}.

\begin{lemma}
\label{lem:one-dim-isoperimetry}
Consider a differentiable function \(\frakg : \bbR \to (0, \infty)\).
Assume that there exists \(\kappa > 0\) such that for all \(x \in \bbR\), \(|\frakg'(x)| \leq \frac{2}{\kappa} \cdot \frakg(x)^{\nicefrac{3}{2}}\).
If a twice differentiable function \(V : \bbR \to \bbR\) satisfies \((1,\frakg)\)-curvature lower bound, then for all \(x \in \bbR\),
\begin{equation*}
    \frac{\exp(-V(x))}{\sqrt{\frakg(x)}} \geq C_{\kappa} \cdot \min\left\{\int_{-\infty}^{x}\exp(-V(t))\rmd t, \int_{x}^{\infty} \exp(-V(t))\rmd t\right\}
\end{equation*}
where \(C_{\kappa} = \frac{\kappa}{8 + 4\kappa}\).
\end{lemma}

\begin{proof}
This proof follows the structure of the proofs of \citet[Lem. 9]{gopi2023algorithmic} and \citet[Lem. B.7]{kook2024gaussian}, and we give a few more details for the convenience of the reader.
We begin by first noting that \(\potential{}\) is \(1\)-relatively convex with respect to itself.
Let the extended version of \(\potential{}\) be \(\tilde{\potential{}} : \bbR^{d} \to \bbR\) which is defined as
\begin{equation*}
    \tilde{\potential{}}(x) = \begin{cases} \potential{x} & x \in \interior{\primalspace} \\
    \infty & x \not\in \interior{\primalspace} \end{cases}
\end{equation*}
and the associated extended density \(\tilde{\targetdens} \propto \exp(-\tilde{\potential{}})\).
For positive functions \(f_{1}, f_{2}, f_{3}, f_{4}\) where \(f_{1}\) and \(f_{2}\) are upper semicontinuous, and \(f_{3}\) and \(f_{4}\) are lower continuous, we have the following equivalence  from \citet[Lem. 8]{gopi2023algorithmic}:
\begin{gather*}
    \left(\int f_{1}(x)\tilde{\targetdens}(x)\rmd x\right) \cdot \left(\int f_{2}(x)\tilde{\targetdens}(x)\rmd x\right) \leq \left(\int f_{3}(x)\tilde{\targetdens}(x)\rmd x\right) \cdot \left(\int f_{4}(x)\tilde{\targetdens}(x)\rmd x\right) \\
    \Updownarrow \\
    \left(\int_{E}f_{1}\exp(-\tilde{\potential{}})\right) \cdot \left(\int_{E}f_{2}\exp(-\tilde{\potential{}})\right) \leq \left(\int_{E}f_{3}\exp(-\tilde{\potential{}})\right) \cdot \left(\int_{E}f_{4}\exp(-\tilde{\potential{}})\right)
\end{gather*}
for any \(a, b \in \bbR^{d}\), \(\gamma \in \bbR\), where \(\int_{E}h \defeq \int_{0}^{1} h(a + t(b - a))\exp(-\gamma t)\rmd t\).
Since \(\tilde{\targetdens}(x) = \exp(-\tilde{\potential{}}(x)) = 0\) for \(x \not\in \interior{\primalspace}\), this equivalence also holds with the substitutions \(\tilde{\targetdens} \leftarrow \targetdens\) and \(\tilde{\potential{}} \leftarrow \potential{}\).
Let \(C_{\sqrt{\mu}} = \frac{\sqrt{\mu}}{8 + 4\sqrt{\mu}}\) and \(\tilde{\metric{}} = \mu \cdot \metric{}\) (which implies that \(d_{\tilde{\metric{}}} = \sqrt{\mu} \cdot d_{\metric{}}\)).
The inequality we would like to show can be expressed in terms of indicator functions as
\begin{equation*}
    \int \targetdens(x) \bm{1}_{S_{3}}(x)\rmd x \cdot \int (C_{\sqrt{\mu}} \cdot d_{\tilde{\metric{}}}(S_{1}, S_{2}))^{-1} \targetdens(x) \rmd x \geq \int \targetdens(x) \bm{1}_{S_{2}}(x)\rmd x \cdot \int \targetdens(x) \bm{1}_{S_{1}}(x) \rmd x~.
\end{equation*}
Consider \(f_{1} = \bm{1}_{\overline{S_{1}}}\), \(f_{2} = \bm{1}_{\overline{S_{2}}}\), \(f_{3} = \bm{1}_{\interior{\primalspace} \setminus \overline{S_{1}} \setminus \overline{S_{2}}}\), and \(f_{4} = (C_{\sqrt{\mu}} \cdot d_{\tilde{\metric{}}}(S_{1}, S_{2}))^{-1}\), where \(\overline{S}\) is closure of a set \(S\).
With this construction, \(f_{1}\), \(f_{2}\) are upper-semicontinuous, and \(f_{3}\), \(f_{4}\) are lower-semicontinuous.
If the top inequality in the equivalence holds with these \(\{f_{i}\}_{i=1}^{4}\), then this implies our required inequality as
\begin{align*}
    \int \targetdens(x) \bm{1}_{S_{3}}\rmd x \cdot \int (C_{\sqrt{\mu}} \cdot d_{\tilde{\metric{}}}(S_{1}, S_{2}))^{-1} \targetdens(x) \rmd x &= \int f_{3}(x) \targetdens(x) \rmd x \cdot \int f_{4}(x) \targetdens(x) \rmd x \\
    &\geq \int f_{1}(x)\targetdens(x)\rmd x \cdot \int f_{2}(x) \targetdens(x) \rmd x \\
    &\geq \int \targetdens(x)\bm{1}_{S_{1}}\rmd x \cdot \int \targetdens(x)\bm{1}_{S_{2}}(x)\rmd x~.
\end{align*}
Due to the equivalence from the localisation lemma, as noted in \citet[Lem. B.7]{kook2024gaussian} it suffices to establish that for all \(\gamma \in \bbR\) and \(a, b \in \bbR^{d}\), and \(a + t \cdot (b - a) \in \interior{\primalspace}\) for \(t \in [0, 1]\)
\begin{multline*}
    C_{\sqrt{\mu}} \cdot d_{\tilde{\metric{}}}(S_{1}, S_{2}) \int_{0}^{1}e^{\gamma t - \potential{a + t \cdot (b - a)}} \bm{1}_{\overline{S_{1}}}(a + t \cdot (b - a))\rmd t \cdot \int_{0}^{1}e^{\gamma t - \potential{a + t \cdot (b - a)}} \bm{1}_{\overline{S_{2}}}(a + t \cdot (b - a))\rmd t \\
    \leq \int_{0}^{1}e^{\gamma t - \potential{a + t \cdot (b - a)}}\rmd t \cdot \int_{0}^{1}e^{\gamma t - \potential{a + t \cdot (b - a)}} \bm{1}_{S_{3}}(a + t \cdot (b - a))\rmd t~.
\end{multline*}
We refer to the above inequality as the ``needle'' inequality.

Now we define the following quantities for arbitrary \(a, b \in \bbR^{d}\) and \(\gamma \in \bbR\).
\begin{align*}
    V(t) &\defeq \potential{a + t \cdot (b - a)} - \gamma t \\
    T_{i} &\defeq \{t \in [0, 1] : a + t \cdot (b - a) \in S_{i}\} \qquad \text{for each } i \in [3] \\
    \frakg(t) &\defeq (b - a)^{\top}\tilde{\metric{}}(a + t \cdot (b - a))(b - a)~.
\end{align*}
Note that \(V\) is \(1\)-relatively convex with respect to \(\frakg\) as
\begin{align*}
    V''(t) &= (b - a)^{\top}\hesspotential{a + t \cdot (b - a)}(b - a) \\
    &\geq \mu \cdot (b - a)^{\top}\metric{a + t \cdot (b - a)}(b - a) \\
    &= (b - a)^{\top}\tilde{\metric{}}{(a + t \cdot (b - a))}(b - a) = \frakg(t)~.
\end{align*}
Additionally note the following property of \(\frakg\).
\begin{align*}
    |\frakg'(t)| &= |\rmD\tilde{\metric{}}(a + t \cdot (b - a))[b - a, b - a, b - a]| \\
    &\leq \mu \cdot |\rmD\metric{a + t \cdot (b - a)}[b - a, b - a, b- a]| \\
    &\leq 2\mu \cdot \|b - a\|_{\metric{a + t \cdot (b - a)}}^{3} \\
    &\leq \frac{2}{\sqrt\mu} \cdot \|b - a\|_{\tilde{\metric{}}(a + t \cdot (b - a))}^{3} = \frac{2}{\sqrt{\mu}} \cdot \frakg(t)^{\nicefrac{3}{2}}~.
\end{align*}

The ``needle'' inequality can be expressed in terms of one-dimensional integrals as shown below.
\begin{equation*}
    C_{\sqrt{\mu}} \cdot d_{\tilde{\metric{}}}(S_{1}, S_{2}) \cdot \int_{t \in T_{1}} e^{-V(t)} \rmd t \cdot  \int_{t \in T_{2}} e^{-V(t)} \rmd t \leq  \int_{t \in [0, 1]} e^{-V(t)} \rmd t \cdot  \int_{t \in T_{3}} e^{-V(t)} \rmd t~.
\end{equation*}
Let \(d_{\frakg}(v_{1}, v_{2}) = \int_{v_{1}}^{v_{2}} \sqrt{\frakg(t)} \rmd t\).
As remarked in the proof of \citet[Lem. 9]{gopi2023algorithmic},
\begin{equation*}
    d_{\frakg}(T_{1}, T_{2}) \geq d_{\tilde{\metric{}}}(S_{1}, S_{2})~.
\end{equation*}

Therefore, if
\begin{equation*}
    C_{\sqrt{\mu}} \cdot d_{\frakg}(T_{1}, T_{2}) \cdot \int_{t \in T_{1}} e^{-V(t)} \rmd t \cdot \int_{t \in T_{2}} e^{-V(t)} \rmd t \leq  \int_{t \in [0, 1]} e^{-V(t)} \rmd t \cdot  \int_{t \in T_{3}} e^{-V(t)} \rmd t~
\end{equation*}
is true, then the ``needle'' inequality is true, which implies the inequality in the statement of the lemma.
The remainder of the proof follows the proof of \citet[Lem. 9]{gopi2023algorithmic}.
Assume that \(T_{3}\) is a single interval.
This implies that \(T_{1} = [c, c'], T_{3} = [c', d'], T_{2} = [d', d]\) for \(0 \leq c < c' < d' < d \leq 1\).
Then,
\begin{equation*}
    d_{\frakg}(T_{1}, T_{2}) = \inf_{u \in T_{1}, v \in T_{2}} d_{\frakg}(u, v) = d_{\frakg}(c', d') = \int_{c'}^{d'} \sqrt{\frakg(t)}\rmd t~.
\end{equation*}
With \Cref{lem:one-dim-isoperimetry} with \(\kappa \leftarrow \sqrt{\mu}\), we show the inequality in this setting as shown below.
\begin{align*}
    \int_{c'}^{d'}\exp(-V(t))\rmd t &\geq \min_{t \in [c', d']} \frac{\exp(-V(t))}{\sqrt{\frakg(t)}} \cdot \int_{c'}^{d'}\sqrt{\frakg(t)} \rmd t\\
    &\geq \min_{t \in [c', d']} C_{\sqrt{\mu}} \cdot \min\left\{\int_{-\infty}^{t}e^{-V(t)}\rmd t, \int_{t}^{\infty}\exp(-V(t))\rmd t \right\} \cdot \int_{c'}^{d'}\sqrt{\frakg(t)} \rmd t \\
    &\geq C_{\sqrt{\mu}} \cdot \min\left\{\int_{c}^{c'}e^{-V(t)}\rmd t, \int_{d'}^{d}e^{-V(t)}\rmd t \right\} \cdot \int_{c'}^{d'}\sqrt{\frakg(t)} \rmd t \\
    &= C_{\sqrt{\mu}} \cdot \min\left\{\frac{\int_{c}^{c'} \exp(-V(t))\rmd t}{Z}, ~\frac{\int_{d'}^{d}\exp(-V(t))\rmd t}{Z}\right\} \cdot Z \cdot \int_{c'}^{d'}\sqrt{\frakg(t)} \rmd t \\
    &\geq C_{\sqrt{\mu}} \cdot \int_{c}^{c'} \exp(-V(t))\rmd t \cdot \int_{d'}^{d}\exp(-V(t))\rmd t \cdot  \int_{c'}^{d'}\sqrt{\frakg(t)} \rmd t \cdot \frac{1}{Z}~.
\end{align*}
where \(Z = \int_{c}^{d}e^{-V(t)}\rmd t\).
The final statement uses \(\min\{a, b\} \geq ab\) for \(0 < a, b \leq 1\).
When \(T_{3}\) is a collection of disjoint intervals, the general trick from \citet{lovasz1993random} applies here as performed in \citet{gopi2023algorithmic}, and this inequality is applied to each interval in \(T_{3}\) and its neighbouring intervals.
Thus, we have shown the inequality, which implies the ``needle'' inequality, and therefore proving the original statement of the lemma.
\end{proof}

\subsubsection{Proof of \Cref{lem:one-dim-isoperimetry}}
\label{sec:proofs:proof-one-dim-isoperimetry}

\begin{proof}
This proof is inspired by \citet[Proof of Lem. 7]{gopi2023algorithmic}.
Assume that \(V'(x) \geq 0\), and define \(r = x + \frac{\kappa}{4\sqrt{\frakg(x)}}\).
By the property of \(\frakg\), we have for any \(t \in [x, r]\).
\begin{equation*}
    \frac{1}{\sqrt{\frakg(t)}} - \frac{1}{\sqrt{\frakg(x)}} =-\int_{x}^{t} \frac{\frakg'(s)}{2\frakg(s)^{\nicefrac{3}{2}}} \leq \frac{t - x}{\kappa} \leq \frac{r - x}{\kappa} = \frac{1}{4\sqrt{\frakg(x)}} \Rightarrow \frakg(t) \geq \frac{1}{2}\frakg(x)~.
\end{equation*}

Since \(V\) is \(1\)-relatively convex w.r.t. \(\frakg\), for all \(t \in [x, r]\),
\begin{equation*}
    V''(t) \geq \frakg(t) \geq \frac{1}{2}\frakg(x)~.
\end{equation*}
and with the assumption that \(V'(x) \geq 0\), we get
\begin{gather*}
    V(t) \geq V(x) + V'(x)(t - x) + \int_{x}^{t}(t - s)V''(s)\rmd s \geq V(x) + \frac{(t - x)^{2}}{4}\frakg(x)~,\\
    V'(r) = V'(x) + \int_{x}^{r}V''(s)\rmd s \geq\frac{r - x}{2} \frakg(x) = \frac{\kappa}{8}\sqrt{\frakg(x)}~.
\end{gather*}
We also have for \(t > r\) that
\begin{equation*}
    V(t) \geq V(r) + V'(r)(t - r) \geq V(x) + V'(x)(r - x) + V'(r)(t - r) \geq V(x) + \frac{\kappa(t - r)}{8}\sqrt{\frakg(x)}~.
\end{equation*}
With these, we get
\begin{align*}
    \int_{x}^{\infty}\exp(-V(t)) \rmd t &= \int_{x}^{r} \exp(-V(t)) \rmd t + \int_{r}^{\infty} \exp(-V(t)) \rmd t \\
    &\leq \exp(-V(x)) \int_{x}^{r} \exp\left(-\frac{(t - x)^{2}}{4}\frakg(x)\right) \rmd t \\
    &\qquad + \int_{r}^{\infty}\exp(-V(x))\exp\left(-\frac{\kappa(t - r)\sqrt{\frakg(x)}}{8}\right) \rmd t \\
    &\leq \exp(-V(x)) \cdot \frac{4 + \frac{8}{\kappa}}{\sqrt{\frakg(x)}}~.
\end{align*}

For the setting \(V'(x) \leq 0\), define \(r = x - \frac{\alpha}{4\sqrt{\frakg(x)}}\) and for any \(t \in [r, x]\),
\begin{equation*}
    \frac{1}{\sqrt{\frakg(x)}} - \frac{1}{\sqrt{\frakg(t)}} = -\int_{t}^{x} \frac{\frakg'(s)}{2\frakg(s)^{\nicefrac{3}{2}}} \geq -\frac{(x - t)}{\kappa} \geq \frac{r - x}{\kappa} = -\frac{1}{4\sqrt{\frakg(x)}} \Rightarrow \frakg(t) \geq \frac{1}{2}\frakg(x)~.
\end{equation*}
With the assumption that \(V'(x) \leq 0\), we get
\begin{align*}
    V(t) \geq V(x) + V'(x)(t - x) + \int_{x}^{t}(t - s)V''(s)\rmd s \geq V(x) &\geq \int_{t}^{x}(t - s)V''(s)\rmd s \\
    &\geq \frac{(t - x)^{2}}{4}\frakg(x)~,\\
    V'(r) = V'(x) + \int_{x}^{r}V''(s)\rmd s \leq -\int_{r}^{x}V''(s)\rmd s \leq \frac{r - x}{2}\frakg(x) &= -\frac{\kappa}{8}\sqrt{\frakg(x)}~.
\end{align*}
We also have for \(t < r\) that
\begin{align*}
    V(t) &\geq V(r) + V'(r)(t - r) \\
    &\geq V(x) + V'(x)(r - x) + V'(r)(t - r) \\
    &\geq V(x) + \frac{\kappa(r- t)}{8}\sqrt{\frakg(x)}~.
\end{align*}
Analogous to the case \(V'(x) \geq 0\), consider the integrals \(\int_{-\infty}^{r}\exp(-V(t))\rmd t\) and \(\int_{r}^{x}\exp(-V(t))\rmd t\).
This also results in
\begin{equation*}
    \int_{-\infty}^{x} \exp(-V(t)) \rmd t \leq \exp(-V(x)) \cdot \frac{4 + \frac{8}{\kappa}}{\sqrt{\frakg(x)}}~.
\end{equation*}
\end{proof}

\subsection{Other technical lemmas and their proofs}

\begin{proof}[Proof of \Cref{lem:nesterov-todd-result}]
Consider any geodesic \(\xi : [0, 1] \to \interior{\primalspace}\) such that \(\xi(0) = x\) and \(\xi(1) = y\).
Let \(\bar{t}\) be the time when \(\xi(\bar{t})\) hits the boundary of \(\calE_{x}^{\metric{}}(r)\), and note that \(\bar{t} \leq 1\) as \(\xi(1) = y\).
Then, for \(\delta(t) := \|\xi(t) - x\|_{\metric{x}}\),
\begin{equation*}
    \frac{\rmd}{\rmd t}\delta(t)^{2} = 2\delta(t)\delta'(t) = 2\langle\xi'(t), \xi(t) - x\rangle_{\metric{x}} \leq 2\|\xi'(t)\|_{\metric{x}}\|\xi(t) - x\|_{\metric{x}}
\end{equation*}
and this implies that \(\delta'(t) \leq \|\xi'(t)\|_{\metric{x}}\).
Therefore,
\begin{equation*}
    d_{\metric{}}(x, y) = \int_{0}^{1}\|\xi'(t)\|_{\metric{\xi(t)}} \rmd t \geq \int_{0}^{\bar{t}}\|\xi'(t)\|_{\metric{\xi(t)}} \rmd t \geq \int_{0}^{\bar{t}} \|\xi'(t)\|_{\metric{x}} \cdot (1 - \|\xi(t) - x\|_{\metric{x}}) \rmd t~,
\end{equation*}
where the last inequality is due to the self-concordance of \(\metric{}\) (\cref{lem:dikin1-in-domain-self-concordant-loewner}(2)) and that \(\|\xi(t) - x\|_{\metric{x}} < 1\) for \(t \in (0, \bar{t})\).
The final integral in the above chain is at least
\begin{equation*}
    \int_{0}^{\bar{t}} \delta'(t)(1 - \delta(t))\rmd t = \delta(\bar{t}) - \frac{1}{2}\delta(\bar{t})^{2} = r - \frac{1}{2}r^{2}~.
\end{equation*}
The second statement of the lemma directly follows from the remainder of proof of \citet[Lem. 3.1]{nesterov2002riemannian}.
\end{proof}

\begin{proof}[Proof of \cref{lem:SC-lemma}]
From \Cref{lem:dikin1-in-domain-self-concordant-loewner}(2) and the assumption that \(\|y - x\|_{\metric{x}} < 1\), we have
\begin{equation*}
    \metric{y} \succeq (1 - \|y - x\|_{\metric{x}})^{2} \cdot \metric{x} \Leftrightarrow \metric{x}^{-\nicefrac{1}{2}}\metric{y}\metric{x}^{-\nicefrac{1}{2}} \succeq (1 - \|y - x\|_{\metric{x}})^{2} \cdot \rmI_{d \times d}~.
\end{equation*}
Hence, for \(M = \metric{x}^{-\nicefrac{1}{2}}\metric{y}\metric{x}^{-\nicefrac{1}{2}}\), we have
\begin{align*}
    \log \det \metric{y}\metric{x}^{-1} = \log \det M &= \sum_{i=1}^{d} \log \lambda_{i}(M) \\
    &\geq 2d \cdot \log(1 - \|y - x\|_{\metric{x}}) \geq -3d \cdot \|y - x\|_{\metric{x}}~,
\end{align*}
where the last step uses \Cref{auxlem:log-is-linear}.
Also, from \Cref{lem:dikin1-in-domain-self-concordant-loewner}(2), we have
\begin{equation*}
    \metric{y} \preceq \frac{1}{(1 - \|y - x\|_{\metric{x}})^{2}} \cdot \metric{x} \Rightarrow \|x - y\|_{\metric{y}}^{2} \leq \frac{\|x - y\|_{\metric{x}}^{2}}{(1 - \|y - x\|_{\metric{x}})^{2}}~.
\end{equation*}

Using \Cref{auxlem:almost-cubic}, this implies that
\begin{equation*}
    \|x - y\|_{\metric{x}}^{2} - \|x - y\|_{\metric{y}}^{2} \geq \|x - y\|_{\metric{x}}^{2} - \frac{1}{(1 - \|x - y\|_{\metric{x}})^{2}} \cdot \|x - y\|_{\metric{x}}^{2} \geq -6 \cdot \|x - y\|_{\metric{x}}^{3}~.
\end{equation*}
\end{proof}

\begin{proof}[Proof of \cref{lem:SSC-LTSC-lemma}]
Let \(s(t) = x + t \cdot (y - x)\), and \(\varphi(t) = \log \det \metric{s(t)}\).
By Taylor's theorem, there exists \(t^{\star} \in (0, 1)\) such that
\begin{equation*}
    \log \det \metric{y}\metric{x}^{-1} = \varphi(1) - \varphi(0) = \varphi'(0) + \varphi''(t^{\star})~.
\end{equation*}
We have the following expressions for \(\varphi'(0)\) and \(\varphi''(t)\).
\begin{align*}
    \varphi'(0) &= \langle \nabla \log \det \metric{x}, y - x\rangle \\
    \varphi''(t) &= \trace(\metric{s(t)}^{-1}\rmD^{2}\metric{s(t)}[y - x, y - x])\\
    &\qquad \quad - \|\metric{s(t)}^{-\nicefrac{1}{2}}\rmD\metric{s(t)}[y - x]\metric{s(t)}^{-\nicefrac{1}{2}}\|_{\frob}^{2}~.
\end{align*}
As strong self-concordance implies self-concordance, \Cref{lem:dikin1-in-domain-self-concordant-loewner}(2) implies that for any \(t \in [0, 1]\)
\begin{equation*}
    \metric{s(t)} \preceq \frac{\metric{x}}{(1 - \|s(t) - x\|_{\metric{x}})^{2}} \preceq \frac{\metric{x}}{(1 - \|y - x\|_{\metric{x}})^{2}}~,
\end{equation*}
where the last step uses the fact that \(\|s(t) - x\|_{\metric{x}} = t \cdot \|y - x\|_{\metric{x}}\).
Finally, we use the definitions of strong self-concordance and \(\alpha\)-lower trace self-concordance to get
\begin{align*}
    \varphi''(t^{\star}) &= \trace(\metric{s(t^{\star})}^{-1}\rmD^{2}\metric{s(t^{\star})}[y - x, y - x]) \\
    &\qquad - \|\metric{s(t^{\star})}^{-\nicefrac{1}{2}}\rmD\metric{s(t^{\star})}[y - x]\frakG(s(t^{\star}))^{-\nicefrac{1}{2}}\|_{\frob}^{2} \\
    &\geq -\alpha \cdot \|y - x\|_{\metric{s(t^{\star})}}^{2} - 4 \cdot \|y - x\|_{\metric{s(t^{\star})}}^{2} \\
    &\geq -(\alpha + 4) \cdot \frac{\|y - x\|_{\metric{x}}^{2}}{(1 - \|y - x\|_{\metric{x}})^{2}}~.
\end{align*}
\end{proof}

\begin{proof}[Proof of \cref{lem:sq-dist-diff-ASC-SC}]
For \(t, s \in [0, 1]\), define the following functions.
\begin{gather*}
    y(t) = y + t \cdot (w - y), \quad p(t) = \|y_{t} - x\|_{\metric{y_{t}}}^{2} - \|y_{t} - x\|_{\metric{x}}^{2} \\
    u(s; t) = x + s \cdot (y(t) - x), \quad q(s; t) = 2 \langle y(t) - x, (w  - y)\rangle_{\metric{u(s; t)}}~.
\end{gather*}
By Taylor's theorem, there exists \(t^{\star} \in (0, 1)\) such that
\begin{align*}
    \Delta(w; x) - \Delta(y; x) &= p(1) - p(0) \\
    &= p'(t^{\star}) \\
    &= 2\langle y(t^{\star}) - x, w - y\rangle_{\metric{y(t^{\star})}} - 2\langle y(t^{\star}) - x, w - y\rangle_{\metric{x}} \\
    &\qquad + \rmD\metric{y(t^{\star})}[y(t^{\star}) - x, y(t^{\star}) - x, w - y] \\
    &= q(1; t^{\star}) - q(0; t^\star) + \rmD\metric{y(t^{\star})}[y(t^{\star}) - x, y(t^{\star}) - x, w - y]~.
\end{align*}
Since \(q(.; t^{\star})\) is differentiable, we also have by Taylor's theorem that there exists \(\bar{s} \in (0, 1)\).
\begin{align*}
    q(1; t^{\star}) - q(0; t^{\star}) &= q'(\bar{s}; t^{\star}) \\
    &= 2\rmD\metric{u(\bar{s}; t^{\star})}[y(t^{\star}) - x, y(t^{\star}) - x, w - y]~.
\end{align*}
The self-concordance of \(\metric{}\) enables the bounds
\begin{align*}
    \Delta(w; x) - \Delta(y; x) &= 2\rmD\metric{u(\bar{s}; t^{\star})}[y(t^{\star}) - x, y(t^{\star}) - x, w - y] \\
    &\qquad + \rmD\metric{y(t^{\star})}[y(t^{\star}) - x, y(t^{\star}) - x, w - y] \\
    &\leq 4\|y(t^{\star}) - x\|_{\metric{u(\bar{s}; t^{\star})}}^{2} \cdot \|w - y\|_{\metric{u(\bar{s}; t^{\star})}} \\
    &\qquad + 2\|y(t^{\star}) - x\|^{2}_{\metric{y(t^{\star})}} \cdot \|w - y\|_{\metric{y(t^{\star})}}~.
\end{align*}

Note that
\begin{equation*}
    \|u(\bar{s}; t^{\star}) - x\|_{\metric{x}} = \bar{s} \cdot \|y(t^{\star}) - x\|_{\metric{x}}~.
\end{equation*}
Since \(y, w \in \calE_{x}^{\metric{}}(1)\) which is convex subset of \(\primalspace\), \(y(t^{\star}) \in \calE_{x}^{\metric{}}(1)\), and this implies that \(u(\bar{s}; t^{\star}) \in \calE_{x}^{\metric{}}(1)\) as well.
Therefore, by \cref{lem:dikin1-in-domain-self-concordant-loewner}(2),
\begin{equation*}
    \Delta(w; x) - \Delta(y;x) \leq 6 \cdot \frac{\|y(t^{\star}) - x\|_{\metric{x}}^{2} \cdot \|w - y\|_{\metric{x}}}{(1 - \|y(t^{\star}) - x\|_{\metric{x}})^{3}}~.
\end{equation*}
\end{proof}

\section{Conclusion}
\label{sec:conclusion}
In summary, we propose a new first-order algorithm for the constrained sampling problem called \nameref{alg:mapla}, which is algorithmically motivated by the natural gradient descent algorithm in optimisation.
This method performs a Metropolis adjustment of the Markov chain resulting from an approximate version of the preconditioned Langevin algorithm (\ref{eq:PLA}), and supersedes the Metropolis-adjusted Mirror Langevin algorithm proposed by \citet{srinivasan2024fast} by working with a general metric \(\metric{}\).
We derive non-asymptotic mixing time guarantees for our method under a variety of assumptions made on the target distribution \(\targetdist\) and the metric \(\metric{}\).
We find that when \(\metric{}\) satisfies certain stronger notions of self-concordance, the dimension dependence in the mixing time guarantee is strictly better than that obtain with standard self-concordance.
Our numerical experiments showcase how including first-order information about the \(\potential{}\) through \(\gradpotential{}\) using the natural gradient can be beneficial in comparison to \Dikin{} which only uses \(\potential{}\), and could motivate the design of more sophisticated first-order methods for constrained sampling.

Several open questions remain.
We exclude the correction term \((\nabla \cdot \metric{}^{-1})\) in \ref{eq:PLA}, which is the key difference compared to \textsf{ManifoldMALA}.
Notwithstanding the computational difficulty, it would be interesting to see the what the effect of including this correction term would be on the mixing time.
More theoretically, drawing from the discussion of the results, it would be interesting to identify other scenarios where the weaker sufficient condition pertaining to \(\|\gradpotential{\cdot}\|_{\metric{\cdot}^{-1}}\) holds, and its implications for the mixing time of \nameref{alg:mapla}.
Another course to eliminating the gradient upper bound is showing that the above local norm quantity concentrates when \(\potential{}\) and \(\metric{}\) satisfy certain properties such as the \((\mu, \metric{})\) and \((\lambda, \metric{})\)-curvature lower and upper bounds, as done in more recent analyses \citep{lee2020logsmooth} in the case where \(\metric{} = \rmI_{d \times d}\) i.e., \textsf{MALA}.
Algorithmically, it would be also be interesting to find other candidate proposal Markov chains that can yield similar or better mixing time guarantees.
While \ref{eq:PLA} serves as a useful proposal Markov chain, its efficacy as a standalone algorithm (with a projection to ensure feasibility) is not investigated in this work. 
As noted earlier, \ref{eq:PLA} is likely to be biased, but whether this bias is \emph{vanishing} (i.e., when the bias \(\to 0\) as \(h \to 0\)) under certain conditions on the metric \(\metric{}\) would be interesting to check.

\section*{Acknowledgments}
The authors would like to thank the reviewers at ALT 2025 for their feedback, and Nawaf Bou-Rabee for helpful remarks.
The authors acknowledge the MIT SuperCloud and Lincoln Laboratory Supercomputing Center for providing high-performance computing resources that have contributed to the experimental results reported within this work.
Vishwak Srinivasan was supported by a Simons Foundation Collaboration on Theory of Algorithmic Fairness Grant.
Andre Wibisono was supported by NSF Award CCF \#2403391.

\bibliography{references.bib}

\appendix
\section{Addendum}
\label{sec:addendum}
\subsection{Concentration inequalities}

\begin{lemma}
\label{auxlem:chisquared}
Let \(\gamma \sim \calN(0, I_{d})\).
For \(\varepsilon \in (0, 1)\), the event \(\|\gamma\|^{2} \leq d \cdot \sfN_{\varepsilon}\) occurs with probability \(1 - \varepsilon\).
\end{lemma}
\begin{proof}
From \citet[Lem. 1]{laurent2000adaptive}, we have for any \(t > 0\),
\begin{equation*}
    \bbP\left(\|\gamma\|^{2} > d \cdot \left\{1 + 2\sqrt{\frac{t}{d}} + 2\frac{t}{d}\right\}\right) \leq e^{-t}
\end{equation*}
As \(d \geq 1\), \(\frac{t}{d} \leq t\), and hence
\begin{equation*}
    \bbP\left(\|\gamma\|^{2} > d \cdot \{1 + 2\sqrt{t} + 2t\}\right) \leq e^{-t}~.
\end{equation*}
Substituting \(t = \log\left(\frac{1}{\varepsilon}\right)\) completes the proof.
\end{proof}
    
\begin{lemma}
\label{auxlem:gaussian-inner}
Let \(\gamma \sim \calN(0, I_{d})\), and \(v \in \bbR^{d}\) be a vector such that \(\|v\| \leq \frakB\).
For \(\varepsilon \in (0, 1)\), the events \(\langle v, \gamma\rangle \leq \frakB \cdot \sfI_{\varepsilon}\) and \(\langle v, \gamma\rangle \geq -\frakB \cdot \sfI_{\varepsilon}\) each occur with probability at least \(1 - \varepsilon\).
\end{lemma}
\begin{proof}
Note that \(\langle v, \gamma\rangle\) is a Gaussian random variable with mean \(0\) and variance \(\|v\|^{2}\).
Thus, for any \(t > 0\),
\begin{equation*}
    \bbP(\langle v, \gamma\rangle < -t) = \bbP(\langle v, \gamma\rangle > t) \leq e^{-\frac{t^{2}}{2\|v\|^{2}}} \leq e^{-\frac{t^{2}}{2\frakB^{2}}}~.
\end{equation*}
Substituting \(t = \frakB \cdot \sfI_{\varepsilon}\) recovers the second statement.
\end{proof}

\subsection{Facts about self-concordant metrics}
\label{sec:app:conseq-SC}

\begin{lemma}
\label{lem:dikin1-in-domain-self-concordant-loewner}
Consider a self-concordant metric \(\metric{} : \interior{\primalspace} \to \bbS_{+}^{d}\).
This satisfies the following properties
\begin{enumerate}[label=(\arabic*)]
\item For every \(x \in \interior{\primalspace}\),
\[\calE_{x}^{\metric{}}(1) \subseteq \interior{\primalspace}~.\]
\item For any pair \(x, y \in \interior{\primalspace}\) such that \(\|x - y\|_{\metric{x}} < 1\), then
\begin{equation*}
    (1 - \|x - y\|_{\metric{x}})^{2} \cdot \metric{x} \preceq \metric{y} \preceq \frac{1}{(1 - \|x - y\|_{\metric{x}})^{2}} \cdot \metric{x}~.
\end{equation*}
\item If \(\primalspace\) does not contain a straight line, then \(\metric{x}\) is non-degenerate for all \(x \in \interior{\primalspace}\).
\end{enumerate}
\end{lemma}
\begin{proof}
Define the matrix function \(\metric{}_{\epsilon} : \interior{\primalspace} \to \bbS_{+}^{d}\), where \(\metric{}_{\epsilon}(x) = \metric{x} + \epsilon \cdot \rmI_{d \times d}\).
For any \(x \in \interior{\primalspace}\) and \(v \in \bbR^{d}\), 
\begin{equation*}
    |\rmD\metric{}_{\epsilon}(x)[v, v, v]| = |\rmD\metric{x}[v, v, v]| \leq 2 \cdot \|v\|^{3}_{\metric{x}} \leq 2 \cdot \|v\|^{3}_{\metric{}_{\epsilon}(x)}~
\end{equation*}
and this shows that \(\metric{}_{\epsilon}\) is self-concordant.
For a given \(x \in \interior{\primalspace}\), let \(v \in \bbR^{d}\) be such that \(\|v\|_{\metric{x}} > 0\).
Then, it holds that \(\|v\|_{\metric{}_{\epsilon}(x)} > 0\) as well.
With such \(x\) and \(v\), consider the univariate function \(\phi(t) = \langle v, \metric{}_{\epsilon}(x + tv)v\rangle^{-\nicefrac{1}{2}}\). The derivative \(\phi'(t)\) satisfies
\begin{equation*}
    \phi'(t) = -\frac{\rmD\metric{}_{\epsilon}(x + tv)[v, v, v]}{2\|v\|_{\metric{}_{\epsilon}}^{3}} \Rightarrow |\phi'(t)| \leq 1~.
\end{equation*}
Note that \(t \in (-\phi(0), \phi(0))\) belongs in the domain of \(\phi\).
This is due to Taylor's theorem, the observation about \(\phi'\), which states for any \(t\) that \(\phi(t) - \phi(0) \geq -|t|\), and the fact that \(\phi(t) > 0\).
Therefore, any point of the form \(x + tv\) for \(t^{2} \leq \phi(0)^{2}\) belongs in \(\interior{\primalspace}\).
In other words,
\begin{equation*}
    \left\{x + tv : t^{2}\|v\|^{2}_{\metric{x}} + \epsilon \cdot t^{2}\|v\|^{2} \leq 1\right\} \subseteq \interior{\primalspace}.
\end{equation*}
As \(\epsilon\) can be arbitrarily close to \(0\), setting \(\epsilon \to 0\) proves the first part of the lemma.

For the second statement of the lemma, consider \(v = y - x\) from the statement, and note that \(\phi(0) > 1\) by the assumption that \(\|y - x\|_{\metric{x}} < 1\).
As a result,
\begin{equation*}
\phi(1) - \phi(0) \geq -1 \Rightarrow \frac{1}{\|y - x\|_{\metric{y}}} \geq \frac{1}{\|y - x\|_{\metric{x}}} - 1~.
\end{equation*}
This is equivalent to
\begin{equation*}
    \|y - x\|_{\metric{y}} \leq \frac{\|y - x\|_{\metric{x}}}{1 - \|y - x\|_{\metric{x}}}~.
\end{equation*}
Let \(x_{t} = x + t \cdot (y - x)\), and define \(\psi(t) = v^{\top}\metric{x_{t}}v\) for some arbitrary \(v \in \bbR^{d}\).
Note that \(x_{t} - x = t \cdot (y - x)\).
By self-concordance of \(\metric{}\),
\begin{equation*}
    \psi'(t) \leq \rmD\metric{x_{t}}[v, v, y - x] \leq 2 \cdot \|v\|_{\metric{x_{t}}}^{2} \cdot \|y - x\|_{\metric{x_{t}}}\leq \frac{2\psi(t)}{t} \cdot \frac{\|y - x\|_{\metric{x}}}{1 - t \cdot \|y- x\|_{\metric{x}}}~.
\end{equation*}
The remainder of the proof follows from the proof of \citet[Thm. 5.1.7]{nesterov2018lectures}.

For the third assertion of the lemma, if for any \(x \in \interior{\primalspace}\), \(\metric{}(x)\) is degenerate, then there exists \(v \in \bbR^{d}\) that is non-zero such that \(\metric{}(x)v = \bm{0}\).
Consider \(\bar{x} = x + r \cdot v\) for \(r \in \bbR\).
Note that \(\|\bar{x} - x\|_{\metric{}(x)} = r \cdot \|v\|_{\metric{}(x)} = 0\), and hence the line \(\{x + r \cdot v : r \in \bbR\}\) belongs in \(\calE_{x}^{\metric{}}(1)\).
Since \(\calE_{x}^{\metric{}}(1) \subseteq \interior{\primalspace}\) by \Cref{lem:dikin1-in-domain-self-concordant-loewner}(1), this is a contradiction of the assumption that \(\primalspace\) does not contain any straight lines, and therefore the assumption that there exists \(x \in \interior{\primalspace}\) such that \(\metric{}(x)\) is degenerate is false.
\end{proof}

\subsection{Miscellaneous algebraic lemmas}

\begin{lemma}
\label{auxlem:log-inequality}
Let \(f(x) = x - 1 - \log(x)\).
Then, for any \(a \in (0, 1)\),
\begin{equation*}
    \max_{x \in [a, a^{-1}]} f(x) \leq \frac{(a - 1)^{2}}{a}~.
\end{equation*}
\end{lemma}
\begin{proof}
For we have for every \(x > 0\) (see \citet[Lem. 22]{srinivasan2024fast} for a proof)
\begin{equation*}
    f(x) \leq \frac{(x - 1)^{2}}{x} \quad \Rightarrow 
    \max_{x \in [a, a^{-1}]} f(x) \leq \max_{x \in [a, a^{-1}]} \frac{(x - 1)^{2}}{x}~.
\end{equation*}
The function \(\frac{(x - 1)^{2}}{x}\) is convex, and hence the maximum over \([a, a^{-1}]\) is attained at the end points.
\begin{equation*}
    \max_{x \in [a, a^{-1}]} \frac{(x - 1)^{2}}{x} = \max\left\{\frac{(a - 1)^{2}}{a}, \frac{(a^{-1} - 1)^{2}}{a^{-1}}\right\} = \frac{(a - 1)^{2}}{a}~.
\end{equation*}
\end{proof}

\begin{lemma}
\label{auxlem:small-quad}
Let \(f(x) = |(1 - x)^{2} - 1|\). If \(x \in [0, 4]\),
\begin{equation*}
    f(x) \leq 2x
\end{equation*}
\end{lemma}
\begin{proof}
\begin{equation*}
    \{(1 - x)^{2} - 1\}^{2} = \{2x - x^{2}\}^{2} = 4x^{2} + x^{4} - 4x^{3} \leq 4x^{2}
\end{equation*}
since \(x^{4} \leq 4x^{3}\) for \(x \in [0, 4]\).
\end{proof}

\begin{lemma}
\label{auxlem:log-is-linear}
For any \(t \in [0, 0.3]\),
\begin{equation*}
    \log(1 - t) \geq -\frac{3}{2}t~.
\end{equation*}
\end{lemma}
\begin{proof}
Let \(f(t) = e^{-\frac{3}{2}t} + t - 1\).
The derivative of \(f(t)\) is \(-\frac{3}{2}\exp(-\frac{3}{2}t) + 1\).
For any \(t \in [0, 0.3]\), \(f'(t) < 0\), and consequently, \begin{equation*}
    f(t) \leq f(0) = 0 \Rightarrow e^{-\frac{3}{2}t} + t - 1 \leq 0 \Leftrightarrow \log(1 - t) \geq -\frac{3}{2}t~.
\end{equation*}
\end{proof}

\begin{lemma}
\label{auxlem:almost-cubic}
For any \(t \in [0, 0.3]\),
\begin{equation*}
    t^{2} - \frac{t^{2}}{(1 - t)^{2}} \geq -6t^{3}~.
\end{equation*}
\end{lemma}
\begin{proof}
Through algebraic simplifications,
\begin{equation*}
    t^{2} - \frac{t^{2}}{(1 - t)^{2}} = \frac{t^{4}}{(1 - t)^{2}} - \frac{2t^{3}}{(1 - t)^{2}} \geq -\frac{2t^{3}}{(1 - t)^{2}} \geq -2t^{3} \cdot \frac{17}{8} \geq -6t^{3}~.
\end{equation*}
\end{proof}

\begin{lemma}
\label{auxlem:b1-implication}
Let \(a, b \geq 0\) and \(d \geq 1\).
If \(x \leq V \cdot \min\left\{\frac{1}{a^{2}}, \frac{1}{a^{\nicefrac{2}{3}}}, \frac{1}{(a \cdot b)^{\nicefrac{2}{3}}}, \frac{1}{d^{3}}, \frac{1}{b \cdot d}\right\}\) for \(V \geq 0\), then each of the terms
\begin{equation*}
\left\{
\begin{tabular}{ccc}
\(a^{2} \cdot x\) & \(a^{2}b \cdot x^{2}\) & \(bd \cdot x\) \\
\(ab \cdot x^{\nicefrac{3}{2}}\) & \(ad \cdot x\) & \(d^{3} \cdot x\) \\
\(ad^{2} \cdot x^{\nicefrac{3}{2}}\) & \(a^{3} \cdot x^{2}\) & \(a^{\nicefrac{3}{2}} \cdot x^{\nicefrac{5}{4}}\)
\end{tabular}\right.
\end{equation*}
are bounded from above as a function of \(V\) alone.
\end{lemma}
\begin{proof}
We make use the following observation: for \(t \geq 0\) and \(p > 0\), \(\min\{t, t^{-p}\} \leq 1\).
\begin{itemize}[leftmargin=*]
\item \(a^{2} \cdot x \leq a^{2} \cdot V \cdot \frac{1}{a^{2}} = V\).
\item \(a^{2}b \cdot x^{2} \leq V^{2} \cdot \min\left\{\frac{a^{2}b}{a^{4}}, \frac{a^{2}b}{a^{\nicefrac{4}{3}} \cdot b^{\nicefrac{4}{3}}}\right\} = V^{2} \cdot \min\left\{\frac{b}{a^{2}}, \frac{a^{\nicefrac{2}{3}}}{b^{\nicefrac{1}{3}}}\right\} \leq V^{2}\).
\item \(bd \cdot x \leq bd \cdot V \cdot \frac{1}{bd} = V\).
\item \(ab \cdot x^{\nicefrac{3}{2}} \leq ab \cdot V^{\nicefrac{3}{2}} \cdot \frac{1}{ab} = V^{\nicefrac{3}{2}}\).
\item \(ad \cdot x \leq V \cdot \min\left\{\frac{ad}{a^{2}}, \frac{ad}{d^{3}}\right\} \leq V \cdot \min\left\{\frac{d}{a}, \frac{a}{d^{2}}\right\} \leq V \cdot \min\left\{\frac{d}{a}, \frac{a}{d}\right\} \leq V\).
\item \(d^{3} \cdot x \leq d^{3} \cdot V \cdot \frac{1}{d^{3}} = V\).
\item \(ad^{2} \cdot x^{\nicefrac{3}{2}} \leq V^{\nicefrac{3}{2}} \cdot \min\left\{\frac{ad^{2}}{a^{3}}, \frac{ad^{2}}{d^{\nicefrac{9}{2}}}\right\} \leq V^{\nicefrac{3}{2}} \cdot \min\left\{\frac{d^{2}}{a^{2}}, \frac{a}{d^{\nicefrac{7}{2}}}\right\} \leq V^{\nicefrac{3}{2}} \cdot \min\left\{\frac{d^{2}}{a^{2}}, \frac{a}{d}\right\} \leq V^{\nicefrac{3}{2}}\).
\item \(a^{3} \cdot x^{2} \leq V^{2} \cdot \min\left\{\frac{a^{3}}{a^{4}}, \frac{a^{3}}{a^{\nicefrac{4}{3}}}\right\} = V^{2} \cdot \min\left\{\frac{1}{a}, a^{\nicefrac{5}{3}}\right\} \leq V^{2}\).
\item \(a^{\nicefrac{3}{2}} \cdot x^{\nicefrac{5}{4}} \leq V^{\nicefrac{5}{4}} \cdot \min\left\{\frac{a^{\nicefrac{3}{2}}}{a^{\nicefrac{5}{2}}}, \frac{a^{\nicefrac{3}{2}}}{a^{\nicefrac{5}{6}}}\right\} = V^{\nicefrac{5}{4}} \cdot \min\left\{\frac{1}{a}, a^{\nicefrac{2}{3}}\right\} \leq V^{\nicefrac{5}{4}}\).
\end{itemize}
\end{proof}

\begin{lemma}
\label{auxlem:b2-implication}
Let \(a, b \geq 0\), \(c \geq 4\), and \(d \geq 1\).
If \(x \leq V \cdot \min\left\{\frac{1}{a^{2}}, \frac{1}{(a \cdot b)^{\nicefrac{2}{3}}}, \frac{1}{(a \cdot c)^{\nicefrac{2}{3}}}, \frac{1}{d \cdot a}, \frac{1}{d \cdot b}, \frac{1}{d \cdot c}\right\}\), then each of the terms
\begin{equation*}
\left\{
\begin{tabular}{cccc}
\(a^{2} \cdot x\) & \(a^{2}b \cdot x\) & \(bd \cdot x\) & \(ab \cdot x^{\nicefrac{3}{2}}\) \\
\(a\sqrt{d} \cdot x\) & \(d \cdot x\) & \(a^{2}c \cdot x^{2}\) & \(cd \cdot x\) \\
\(ac \cdot x^{\nicefrac{3}{2}}\) & \(a^{3} \cdot x^{2}\) & \(ad \cdot x\) & \(a^{2} \cdot x^{\nicefrac{3}{2}}\)
\end{tabular}
\right.
\end{equation*}
are bounded from above as a function of \(V\) alone.
\end{lemma}

\begin{proof}
We again make use the following observation: for \(t \geq 0\) and \(p > 0\), \(\min\{t, t^{-p}\} \leq 1\).
\begin{itemize}[leftmargin=*]
\item \(a^{2} \cdot x \leq a^{2} \cdot V \cdot \frac{1}{a^{2}} = V\).
\item \(a^{2}b \cdot x^{2} \leq V^{2} \cdot \min\left\{\frac{a^{2}b}{a^{4}}, \frac{a^{2}b}{a^{\nicefrac{4}{3}} \cdot b^{\nicefrac{4}{3}}}\right\} = V^{2} \cdot \min\left\{\frac{b}{a^{2}}, \frac{a^{\nicefrac{2}{3}}}{b^{\nicefrac{1}{3}}}\right\} \leq V^{2}\).
\item \(bd \cdot x \leq bd \cdot V \cdot \frac{1}{bd} = V\).
\item \(ab \cdot x^{\nicefrac{3}{2}} \leq ab \cdot V^{\nicefrac{3}{2}} \cdot \frac{1}{ab} = V^{\nicefrac{3}{2}}\).
\item \(a\sqrt{d} \cdot x \leq V \cdot \frac{a\sqrt{d}}{a \cdot d} \leq V\).
\item \(d \cdot x \leq V \cdot \frac{d}{d \cdot c} \leq \frac{V}{4}\).
\item \(a^{2}c \cdot x^{2} \leq V^{2}\) (analogous to \(a^{2}b \cdot x^{2} \leq V^{2}\)).
\item \(cd \cdot x \leq V \cdot \frac{cd}{cd} = V\).
\item \(ac \cdot x^{\nicefrac{3}{2}} \leq V^{\nicefrac{3}{2}}\) (analogous to \(ab \cdot x^{\nicefrac{3}{2}} \leq V^{\nicefrac{3}{2}}\)).
\item \(a^{3} \cdot x^{2} \leq V^{2} \cdot \min\left\{\frac{a^{3}}{a^{4}}, \frac{a^{3}}{d^{2}a^{2}}\right\} \leq V^{2} \cdot \min\left\{\frac{1}{a}, \frac{a}{d^{2}}\right\} \leq V^{2} \cdot \min\left\{\frac{1}{a}, a\right\} \leq V^{2}\).
\item \(ad \cdot x \leq V \cdot \frac{ad}{ad} = V\).
\item \(a^{2} \cdot x^{\nicefrac{3}{2}} \leq a^{2}c \cdot x^{\nicefrac{3}{2}} \leq V^{\nicefrac{3}{2}}\).
\end{itemize}
\end{proof}

\end{document}